\documentclass[reqno,onefignum,onetabnum,final]{siamart171218}

\usepackage{url}
\usepackage{enumitem}
\usepackage[mathscr]{euscript}
\usepackage{algorithm,algorithmic}
\usepackage{amsmath,amssymb,amsfonts}
\usepackage{graphicx}
\usepackage{xcolor}
\usepackage{booktabs}
\usepackage{bm}
\usepackage{tikz}
\usepackage{pgfplots}
\usepackage{capt-of,graphicx}%

\newcommand{\fac}{\frac{1}{\sigma_i}}

\newcommand{\N}{{\mathfrak{N}}}
\newcommand{\Cp}{\mathcal{C}_\text{par}}
\newcommand{\Cq}{\mathcal{C}_\text{qoi}}
\newcommand{\Nq}{{N_\text{qoi}}}

\newcommand{\Np}{{N_\text{par}}}

\newcommand{\Ns}{{N_\text{MC}}}

\newcommand{\R}{\mathbb{R}}
\newcommand{\B}{\mathcal{B}}
\newcommand{\D}{\mathcal{D}}
\newcommand{\X}{\mathscr{X}}

\newcommand{\Vs}{\mathscr{V}_{\!{g}}}
\newcommand{\Vo}{\mathscr{V}_{\!\scriptscriptstyle{{0}}}}

\newcommand{\var}[1]{\mathrm{Var}\left\{ {#1} \right\}}

\renewcommand{\S}{\mathfrak{S}}
\newcommand{\ST}{\mathfrak{S}^{\text{tot}}}
\newcommand{\trace}{\mathrm{Tr}}
\newcommand{\E}[1]{\mathrm{E}\left\{{#1}\right\}}
\newcommand{\EE}[2]{\mathrm{E}_{{#1}}\left\{{#2}\right\}}
\newcommand{\delj}[1]{\frac{\partial f_{{#1}}(\theta)}{\partial \theta_j}}

\newcommand{\ee}{\mathcal{E}}
\newcommand{\err}[2]{\varepsilon\left( {#1}; s; {#2} \right)}

\newcommand{\ut}[1]{\tilde{{#1}}}
\newcommand{\GD}{{\Gamma_{D}}}
\newcommand{\GN}{{\Gamma_{N}}}
\newcommand{\ip}[2]{\langle{#1}, {#2}\rangle}
\renewcommand{\L}{\mathcal{L}}
\newcommand{\Q}{\mathcal{Q}}
\newcommand\restr[2]{{%
  \left.\kern-\nulldelimiterspace %
  #1 %
  \vphantom{\big|} %
  \right|_{#2} %
  }}
\newcommand{\ipg}[2]{\ip{{#1}}{{#2}}_{\GN}}

\newtheorem{assumption}{Assumption}[section]

\newcommand{\nominal}{\eta}
\newcommand{\Thetau}{{\Theta_U}}
\newcommand{\Thetaz}{{\Theta_{U^c}}}
\newcommand{\freduced}{f^{(\nominal)}}
\renewcommand\qedsymbol{\proofbox}
\makeatletter
\newenvironment{named_proof}[1][\proofname]{\par
  \normalfont \topsep6\p@\@plus6\p@\relax
  \trivlist
  \item[\hskip\labelsep
        \itshape
    #1\@addpunct{.}]\ignorespaces
}{%
\nolinebreak\hfill\qedsymbol\endtrivlist\@endpefalse
}
\makeatother

\graphicspath{ {images/} }

\def\addressncsu{Department of Mathematics, 
North Carolina State University, Raleigh, NC, USA }
\def\addressumbc{Department of Mechanical Engineering, 
University of Maryland, Baltimore County, Baltimore, MD, USA }

\title{Derivative-based global sensitivity analysis for models with
high-dimensional inputs and functional outputs}

\headers{Derivative-based GSA for functional outputs}{H.~Cleaves, A.~Alexanderian, 
H.~Guy, R.~Smith, and M.L.~Yu}

\author{Helen L.~Cleaves\thanks{\addressncsu (\email{hlcleave@ncsu.edu,
alexanderian@ncsu.edu,
hguy@ncsu.edu,
rsmith@ncsu.edu})} \and
Alen Alexanderian\footnotemark[1]\and
Hayley Guy\footnotemark[1] \and
Ralph C.~Smith\footnotemark[1] \and
Meilin Yu\thanks{\addressumbc (\email{mlyu@umbc.edu}})
}

\date{\today}

\begin{document}

\maketitle

\begin{abstract}
We present a framework for derivative-based global sensitivity analysis (GSA)
for models with high-dimensional input parameters and functional outputs. We
combine ideas from derivative-based GSA, random field representation via
Karhunen--Lo\`{e}ve expansions, and adjoint-based gradient computation to
provide a scalable computational framework for computing the proposed
derivative-based GSA measures.  We illustrate the strategy for a nonlinear ODE
model of cholera epidemics and for elliptic PDEs with application examples from
geosciences and biotransport.  
\end{abstract}

\begin{keywords}
Global sensitivity analysis, DGSMs, functional Sobol' indices, 
Karhunen--Lo\`{e}ve expansions.
\end{keywords}

\begin{AMS}
65C20,   %
65C50,   %
62H99,   %
65D15.   %
\end{AMS}

\section{Introduction}\label{sec:intro}
\setlength{\abovedisplayskip}{3pt}
\setlength{\belowdisplayskip}{3pt}
\allowdisplaybreaks

The field of global sensitivity analysis (GSA) provides methods for quantifying
how the uncertainty in the output of mathematical models can be apportioned to
uncertainties in the input model parameters~\cite{Saltelli00}.  Specifically,
variance-based GSA enables ranking the importance of model parameters by
computing their relative contribution to the variance of the output quantities of
interest (QoIs), as quantified by Sobol' indices~\cite{Sobol:1990, Saltelli00,
Sobol:2001}. 
Another popular GSA approach involves using derivative-based 
global sensitivity measures (DGSMs)~\cite{SobolKucherenko09,KucherenkoIooss17}, 
which have
been shown to provide efficient means of screening for unimportant input
parameters. 
In this article,
we consider mathematical models of the form
\begin{equation}\label{equ:QoI}
    y = f(s, \theta),
\end{equation}
where $s$ belongs to a compact set $\X \subset \R^d$ with $d = 1, 2$, or $3$,
and $\theta$ is an element of an uncertain parameter space $\Theta \subseteq\R^\Np$.
We present a mathematical framework for derivative-based 
GSA for \emph{functional} QoIs of the form~\cref{equ:QoI} and present a
scalable computational framework for computing the corresponding 
derivative-based GSA measures.

\textbf{Survey of literature and existing approaches}. 
A great amount of progress has been made in theory and numerical methods for
variance-based GSA over the past three decades~\cite{Sobol:1990,
Saltelli00,
Sobol:2001,
Sobol07,
Sudret:2008,
Crestaux:2009,
SobolKucherenko09,
GamboaJanonKleinEtAl14,
KucherenkoIooss17,
PrieurTarantola17,
MarrelSaintGeoursDeLozzo17,
GratietMarelliSudret17}.
The majority of works on GSA focus on scalar-valued QoIs. However,
in recent years there have been a number of efforts targeting GSA for
vectorial or functional QoIs.  Specifically, the
works~\cite{CampbellMcKayWilliams06,LamboniMonodMakowski11,
GamboaJanonKleinEtAl14,XiaoLi16,AlexanderianGremaudSmith17} discuss variance-based 
GSA for vectorial and functional outputs. Computing GSA measures for
functional QoIs, as is the case for their scalar counterparts, is computationally
challenging. The computational challenges can be reduced significantly by 
employing
surrogate models~\cite{Sudret:2008,Crestaux:2009,Alexanderian13,Sargsyan2017,
HartAlexanderianGremaud17,AlexanderianGremaudSmith17}.
However, surrogate model construction itself becomes
computationally challenging for models with high-dimensional input parameters. 

DGSMs have been shown to provide efficient means for detecting unimportant
input
parameters~\cite{Kucherenko0Rodriguez-FernandezPantelidesEtAl09,
KucherenkoIooss17,VohraAlexanderianSaftaEtAl18}.
For a scalar QoI $g(\theta)$ that has square integrable partial derivatives, 
the DGSMs, defined as $\mathrm{E}\big\{(\frac{\partial g}{\partial
\theta_j})^2\big\}$, $j = 1, \ldots, \Np$, 
are commonly used. 
(Here $\mathrm{E}$ denotes expectation with respect to
$\theta$.)  
These DGSMs
can be used to bound the total Sobol' indices, for models
with statistically independent inputs, which justifies their use in screening
for unimportant inputs. 

An alternate approach for approximating the DGSMs, for scalar-valued QoIs,
using the active subspace method~\cite{constantine2014,Constantine15} is
presented in~\cite{ConstantineDiaz17}.  Namely,~\cite{ConstantineDiaz17}
presents a method for
approximating the DGSMs using dominant eigenpairs of the
matrix $\E{ \nabla f(\theta) \nabla f(\theta)^\top}$.  
While active subspace methods
have mostly targeted scalar QoIs, recently there have been 
initial efforts in generalizing these methods to vectorial
outputs; see e.g.,~\cite{JiWangZahmEtAl18,ZahmConstantinePrieurEtAl18}.

\textbf{Our approach and contributions}.
We focus on functional QoIs of the form $f:\X \times \Theta \to \R$,
as defined in~\cref{equ:QoI}, where $\X$ and $\Theta$ are as before. 
We focus on models with independent random input parameters.
Moreover, 
in our target applications,
$f(s, \theta)$ is defined in terms of the
solution of a system of differential equations. 

We begin our developments by defining a suitable DGSM for functional QoIs, 
in~\cref{sec:functional_DGSMs}, and
prove that it provides a computable bound for the 
generalized total Sobol' indices for functional QoIs as defined 
in~\cite{GamboaJanonKleinEtAl14,AlexanderianGremaudSmith17}; see \cref{thm:bound_main}. Next, we present a framework for efficient
computation of the functional DGSMs that uses low-rank representation
of the functional QoIs via truncated Karhunen--Lo\`{e}ve (KL)
expansions~\cite{LeMaitreKnio10}. Expressions
for DGSMs, and DGSM-based bounds on functional total Sobol' indices 
for a truncated KL expansion are established in~\cref{thm:finite_rank}. 
The DGSMs of the approximate
models, given by truncated KL expansions, are then computed using adjoint-based 
gradient computation. This approach is elaborated for models governed
by linear elliptic PDEs in \cref{sec:DGSM_PDE}. 

Additionally, we present a comprehensive set of numerical results that
illustrate various aspects of the proposed approach and demonstrate its
effectiveness.  We consider three application problems: (i) a nonlinear system
of ODEs modeling the spread of cholera~\cite{HartleyEtAl2005}, where we perform GSA
for the infected population as a function of time (\cref{sec:cholera}); (ii)
a problem motivated by porous medium flow applications, 
with permeability data
adapted from~\cite{SPE},
where we assess
parametric sensitivities of the pressure field on a domain boundary
(\cref{sec:poi_spe}); and (iii) an application problem 
involving biotransport
in tumors~\cite{AlexanderianZhuSalloumEtAl17}, where we consider the pressure distribution in certain subdomains
of a tumor model (\cref{sec:biotransport}).

\textbf{Article overview}.
This article is structured as follows.  In~\cref{sec:prelim}, we set up
the notation used throughout the article, and collect the assumptions on the
functional QoIs under study. We also provide a brief review of variance-based
GSA for functional QoIs, following the developments
in~\cite{GamboaJanonKleinEtAl14,AlexanderianGremaudSmith17}, in~\cref{sec:prelim}.  
In \cref{sec:functional_DGSMs} we present a mathematical framework
for derivative-based GSA of functional QoIs.  We elaborate our proposed
adjoint-based framework for models governed by linear elliptic PDEs in \cref{sec:DGSM_PDE}.  
This is followed by our computational experiments
that are detailed in \cref{sec:numerics}. Finally, we provide some
concluding remarks in \cref{sec:conc}.

\section{Preliminaries}\label{sec:prelim}

\subsection{The basic setup}
Let $\Theta \subseteq \R^\Np$  be the uncertain
parameter space, and consider the probability space 
$(\Theta, \B, \mu)$, where $\B$ is the Borel $\sigma$-algebra on $\Theta$ and
$\mu$ is the law of the uncertain parameter vector $\theta$. 
In the present work, $\Theta$ is of the form 
$\Theta = \Theta_1 \times \Theta_2 \times \cdots 
\times \Theta_\Np$, where $\Theta_j \subseteq \R$, 
$j = 1, \ldots, \Np$.
The expectation of a random variable $g:\Theta \to \R$ is  
denoted by
\[
   \E{g} = \int_\Theta g(\theta) \mu(d\theta).
\]
We assume the components of the random vector $\theta$
are independent and
admit probability density functions $\pi_j(\theta_j)$, in which case
$\mu(d\theta) = \prod_{j = 1}^\Np \pi_j(\theta_j) d\theta_j$.
Next, let
$\X \subset \R^d$, with $d = 1, 2$, or $3$ be a compact set. With this setup, 
we consider a process,
$f:\X \times \Theta \to \R$ as in~\cref{equ:QoI}. 
Note that this
setup covers both time-dependent and spatially distributed processes. In the
former case, $\X$ is a time interval, and in the latter
case, $\X$ is a spatial region. 

\textbf{Assumptions on the process.}
We consider random processes that satisfy the following
assumptions.
\begin{assumption}\label{assump:main}
We assume 
\begin{enumerate}[label=(\alph*)]
\item $f\in L^2(\X \times \Theta)$ and $f$ is mean square continuous; that is, 
for any sequence $\{s_n\}$ in $\X$ converging to $s\in \X $ we have that
$\lim_{n\rightarrow \infty} \E{[f(s_n, \theta)-f(s, \theta)]^2} = 0$.
\item $\frac{\partial f}{\partial \theta_j}(s, \theta) $is defined for all $s\in \X$ and $\theta\in \Theta$, $j = 1, \ldots, \Np$;
\item 
$\frac{\partial f}{\partial \theta_j}(s, \theta)\in L^2(\X \times \Theta)$, 
$j = 1, \ldots, \Np$;
\item and $\{\theta_j\}_{j = 1}^\Np$, are real-valued independent 
random variables, and have distribution laws
that are absolutely continuous with respect to the Lebesgue measure. 
\end{enumerate}
\end{assumption}

We remark that (a) is a fundamental assumption on the process $f$.  From this,
we can conclude continuity of the mean and covariance function of the process;
see, e.g.,~\cite[Theorem 7.3.2]{HsingEubank15},~\cite[Theorem
2.2.1]{Adler10}.  This in turn facilitates application of Mercer's
Theorem~\cite{Mercer1909,Lax02} (needed below) and implies that $f$ admits a KL
expansion~\cite{Loeve77}. 
The assumptions (b) and (c) are needed in the context of 
derivative-based global sensitivity analysis.
Note that \cref{assump:main}(b) can be relaxed by requiring 
$\frac{\partial f}{\partial \theta_j}(s, \theta)$
be defined almost everywhere in $\X \times \Theta$.

\subsection{Variance-based sensitivity analysis for functional 
outputs}
We first recall the classical Sobol' 
indices and 
Analysis of Variance (ANOVA)
decomposition~\cite{Sobol07,Sobol:2001,Sobol:1990}, 
which can be
defined pointwise in $\X$. Let $K = \{1 ,2, \ldots, 
\Np\} $ be an index set, let $U = \{j_1, j_2, \ldots, j_m\}$ be a
subset of $K$, and let $U^c$ be the complement of $U$ in $K$, $U^c = K \setminus U$. 
We denote $\theta_U = \{\theta_{j_1}, \theta_{j_2}, \ldots, \theta_{j_m}\}$.
For each $s \in \X$, we have the ANOVA decomposition~\cite{Sobol07} 
\begin{equation}\label{equ:ANOVA}
        f(s, \theta) = f_0(s) + f_1(s, \theta_U) + f_2(s, \theta_{U^c}) + f_{12}(s, \theta), 
\end{equation}
where $f_0$ is the mean of the process, and
\[
f_1(s, \theta_U) = \E{f(s, \cdot) | \theta_U} -f_0(s), 
\quad
f_2(s,\theta_{U^c}) = \E{f(s, \cdot) | \theta_{U^c}} - f_0(s), 
\]
and $f_{12}(s, \theta) = f(s, \theta) - f_0(s) - f_1(s, \theta_U) - 
f_2(s,\theta_{U^c})$.
This enables decomposing
the total variance $D(f; s) = \var{f(s, \cdot)}$ of $f(s,\cdot)$ according
to 
\[
        D(f; s) = D_{U}(f; s) + D_{U^c}(f; s) + D_{U,U^c}(f; s),
\]
where 
$D_U(f; s) = \EE{\theta_U}{f_1(s, \theta_U)^2}$,  
$D_{U^c}(f; s) = \EE{\theta_{U^c}}{f_2(s, \theta_{U^c})^2}$, and
$D_{U,U^c}(f; s)$ is the remainder. (Here $\EE{\theta_U}{\cdot}$ 
indicates expectation with respect to $\theta_U$.) 
Then, we can define the first and total order Sobol' indices as follows:
\[
S_U(f; s) = \frac{D_U(f; s)}{D(f; s)}
\quad \text{and} \quad 
\quad S_{U}^\text{tot}(f; s) = \frac{D_U^{tot}(f; s)}{D(f; s)}, 
\]
where $D_U^{tot}(f; s) = D_U(f; s) + D_{U,U^c}(f; s)$.
Note that,
\[
S_{U}^\text{tot}(f; s) = \frac{D(f; s) - D_{U^c}(f; s)}{D(f; s)}
= 1 -\frac{D_{U^c}(f; s)}{D(f; s)} = 1 - S_{U^c}(f; s).
\]
When the index set $U$ is a singleton, $U = \{ j \}$, $j \in \{1, \ldots, \Np\}$,
we denote the corresponding first and total order Sobol' indices by $S_j(f; s)$ and
$S_j^\text{tot}(f; s)$, respectively. 

Here we assume that $D(f; s) > 0$ almost everywhere in $\X$.  If $D(f; s) = 0$
for some $s \in \X$, we use the convention $S_{U}(f; s) = 0$.

\subsection{Functional Sobol' indices}
Following~\cite{AlexanderianGremaudSmith17}, 
we define the functional first order Sobol' index as
\[
    \S_U(f; \X) = \frac{ \int_\X D_U(f; s) \, ds}{\int_\X D(f; s) \, ds}.
\]
The following lemma provides a simple representation for the functional 
Sobol' index in terms of the pointwise classical Sobol' indices:
\begin{lemma}\label{lem:Sobol_weighted_avg}
We have
$\S_U(f; \X) = \int_\X S_U(f; s) w(s) \, ds$,
with   
$w(s) = \dfrac{D(f; s)}{\int_\X D(f; y) \, dy}$. 
\end{lemma}
\begin{proof}
The result follows by a straightforward calculation.
\end{proof}
We can also define the functional total Sobol' indices
\[
\ST_U(f; \X) = \frac{ \int_\X D_U^\text{tot}(f; s) \, ds}{\int_\X D(f; s) \, ds}
= 1 - \S_{U^c}(f; \X).
\]
Using \cref{lem:Sobol_weighted_avg}, we note 
\begin{equation}\label{equ:total_Sobol_weighted_avg}
\ST_U(f; \X) = 1 - \S_{U^c}(f; \X) =  
\int_\X (1 - S_{U^c}(f; s)) w(s) \, ds =  
\int_\X S_{U}^\text{tot}(f; s) w(s) ds.
\end{equation}

\textbf{Error estimates.}
We can use the total Sobol' index of a parameter to rank its importance.  In
particular, parameters with small Sobol' indices can be deemed unimportant. In
this section, we briefly 
discuss the impact of fixing
these unimportant parameters in terms of approximation errors.
Let $U = \{j_1, j_2, \ldots, j_m\} \subset \{1, \ldots,\Np\}$ index the 
set of important parameters, and suppose we set
$\theta_{U^c}$ to a nominal vector $\nominal$. Consider the ``reduced'' model:
\[
   \freduced(s, \theta_U) = f(s, \theta_U, \nominal),
\] 
where the right hand side function is understood to be $f(s, \theta)$, with
entries of $\theta_{U^c}$ fixed at $\nominal$.

For $U = \{j_1, j_2, \ldots, j_m\}$ we define 
$\Thetau = \Theta_{j_1} \times \cdots \times \Theta_{j_m}$.
Integration on $\Thetau$ will be with respect to 
$\mu(d\theta_U) = \prod_{k=1}^{m} \pi_{j_k}(\theta_{j_k})d\theta_{j_k}$. 

For $\nominal \in \Thetaz$ we define the mean-square error 
\[
        \err{f}{\nominal} = \int_\Theta 
                   (f(s, \theta)- \freduced(s, \theta_U))^2\, \mu(d\theta).
\]

We consider the relative mean square error
\begin{equation}\label{equ:globalL2}
     \ee(f; \nominal) = \frac{\displaystyle \int_\X \int_\Theta 
                     (f(s, \theta)- \freduced(s, \theta_U))^2\, \mu(d\theta) ds}
       {\displaystyle\int_\X \int_\Theta f(s, \theta)^2\, \mu(d\theta) ds}.
\end{equation}
This provides a measure of the error that occurs when fixing the values 
of $\theta_{U^c}$. 
The following proposition quantifies this error in terms of the
functional total Sobol' indices. This result is a straightforward
modification of the error estimate presented in~\cite{AlexanderianGremaudSmith17};
we provide a proof in \cref{sec:proof_L2} for completeness. 
\begin{proposition}\label{prop:L2}
$\int_\Thetaz \ee(f; \nominal) \mu(d\nominal) \leq 2\ST_{\theta_{U^c}}(f, \X)$.
\end{proposition}
\begin{proof}  
See \cref{sec:proof_L2}.
\end{proof}
The estimate in Proposition~\ref{prop:L2} says that when 
fixing $\theta_{U^c}$ to a nominal parameter $\nominal \in \Thetaz$, in average,
the relative error $\ee(f; \nominal)$ is bounded by $2\ST_{\theta_{U^c}}(f, \X)$.  

\section{Derivative-based GSA for 
functional QoIs}\label{sec:functional_DGSMs}
Let us first consider a scalar-valued random variable $g:\Theta \to \R$.
Here $g$ and its partial derivatives are assumed to be square integrable.
We recall the following commonly used DGSM~\cite{SobolKucherenko09}:
\[
    \nu_j(g) = \int_\Theta \Big(\frac{\partial g}{\partial \theta_j}\Big)^2 \mu(d\theta).
\]
DGSMs can be used to screen for unimportant variables.
This is justified by the relation between DGSMs and total Sobol'
indices, which was first addressed in~\cite{SobolKucherenko09} 
for scalar-valued random variables. 
While the estimation of $\nu_j$
requires a Monte Carlo (MC) sampling procedure, it has been observed that in
practice the number of samples required for estimation of $\nu_j$'s does not
need to be very large to provide sufficient accuracy in identifying
unimportant variables.  
We present the following result which partially
explains this phenomenon.

\begin{proposition}\label{prp:BD}
Assume that 
\[
   a_j \leq  \Big(\frac{\partial g}{\partial \theta_j}(\theta)\Big)^2 \leq b_j, \quad 
   j = 1, \ldots, \Np, \quad \text{for all } \theta \in \Theta.
\]
Consider the MC estimator 
\[
   \nu_j^{(\Ns)}(g) := \frac1\Ns \sum_{k = 1}^\Ns \Big(\frac{\partial g}{\partial\theta_j}
                                            (\theta^k) \Big)^2,
\]
with $\theta^k$ independent and identically distributed according to the law of $\theta$.
Then,
\begin{equation}\label{equ:ineq}
    \var{\nu_j^{(\Ns)}(g)} \leq \frac1\Ns \big(b_j - \nu_j(g)\big)\big(\nu_j(g) - a_j\big) 
                         \leq \frac{1}{4\Ns} (b_j - a_j)^2, 
\end{equation}
for $j = 1, \ldots, \Np$. 
\end{proposition}
\begin{proof}
See \cref{sec:proof_BD}.
\end{proof}
This proposition says that if the partial derivatives
do not vary too much (i.e., $a_j$ and $b_j$ are not
too far from one another), indicating a desirable
regularity property of the parameter-to-QoI mapping, 
then the MC estimator $\nu_j^{(\Ns)}(g)$ will have a
small variance for a modest choice of $\Ns$. In such 
cases the MC sample size for 
estimating $\nu_j(g)$ does not need to be very large. 

\textbf{Functional DGSMs.}
Next, we turn to DGSMs for functional QoIs. 
We propose the following definition 
for a functional 
DGSM
\begin{equation}\label{equ:DGSM_functional}
   \N_j(f; \X) = \int_\X \int_\Theta \Big(\frac{\partial f}{\partial \theta_j}(s, \theta)\Big)^2 
   \mu(d\theta) ds = 
   \int_\X \nu_j(f(s, \cdot))\, ds,
\end{equation}
which is a natural choice. These indices can be normalized in different ways
to make their comparison easier. For instance, we may consider the normalized indices
\[
     \frac{\N_j(f; \X)}{\sum_{k = 1}^\Np \N_k(f; \X)}, \quad j = 1, \ldots, \Np.
\]

We can relate $\N_j(f; \X)$ to the corresponding 
functional total Sobol' indices
$\ST_j(f; \X)$, $j = 1, \ldots, \Np$, analogously to the scalar case. 
Specifically, we present the following result
that shows functional total Sobol' indices can be bounded in terms of the
proposed functional DGSMs.
\begin{theorem}\label{thm:bound_main}
Let $f(s, \theta)$ be a random process satisfying
\cref{assump:main}. Suppose $\theta_j$ 
are independent and distributed 
according to uniform or normal distribution, for
$i = 1, \ldots, \Np$. Then,
\begin{equation}\label{equ:bound}
\ST_j(f; \X) \leq \alpha_j \frac{\N_j(f; \X)}{\trace(\Cq)}, 
\quad j = 1, \ldots, \Np,
\end{equation}
where $\Cq$ is
the covariance operator of the random function $f(s, \theta)$,
and
\[
\alpha_j = \begin{cases} 
(b-a)^2/\pi^2, \quad \text{if } \theta_j \sim U(a, b),\\
\sigma_j^2,   \quad \text{if } \theta_j \sim \mathcal{N}(0, \sigma_j^2).
\end{cases} 
\]
\end{theorem}
\begin{proof}
For a fixed $s \in \X$, by the results in~\cite{KucherenkoIooss17},
\begin{equation}\label{equ:poincare}
    S_j^\text{tot}(f; s) \leq \frac{\alpha_j}{D(f; s)} \nu_j(f; s).
\end{equation}
Then, using~\cref{equ:total_Sobol_weighted_avg},
\begin{equation}\label{equ:bound_almost_there}
\begin{aligned}
  \ST_j(f; \X) &= \int_\X S_j^\text{tot}(f; s) w(s)\, ds
\\
               &\leq \int_\X \frac{\alpha_j}{D(f; s)} \nu_j(f; s) w(s) \, ds =\alpha_j\frac{\N_j(f; \X)}{\int_\X D(f; s)\, ds}.
\end{aligned}
\end{equation}
Now, let $\Cq$ be the covariance operator of the random process 
$f(s, \theta)$, 
and let $c(s, t)$ be its covariance function.
As a consequence of Mercer's Theorem~\cite{Mercer1909,Lax02}, we have 
\[
    \int_\X D(f; s) \, ds =  \int_\X \var{f(s, \cdot)} \, ds = 
   \int_\X c(s, s) \, ds
   = \trace(\Cq).
\]
Combining this with~\cref{equ:bound_almost_there} we obtain the desired result. 
\end{proof}

The DGSM-based upper bounds on the functional total Sobol' indices 
provided by \cref{thm:bound_main}
enable identifying inputs with small total Sobol' indices, hence providing
an efficient way of identifying unimportant parameters. 
Note that the theorem is stated for $\theta_j$ that are distributed
uniformly or normally, because these distributions are commonly used in
modeling under uncertainty.  However, the result holds for other families of
distributions.  Specifically, in~\cite{LamboniIoossPopelinEtAl13}, it is shown
that~\cref{equ:poincare} holds for   the Boltzmann family of distributions with
appropriate choices of the constants $\alpha_j$, $j = 1, \ldots, \Np$, which
provides immediate extension of \cref{thm:bound_main} to  Boltzmann
family of distributions.
We mention that an important class of Boltzmann distributions is the family 
of log-concave distributions that includes Normal, Exponential, Beta, Gamma,
Gumbel, and Weibull distributions~\cite{LamboniIoossPopelinEtAl13}.

Similar to the case of scalar QoIs, estimating functional DGSM often
requires fewer samples than are required for direct calculation of the
Sobol' indices via MC Sampling.
The following result, which is similar to \cref{prp:BD}, provides a bound on the variance of the 
corresponding MC estimator, given appropriate boundedness assumptions on 
the partial derivatives of the functional QoI.
\begin{proposition}\label{prp:BD_functional}
Assume that there exist non-negative integrable functions $a_j$ and $b_j$, defined on 
$\X$ such that for each $s \in \X$,
\[
   a_j(s) \leq  \Big(\frac{\partial f(s, \theta)}{\partial \theta_j}\Big)^2 \leq b_j(s), \quad
   j = 1, \ldots, \Np, \quad \text{for all } \theta \in \Theta.
\]
Consider the MC estimator
\[
   \N_j^{(\Ns)}(f; \X) := \frac1\Ns \sum_{k = 1}^\Ns \int_\X 
\Big(\frac{\partial f}{\partial\theta_j}(s, \theta^k)\Big)^2\, ds,
\]
with $\theta^k$ independent and identically distributed according to the law of $\theta$.
Then,
\begin{equation*}
    \var{\N_j^{(\Ns)}(f; \X)} 
     \leq \frac{1}{4\Ns} \|b_j - a_j\|_{L^1(\X)}^2, \quad j = 1, \ldots, \Np.
\end{equation*}
\end{proposition}
\begin{proof}
See \cref{sec:proof_BD}.
\end{proof}

The indices $\N_j$ can be computed by sampling the partial 
derivatives. 
Gradient computation can be performed using
various techniques. The simplest approach is to use the finite difference method.
However, this approach becomes prohibitive for computationally 
intensive models with a
large number of input parameters. For models governed by differential
equations, one can use the
so called \emph{sensitivity equations} 
for computing derivatives. We demonstrate this
in one of our numerical examples in \cref{sec:numerics}. Unfortunately, 
this approach also suffers from the curse of dimensionality, and becomes cumbersome
for complex systems.  Another approach, not explored in the
present work, is that of automatic differentiation.  The challenges of gradient
computation are compounded for models governed by expensive-to-solve PDEs
with high-dimensional input parameters. For such models, we propose an approach 
that combines low-rank KL expansions and adjoint-based gradient computation.

With the strategy of using low-rank KL expansions for the purposes
of computing DGSMs in mind, we examine 
functional QoIs of the form
\begin{equation}\label{equ:finite_rank}
    f(s, \theta) = \sum_{i = 1}^\Nq \gamma_i f_i({\theta}) \phi_i(s),
\end{equation}
where $\phi_i$ are orthonormal with respect to $L^2(\X)$ inner product,
$\{\gamma_i\}$ are non-negative and sorted in descending order, 
$\E{f_i} = 0$, $i = 1, \ldots, \Nq$, and $\E{f_i f_j} = \delta_{ij}$. 
Suppose also that $f_i$ have square integrable partial derivatives. 

\begin{theorem}\label{thm:finite_rank}
Let $f$ be a random process of the form~\cref{equ:finite_rank}. 
The following hold:
\begin{enumerate}
\item
$\N_j(f; \X) = \sum_{i=1}^\Nq \gamma_j^2 \nu_j(f_i)$, $j = 1, \ldots, \Np$.
\item We have the bound 
\[
   \ST_j(f; \X) \leq \frac{\N_j(f; \X)}{\sum_{i=1}^\Nq \gamma_i^2}
   = \frac{\sum_{i = 1}^\Nq \gamma_i^2 \nu_j(f_i)}
          {\sum_{i=1}^\Nq \gamma_i^2}, 
   \quad j = 1, \ldots, \Np.
\]
\end{enumerate}
\end{theorem}
\begin{proof} 
First, we note 
\begin{multline*}
   \nu_j(f(s, \cdot)) = 
   \int_\Theta \Big(\frac{\partial}{\partial \theta_j}\sum_{i = 1}^\Nq 
   \gamma_i f_i(\theta) \phi_i(s)\Big)^2 \, \mu(d\theta) 
   =\int_\Theta \Big(\sum_{i = 1}^\Nq
   \gamma_i \delj{i} \phi_i(s)\Big)^2 \, \mu(d\theta)
   \\ =
   \sum_{i,k = 1}^\Nq \gamma_i\gamma_k 
   \Big(\int_\Theta\delj{i} \delj{k} \,\mu(d\theta)
   \Big)\phi_i(s)\phi_k(s). 
\end{multline*}
Therefore, 
\begin{align*}
	\N_j(f; \X)   &= \int_\X\nu_j(f(s,\cdot))\,ds\\
	          &=	\int_\X \left(\sum_{i,k}^\Nq \gamma_i\gamma_k\,
			\int_\Theta\delj{i}\delj{k}\,
			\mu(d\theta)\phi_i(s)\phi_k(s)\right)\,ds\\
		  &=    \sum_{i,k = 1}^\Nq \gamma_i\gamma_k 
  		        \Big(\int_\Theta\delj{i} \delj{k} \,\mu(d\theta)
   		        \Big)\int_\X \phi_i(s)\phi_k(s) \,ds \\
	          &=    \sum_{i = 1}^\Nq \gamma_i^2
  		    \Big[\int_\Theta\Big(\delj{i}\Big)^2\,\mu(d\theta)\Big] 
		  = \sum_{i=1}^\Nq \gamma_j^2 \nu_j(f_i). 		
\end{align*}
This establishes the first assertion of the theorem.
Next,
letting $\Cq$ be the covariance operator of $f(s, \theta)$,  
it is straightforward to see that $\trace(\Cq) = \sum_{i = 1}^\Nq \gamma_i^2$.
Thus, combining the first assertion of the 
theorem with \cref{thm:bound_main}, we have
\[
   \ST_j(f; \X) \leq \frac{\N_j(f; \X)}{\sum_{i=1}^\Nq \gamma_i^2}
   = \frac{\sum_{i = 1}^\Nq \gamma_i^2 \nu_j(f_i)}
          {\sum_{i=1}^\Nq \gamma_i^2}, 
   \quad j = 1, \ldots, \Np. 
\]
\end{proof}

\textbf{Computing DGSMs for functional outputs.}
To enable efficient computation of functional DGSMs, we use a truncated
KL expansion of $f$. Let $(\lambda_i(\Cq), \phi_i)$ be
the eigenpairs of the covariance operator of $f$; we consider the truncated KL
expansion 
\begin{equation}\label{equ:KLE_f}
    f(s, \theta) \approx \hat{f}(s, \theta):= \bar{f}(s) + \sum_{i = 1}^\Nq
            \sigma_i f_i(\theta) \phi_i(s),
            \quad \text{with } \sigma_i = \sqrt{\lambda_i(\Cq)},
\end{equation}
where $\bar{f}$ is the mean of the process and the KL modes $f_i$ are given by
\begin{equation}\label{equ:KLmodes}
    f_i(\theta) = \fac
\int_\X (f(s, \theta) - \bar{f}(s)) \phi_i(s) \, ds,
\quad i = 1, \ldots, \Nq.
\end{equation}

In many applications of interest, where the process $f$ is defined in terms
of the solution of a differential equation, the eigenvalues
$\lambda_i(\Cq)$ decay rapidly, and thus a small $\Nq$ can be
afforded. Such processes, which we refer to as low-rank, are 
common in physical and biological applications.
Computing the KL expansion numerically
can be accomplished e.g., using Nystr\"{o}m's method, which is the
approach taken in the numerical experiments in the present work.  We refer
to~\cite{AlexanderianReeseSmithEtAl18}, for a convenient reference for
numerical computation of KL expansions using Nystr\"{o}m's method.  
We point out that this process requires approximating the covariance function 
of $f$, through sampling, when solving the eigenvalue problem for $\{\lambda_i(\Cq)\}_{i \geq 1}$
and the corresponding eigenvectors $\{\phi_i\}_{i \geq 1}$. 
This computation requires an ensemble of model
evaluations $\{f(\cdot, \theta^k)\}_{k = 1}^\Ns$. Typically a modest sample
size $\Ns$ is sufficient for computing the dominant eigenpairs of $\Cq$. This 
is demonstrated in our numerical results in \cref{sec:numerics}.

The approximate model $\hat{f}$ can then be used as a surrogate for $f$ for the
purposes of sensitivity analysis.  Specifically we compute the functional DGSMs
of $\hat{f}$ as a proxy for those of $f$. The computation of functional DGSMs
for $\hat{f}$ and the DGSM-based bound on functional Sobol' indices is
facilitated by \cref{thm:finite_rank}.

The expression for the functional DGSM
given in \cref{thm:finite_rank} requires computing 
DGSMs for the KL modes $f_i$, $i = 1, \ldots, \Nq$, which are
scalar-valued random variables. 
Differentiability of $f_i$ can be established by requiring 
certain boundedness assumptions on the partial 
derivatives. We consider a generic KL mode, which we denote by 
\begin{equation}\label{equ:F}
F(\theta) := \int_\X (f(s, \theta) - \bar{f}(s)) v(s) \, ds, 
\end{equation}
where we use a generic $v \in L^2(\X)$ in the place 
of the eigenvectors.
\begin{proposition}\label{prp:KL_mode_differentiability}
Let $f$ be a process satisfying \cref{assump:main}, and
moreover assume partial derivatives of $f$ with respect to 
$\theta_j$, $j = 1, \ldots, \Np$ satisfy 
\begin{equation}\label{equ:bdd_deriv}
    \left|\frac{\partial f}{\partial \theta_j}(s, \theta)\right| \leq z_j(s),
\quad \text{ for all } (s, \theta) \in \X \times \Theta,
\end{equation}
where $z_j \in L^2(\X)$, $j=1, \ldots, \Np$.
Let $F$ be as in~\cref{equ:F}. Then, 
for $j = 1, \ldots, \Np$,
\begin{enumerate}[label=(\alph*)]
\item 
$\displaystyle\frac{\partial F}{\partial \theta_j}(\theta) 
= \int_\X \frac{\partial f}{\partial \theta_j}(s, \theta) v(s) \, ds$,
\item and $\displaystyle\frac{\partial F}{\partial \theta_j} \in L^2(\Theta)$. 
\end{enumerate}
\end{proposition}
\begin{proof}
Showing (a) amounts to establishing the standard requirements
for differentiating under the integral sign; see e.g.,~\cite[Theorem 2.27]{Folland99}.
Without loss of generality, we assume 
$\bar{f} \equiv 0$.
First, we note that for each $\theta \in \Theta$,
\[
    \int_\X \left| \frac{\partial f}{\partial \theta_j}(s, \theta) v(s)\right| \, ds
    \leq 
\Big[\int_\X \Big(\frac{\partial f}{\partial \theta_j}(s, \theta)\Big)^2 \, ds\Big]^{1/2}
\Big[\int_\X v(s)^2 \, ds\Big]^{1/2} < \infty,
\quad j = 1, \ldots, \Np,
\] 
where we used the Cauchy--Schwarz inequality and
\cref{assump:main}(b),(c). Next, we note that 
$|\frac{\partial f}{\partial \theta_j}(s, \cdot) v(\cdot)| \leq z_j |v|$ and 
applying the Cauchy--Schwartz inequality, we get that $\int_\X |z_j(s)v(s)| \, ds < \infty$.
Thus, assertion (a) follows from~\cite[Theorem 2.27]{Folland99}. The assertion 
(b) of the proposition follows from, \cref{assump:main}(c) and
\begin{multline*}
\int_\Theta \Big(\frac{\partial F}{\partial \theta_j}(\theta)\Big)^2 \, 
\mu(d\theta)
= 
\int_\Theta \Big( \int_\X \frac{\partial f}{\partial \theta_j}(s, \theta) v(s) \, ds
\Big)^2 \, \mu(d\theta)\\
\leq 
\int_\Theta 
\Big[\int_\X \Big(\frac{\partial f}{\partial \theta_j}(s, \theta)\Big)^2 \,ds\Big]
\Big[ \int_\X v(s)^2 \, ds\Big]
\mu(d\theta)
=\| v\|_{L^2(\X)}^2 
\left\| \frac{\partial f}{\partial \theta_j}\right\|^2_{L^2(\X \times \Theta)}
<\infty
\end{multline*}
\end{proof}

Note that the assumption~\cref{equ:bdd_deriv} can in fact be used to conclude
$\frac{\partial F}{\partial \theta_j} \in L^\infty(\Theta)$; we showed square
integrability of these partial derivatives for clarity as this is the result
needed for the purposes of derivative-based GSA.  Note also that the
assumption~\cref{equ:bdd_deriv} can be relaxed in the statement of the
proposition by requiring local (in $\Theta$) boundedness of the partial
derivatives by square integrable (in $\X$) functions. 

The above framework, based on low-rank KL expansions, is useful as 
it provides a natural setting for deploying an adjoint-based
approach for computing the derivatives of the 
KL modes, in models governed by
PDEs (or ODEs). The computational advantage of 
adjoint-based approach is immense: the cost of computing the 
gradient of $f_i$'s does not scale with the dimension of
the input parameter $\theta$. This leads to a computationally 
efficient and scalable framework for computing DGSMs. We detail this 
approach in the next section for models governed by elliptic PDEs 
and demonstrate its effectiveness 
in numerical examples in~\cref{sec:numerics}.

\section{Adjoint-based GSA for models governed by elliptic PDEs} \label{sec:DGSM_PDE}
We consider a linear elliptic PDE with a random coefficient function:
\begin{equation}\label{equ:fwd}
    \begin{aligned}
      -\nabla \cdot (\kappa  \nabla p) &= b \quad \text{ in }\D, \\
                    p  &= g \quad \text{ on } \GD, \\
                    \kappa \nabla p \cdot n &= h \quad \text{ on } \GN.
    \end{aligned}
\end{equation}
The coefficient field $\kappa$ is modeled as a log-Gaussian random field whose
covariance operator is given by $\Cp$. As is common practice in 
the uncertainty quantification community, we represent the random field
coefficient $\kappa$ using a truncated KL expansion. Namely, 
let 
\[
\hat{a}(x, \theta) = \bar{a}(x) + \sum_{j = 1}^\Np \sqrt{\lambda_j(\Cp)} \theta_j e_j(x)
\]
be a truncated KL expansion of the log-permeability field, $a(x, \theta) = \log \kappa(x, \theta)$. 
We consider the weak form of the PDE. The associated 
trial and test function spaces are, respectively,
\[
   \Vs = \{ v \in H^1(\D) : v\mid_{\GD} = g\}, \quad 
   \Vo = \{ v \in H^1(\D) : v\mid_{\GD} = 0\}.
\]
The weak form of~\cref{equ:fwd} is as follows: find $p \in \Vs$ such that 
\begin{equation}\label{equ:fwd_weak}
    \ip{ e^{\hat{a}(x, \theta)}  \nabla p}{\nabla \ut{p}} = \ip{b}{\ut{p}} + \ipg{h}{\ut{p}}, \quad \text{for all }
    \ut{p} \in \Vo,
\end{equation}
where $\ip{\cdot}{\cdot}$ is the $L^2(\D)$ inner product, and
$\ipg{\cdot}{\cdot}$ is $L^2(\GN)$ inner product. 
Let $\X$ a closed subset of $\D$, and let
$\Q : L^2(\D) \to L^2(\X)$ be the restriction operator
\[
\Q u = \restr{u}{\X}.
\]
Below we also need the adjoint $\Q^*$ of $\Q$:
it is straightforward to see that 
$\Q^*:L^2(\X) \to L^2(\D)$ is given by 
\[
    (\Q^* u)(x) = \begin{cases} u(x), \quad x \in \X\\ 0, \quad x \notin \X.\end{cases}.
\]
We consider the QoI, 
\[
   f(x, \theta) = \Q p(x, \theta), 
\]
and consider its truncated KL expansion
\begin{equation}\label{equ:KLE_QoI}
    f(x, \theta) \approx \hat{f}(x, \theta) := \bar{f}(x) + \sum_{i = 1}^\Nq 
            \sigma_i f_i(\theta) \phi_i(x),
            \quad \text{with } \sigma_i = \sqrt{\lambda_i(\Cq)}. 
\end{equation}
where
\[
    f_i(\theta) = \fac 
\int_\X \big(f(x, \theta) - \bar{f}(x)\big) \phi_i(x) \, dx 
=\fac \int_\X \big(\Q p(x, \theta) - \bar{f}(x)\big)\phi_i(x) \, dx,
\]
where $p$ is the solution of~\cref{equ:fwd_weak}. 
We consider adjoint-based computation of 
$\frac{\partial f_i}{\partial \theta_j}$ for $i,j \in \{1, \ldots, \Nq\} \times 
\{1, \ldots, \Np\}$. 

\textbf{Computing gradient of $f_i$'s}. 
To compute the gradient we follow a formal Lagrange approach. 
We consider the Lagrangian
\[
   \L(p, \theta, q) = \fac \int_\X (\Q p - \bar{f})\phi_i \, dx 
   + \ip{ e^{\hat{a}(x, \theta)}  \nabla p}{\nabla q} - \ip{b}{q} - \ipg{h}{q}.
\]
Here $q$ is a Lagrange multiplier, which in the present context is referred to as the 
adjoint variable.
Taking variational derivatives of $\L$ with respect to $q$, and
$p$, give the state and the adjoint equations, respectively. 
In particular, the adjoint equation is found 
by considering
\[
   \frac{d}{d\epsilon} \L(p + \epsilon \ut{p}, \theta, q) \mid_{\epsilon = 0}~ = 0, \quad \text{for all } \ut{p} \in \Vo.
\]
This gives, 
\[
   \fac \ip{\Q\ut{p}}{\phi_i}_\X +  \ip{ e^{\hat{a}(x, \theta)}  \nabla \ut{p}}{\nabla q} = 0, \quad \text{for all } \ut{p} \in \Vo.
\]
The weak form of the adjoint equation can be stated as:
find $q \in \Vo$ such that
\[
\ip{ e^{\hat{a}(x, \theta)}  \nabla q}{\nabla \ut{p}}
= - \fac \ip{\Q^*\phi}{\ut{p}}, \quad \text{for all } \ut{p} \in \Vo.
\]
The strong form of the adjoint equation is 
\begin{equation}\label{equ:adj}
    \begin{aligned}
      -\nabla \cdot (\kappa  \nabla q) &= -\fac\Q^*\phi_i \quad \text{ in }\D, \\
                    q  &= 0 \quad \text{ on } \GD, \\
                    \kappa \nabla q \cdot n &= 0 \quad \text{ on } \GN.
    \end{aligned}
\end{equation}
Letting $p$ and $q$ be the solutions of the state and adjoint 
equations respectively, 
\begin{equation}\label{equ:grad}
   (\nabla_\theta f_i)^\top \ut{\theta} =
   \frac{d}{d\epsilon} \L(p, \theta + \epsilon \ut{\theta}, q) \mid_{\epsilon = 0}
  ~= \ip{ (\hat{a}(x, \ut{\theta}) -\bar{a}(x)) e^{\hat{a}(x, \theta)}  \nabla p}{\nabla q}, 
   \quad 
   \ut{\theta} \in \mathbb{R}^\Np.
\end{equation}
In particular, letting $\ut{\theta}$ be the $j$th coordinate direction in 
$\R^\Np$, we get
\[
   \frac{\partial f_i}{\partial \theta_j} = 
   \sqrt{\lambda_j(\Cp)} \ip{ e_j e^{\hat{a}(x, \theta)}  \nabla p}{\nabla q}.
\]

We can also consider a QoI of the form
\[
    f(\cdot, \theta) = \restr{p(\cdot,\theta)}{\GN},
\]
as done in one of our numerical examples in~\cref{sec:numerics}.
Computing the gradient for this QoI can be done in a similar way as above, 
except, in this case the adjoint equation takes the form:
\begin{equation}\label{equ:adj_trace}
    \begin{aligned}
      -\nabla \cdot (\kappa  \nabla q) &= 0 \quad \text{ in }\D, \\
                    q  &= 0 \quad \text{ on } \GD, \\
                    \kappa \nabla q \cdot n &= -\fac \phi_i \quad \text{ on } \GN.
    \end{aligned}
\end{equation}

Notice that evaluating the adjoint-based expression for the 
gradient of $f_i$, 
requires two PDE solves: we need to 
solve the state (forward) equation~\cref{equ:fwd} 
and the adjoint equation~\cref{equ:adj}. 
Moreover, the forward solves can be reused across the KL modes, 
and thus, computing the gradient of $\hat{f}$ in~\cref{equ:KLE_QoI}
requires $1 + \Nq$ PDE solves, independently of the dimension $\Np$ of the uncertain
parameter $\theta$. As shown in our numerical examples, a small $\Nq$ often 
results in suitable representations of the QoI $f$, due to the, often 
observed, rapid decay
of the eigenvalues $\lambda_i(\Cq)$.

\textbf{DGSM computation}.
In practice, the KL expansion should be computed numerically.  As mentioned
before, this can be accomplished using Nystr\"{o}m's method, which is the
approach taken in the present work, and requires an ensemble of model
evaluations $\{f(\cdot, \theta^k)\}_{k = 1}^\Ns$, typically with a modest sample size
$\Ns$. The model evaluations can be used to compute the approximate KL expansion
following~\cite[Algorithm 1]{AlexanderianReeseSmithEtAl18}. This same set of 
samples can be used for computing the DGSMs, $\nu_j(f_i)$, $j = 1, \ldots, 
\Np$, $i = 1, \ldots, \Nq$. These require an additional adjoint solve per
KL mode, and for each sample point $\theta^k$, $k = 1, \ldots, \Ns$.
Thus, the overall computational cost is
$\Ns(1 + \Nq)$ PDE solves. Note that the computational cost, in terms of 
PDE solves, is  
independent of the dimension $\Np$ of the uncertain parameter vector.
To compute the DGSM-based bound on functional Sobol' indices we also need
to compute $\trace(\Cq)$; this can be approximated accurately by summing
the dominant eigenvalues of $\Cq$, available from computing the KL expansion of 
$f$. The steps for DGSM computation using the present strategy are outlined
in~\cref{alg:DGSM_KLE}.

\newcommand{\Compute}{\textbf{compute}~}
\newcommand{\getKLE}{\textbf{getKLE}}
\newcommand{\Solve}{\textbf{solve}~}
\renewcommand{\algorithmicrequire}{\textbf{Input:}}
\renewcommand{\algorithmicensure}{\textbf{Output:}}
\begin{algorithm}
\caption{Algorithm for computing $\mathfrak{B}_j := \N_j(f; \X)/\trace(\Cq)$, 
$j = 1, \ldots, \Np$.}
\label{alg:DGSM_KLE}
\begin{algorithmic}[1]
\REQUIRE Parameter samples $\{ \theta^k\}_{k=1}^\Ns$ 
\ENSURE  Approximate DGSM-based bounds 
$\hat{\mathfrak{B}}_j$, $j = 1, \ldots, \Np$ 

\FOR{$k = 1, \ldots, \Ns$}
\STATE \Solve forward model~\cref{equ:fwd} with $\kappa = \exp{\hat{a}(\cdot, \theta^k)}$ 
\STATE \Compute QoI $f(\cdot, \theta^k)$
\ENDFOR
\STATE $[\{\lambda_i\}_{i=1}^\Nq, \{\phi_i\}_{i=1}^\Nq] = \getKLE(\{f(\cdot, \theta^k)\}_{k=1}^\Ns)$
\FOR{$k = 1, \ldots, \Ns$}
\FOR{$i = 1, \ldots, \Nq$}
\STATE \Solve adjoint problem~\cref{equ:adj} with $\kappa = \exp{\hat{a}(\cdot, \theta^k)}$
\STATE \Compute $\frac{\partial f_i(\theta^k)}{\partial \theta_j}$, $j = 1, \ldots, \Np$
using~\cref{equ:grad}
\ENDFOR
\ENDFOR

\STATE \Compute $\hat\nu_j(f_i) = \frac{1}{\Ns} \sum_{k=1}^\Ns 
\left[ \frac{\partial f_i(\theta^k)}{\partial \theta_j}\right]^2$
\STATE \Compute $T = \sum_{i=1}^\Nq \lambda_i$
\STATE \Compute $\hat{\mathfrak{B}}_j = \sum_{i=1}^\Nq \lambda_i \hat{\nu}_j(f_i) / T$, 
$j = 1, \ldots, \Np$

\end{algorithmic}
\end{algorithm}
In step 5 of~\cref{alg:DGSM_KLE}, \getKLE~indicates a
procedure that given sample realizations of the process $f$, computes its KL
expansion numerically. As mentioned before, this can be done, e.g., using
Nystr\"{o}m's method; see e.g.,~\cite[Algorithm 1]{AlexanderianReeseSmithEtAl18}.

\section{Numerical examples}\label{sec:numerics}
In this section, we present three numerical examples. In 
\cref{sec:cholera}, we consider an example involving  a nonlinear ODE
system with a time-dependent QoI, which is used to illustrate functional DGSMs
and the DGSM-based bound derived in \cref{thm:bound_main}.
Sections~\ref{sec:poi_spe} and~\ref{sec:biotransport} concern models governed
by elliptic PDEs that have  spatially distributed QoIs in one and two space
dimensions, respectively. For the PDE-based examples we implement the 
adjoint-based GSA framework described in \cref{sec:DGSM_PDE} and
illustrate its effectiveness.
\subsection{Sensitivity analysis for a model of cholera epidemics}
\label{sec:cholera} Consider the cholera model developed
in~\cite{HartleyEtAl2005}. We analyze the sensitivity of the infected
population as a function of time to uncertainties in model parameters.  This
problem was also studied in~\cite{AlexanderianGremaudSmith17} within the
context of variance-based GSA for time-dependent processes. 

\subsubsection{Model description}
A population of $N_\text{pop}$ individuals is split into susceptible,
infectious, and recovered individuals, which are denoted by $S$, $I$, and $R$,
respectively.  The concentrations of highly-infectious bacteria, $B_H$ and 
lowly-infectious bacteria, $B_L$ are also considered. These concentrations are measured in cells per milliliter.
According to the model developed in~\cite{HartleyEtAl2005}, 
the time-evolution of the state variables is governed by 
the following system of ODEs.
\newcommand{\ds}{\displaystyle}
\begin{equation}
\begin{array}{l}
  {\ds \frac{dS}{dt} = bN_\text{pop} - \beta_L S \frac{B_L}{\kappa_L + B_L} - \beta_H S \frac{B_H}{\kappa_H + B_H}  - b S} \\
  \noalign{\medskip}
  {\ds \frac{dI}{dt} =  \beta_L S \frac{B_L}{\kappa_L + B_L} + \beta_H S \frac{B_H}{\kappa_H + B_H} - (\gamma + b)I} \\
  \noalign{\medskip}
  {\ds \frac{dR}{dt} = \gamma I - b R} \\
  \noalign{\medskip}
  {\ds \frac{d B_H}{d t} = \xi I - \chi B_H} \\
  \noalign{\medskip}
  {\ds \frac{d B_L}{dt} = \chi B_H - \delta B_L}
\end{array}
\label{eq1}
\end{equation}
with initial conditions 
$(S(0), I(0), R(0), B_H(0), B_L(0)) = (S_0, I_0, R_0, B_{H_0},  B_{L_0})$.
The parameter units and nominal values from \cite{HartleyEtAl2005} are compiled in 
\cref{tbl:cholera_param}.  
We consider a 
total population of $N_\text{pop}=10{,}000$ and let the 
initial states be as follows: $S_0 = N_\text{pop} - 1$, $I_0 = 1$, $R_0 = 0$, and $B_{H_0} =
B_{L_0} = 0$. We solve the problem up to time  $T = 150$ using the 
\verb+ode45+ solver provided in \textsc{Matlab}~\cite{MATLAB}. 

\renewcommand{\arraystretch}{1.1}
\begin{table}[!tt]
\begin{center}
\begin{tabular}{l|c|c|c} \hline
Model Parameter & Symbol & Units & Values \\ \hline
Rate of drinking $B_L$ cholera & $\beta_L$ & $\frac{1}{\mbox{\rm week}}$ & 1.5 \\
Rate of drinking $B_H$ cholera & $\beta_H$ & $\frac{1}{\mbox{\rm week}}$  & 7.5\\
$B_L$ cholera carrying capacity & $\kappa_L$ & $\frac{\mbox{\rm \# bacteria}}{m\ell}$ & $10^6$  \\
$B_H$ cholera carrying capacity & $\kappa_H$ & $\frac{\mbox{\rm \# bacteria}}{m\ell}$ & $\frac{\kappa_L}{700}$ \\
Human birth and death rate & $b$ & $\frac{1}{\mbox{\rm week}}$  & $\frac{1}{1560}$ \\
Rate of decay from $B_H$ to $B_L$ & $\chi$ & $\frac{1}{\mbox{\rm week}}$ & $\frac{168}{5}$ \\
Rate at which infectious individuals & $\xi$  & $\frac{\mbox{\rm \# bacteria}}{\mbox{\rm \# individuals} \cdot m \ell \cdot \mbox{\rm week}}$ & 70 \\ \noalign{\vspace{-.1truein}}
spread $B_H$ bacteria to water & & & \\
Death rate of $B_L$ cholera & $\delta$ & $\frac{1}{\mbox{\rm week}}$ & $\frac{7}{30}$ \\
Rate of recovery from cholera & $\gamma$ & $\frac{1}{\mbox{\rm week}}$ & $\frac{7}{5}$ \\ \hline
\end{tabular}
\end{center}
\renewcommand{\arraystretch}{1.0}
\caption{Cholera model parameters from \cite{HartleyEtAl2005,AlexanderianGremaudSmith17}.}
\label{tbl:cholera_param}
\end{table}

To simplify the notation we use a generic vector $y \in \R^5$ 
to denote the 
state vector---$y = (y_1, y_2, y_3, y_4, y_5)^\top = (S, I, R, B_H, B_L)^\top$---and
denote the right hand side of the ODE system by $g(y; c)$, where
$c = (\beta_L, \beta_H, \kappa_L, b, \chi, \theta, \delta, \gamma)$ 
is the vector of uncertain model parameters.
The uncertainties in $c$ are parameterized by a random vector $\theta \in \R^8$
with iid $U(-1, 1)$ entries as follows:
\[
    c_i(\theta_i) = \frac12(a_i + b_i) + \frac12(b_i - a_i) \theta_i, \quad
                    i = 1, \ldots, 8,
\]
with $[a_i, b_i]$ the physical parameter ranges for $c_i$, adapted
from~\cite{AlexanderianGremaudSmith17}. 
The solution of the 
system is a random process, $y = y(t; \theta)$.
We focus on the infected population $I(t, \theta) = y_2(t; \theta)$, 
for $t \in [0, 150]$.
In~\cref{fig:cholera_infected}, we depict the time
evolution of $I(t, \theta)$ at the nominal parameter vector given by $\theta =
(0, 0, \ldots, 0)^\top \in \R^8$.  
\begin{figure}[ht]\centering
\includegraphics[width=.45\textwidth]{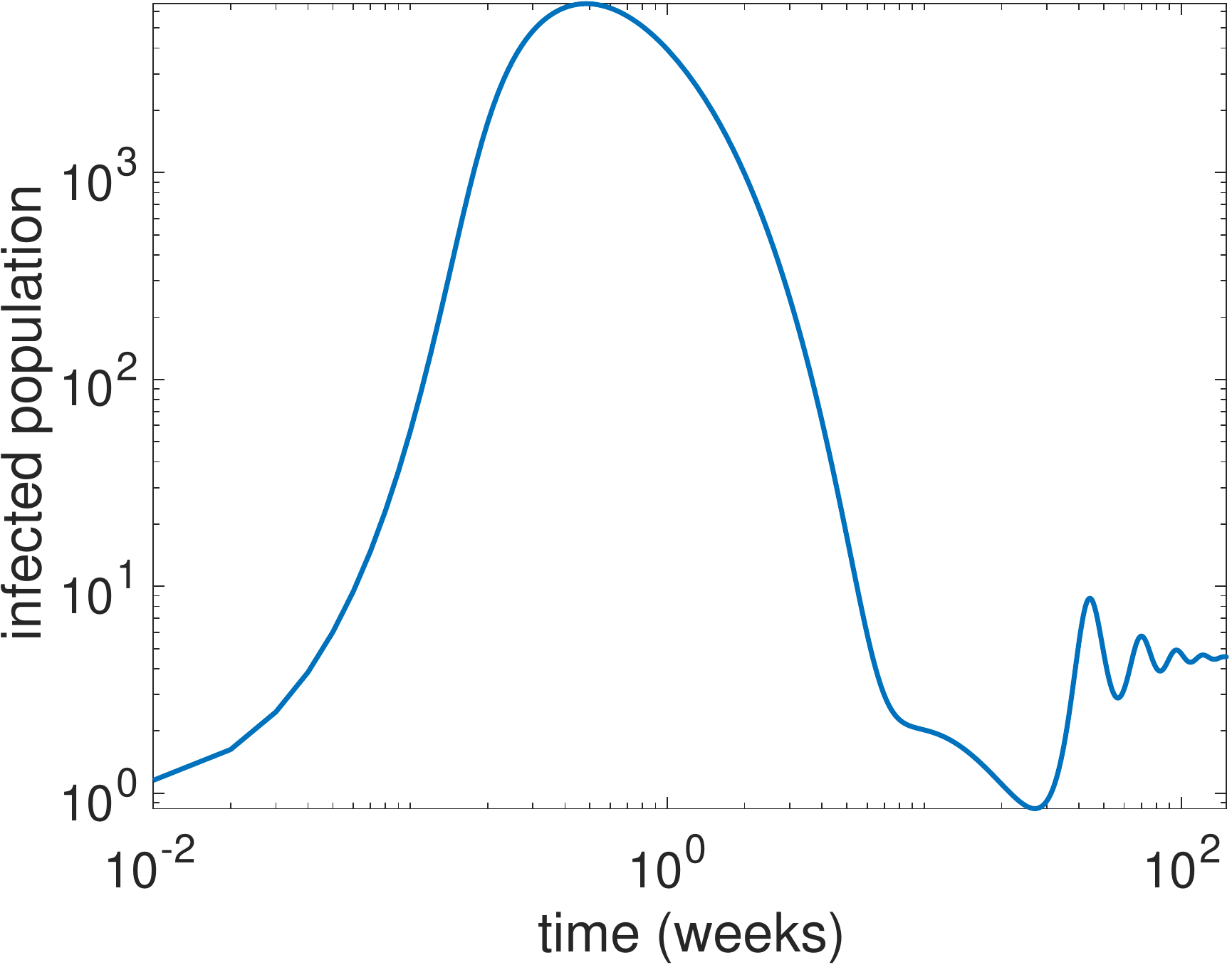}
\caption{The infected population $I(t; \theta)$ with $\theta = {0}$.}
\label{fig:cholera_infected}
\end{figure}
\subsubsection{Derivative-based GSA}
To compute the the partial derivatives 
$s_j(t; \theta) = \frac{\partial y(t; \theta)}{\partial \theta_j}$, $j = 1, \ldots, 8$,
needed for DGSM computation, we 
rely on the so called direct 
approach; this involves integrating the sensitivity 
equations~\cite{MalyPetzold96,SanduDaescuCarmichael03} along 
with the ODEs describing the system state. Specifically, we need to integrate 
the system
\[
    \begin{aligned}
    y' &= g(y; c(\theta)), \quad y(0) = y_0,\\
    s'_i &= {J}{s}_i + \frac{\partial g}{\partial \theta_i}, 
    \quad {s}_i(0) = {0}, \quad i = 1, \ldots, \Np.
    \end{aligned}
\] 
Here ${J}$ is the Jacobian 
$J_{ij} = \frac{\partial g_i}{\partial \theta_j} =  
\frac{\partial g_i}{\partial c_j}
\frac{\partial c_j}{\partial \theta_j}$, $i, j = 1, \ldots, \Np$.
In the present example this results in an ``augmented state vector''
$[y^\top \, s_1^\top \, \cdots \, s_8^\top]^\top \in \R^{45}$.

First, we consider the pointwise-in-time DGSMs, 
$\nu_j(f(t, \cdot))$, $j = 1, \ldots, 8$, for $t \in [0, 150]$ in
\cref{fig:GSA_overtime}~(left). To ensure an accurate estimate of 
the DGSMs, we approximate the integral over the parameters with a Monte
Carlo sample of size $10^5$. As seen in \cref{fig:GSA_overtime}~(left), 
these pointwise-in-time DGSMs are not 
straightforward to interpret. 
A clearer picture is obtained by considering 
\[
    \N_j(I; [0, t]) := \int_0^t \nu_j(I(s, \cdot))\, ds, \quad t \in [0, 150], 
\]
which amounts to computing the functional DGSMs over successively larger time
intervals; the results are reported in \cref{fig:GSA_overtime}~(right).

Finally, to get an overall picture, we compute the DGSM-based upper bounds on
the functional Sobol' indices, as given by \cref{thm:bound_main},
with $\X = [0, 150]$; see \cref{fig:cholera_bound}, where
we report the functional total Sobol' indices along with the DGSM-based bounds which are computed with Monte Carlo (MC) sample
sizes of $10^5$ and 100. Note that a small MC sample is very effective in
detecting the unimportant parameters. 

By \cref{thm:bound_main}, we know that a small DGSM-based bound for a
given parameter implies the corresponding total Sobol' index is small,
indicating the parameter is unimportant.  In the present experiment, we set an
\emph{importance threshold} of $0.05$. A parameter whose DGSM-based bound is
smaller than this importance threshold will be considered unimportant.  The
results reported in \cref{fig:cholera_bound} indicate that unimportant
parameters are given by $\theta_j$ with $j \in \{1, 4, 5, 7\}$. This is
consistent with results reported in~\cite{AlexanderianGremaudSmith17}, where
the statistical accuracy of the reduced model, obtained by fixing these unimportant
parameters was demonstrated numerically. Both panels of 
\cref{fig:cholera_bound} show the same information; however, in the right
panel we use a logarithmic scale in the vertical axis to clearly illustrate the
bound derived in \cref{thm:bound_main},  for the small functional Sobol'
indices.
\begin{figure}[ht]\centering
\includegraphics[width=.4\textwidth]{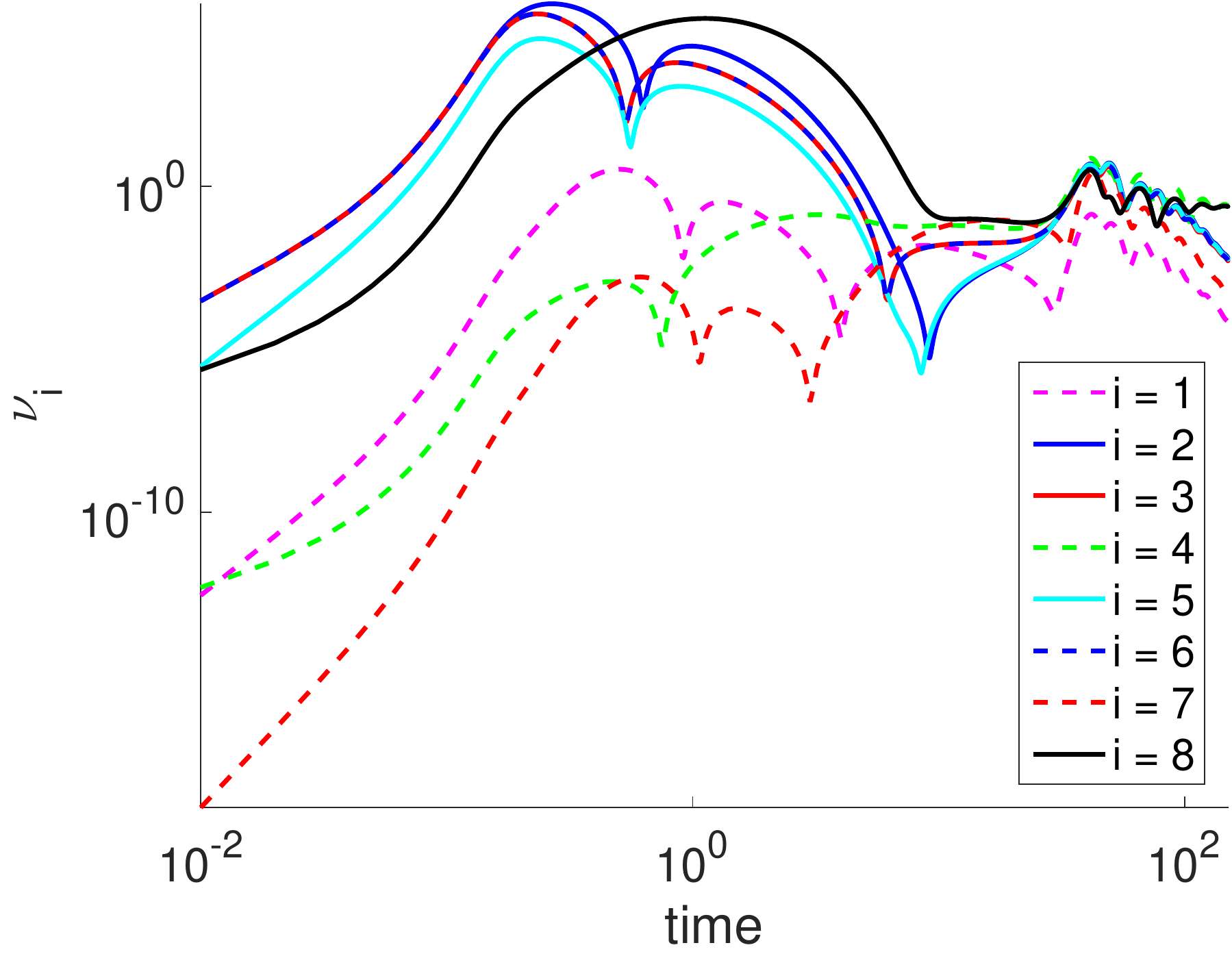}
\includegraphics[width=.4\textwidth]{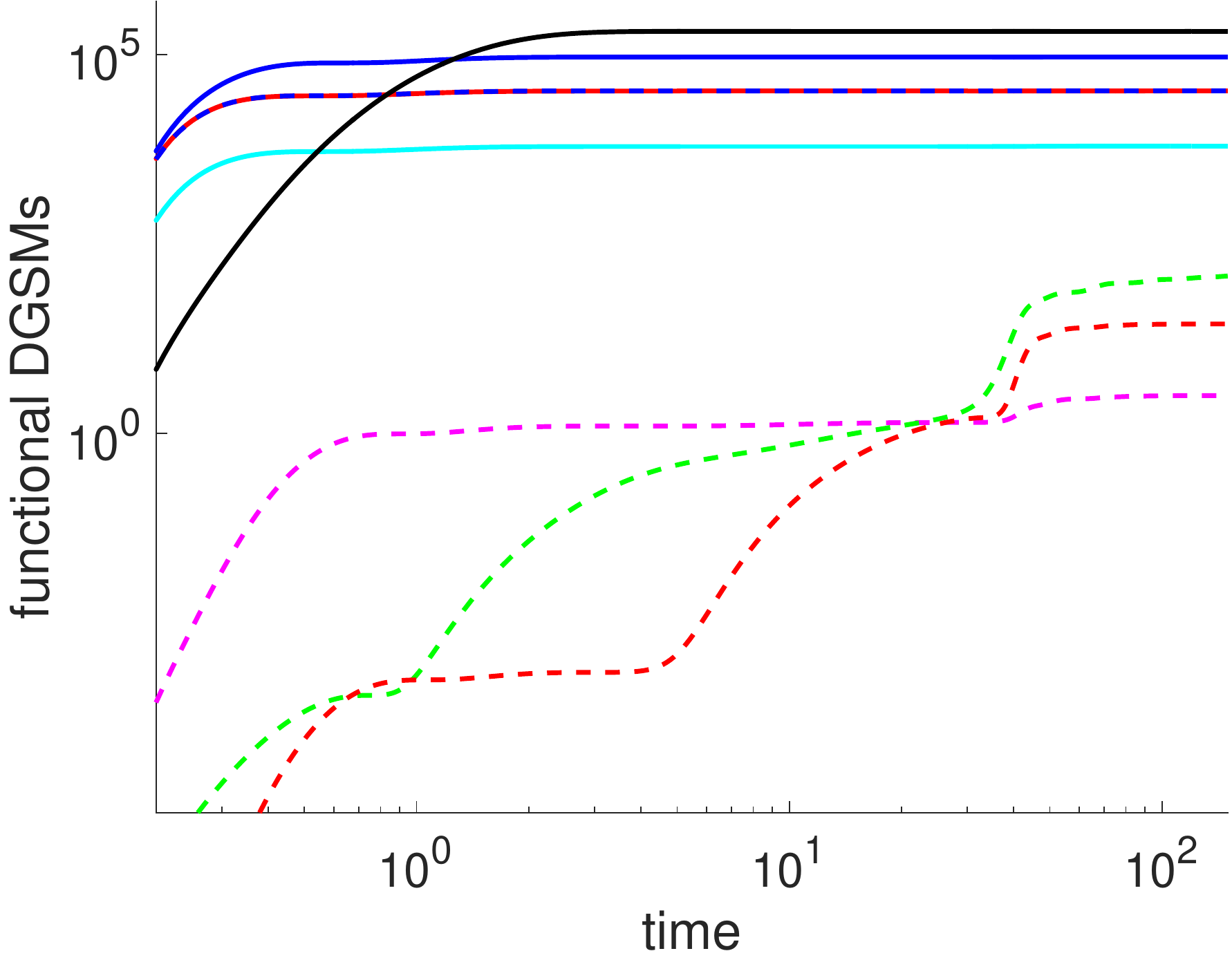}
\caption{Pointwise-in-time DGSMs $\nu_j(I(t, \cdot))$ (left) and 
functional DGSMs 
$\N_j(I; [0, t])$ (right) for $t\in[0,150]$.}
\label{fig:GSA_overtime}
\end{figure}
\begin{figure}[ht!]\centering
\includegraphics[width=.4\textwidth]{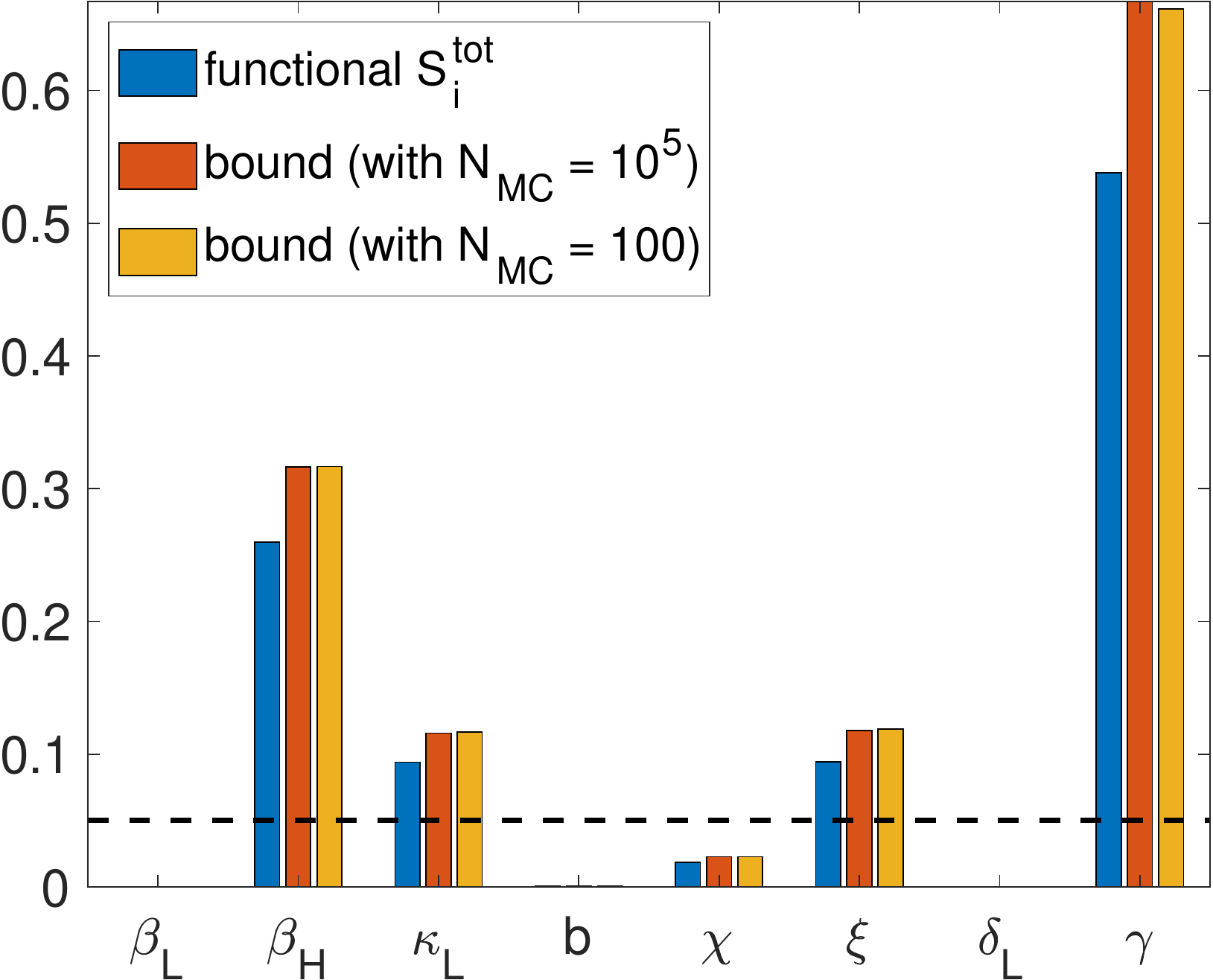}
\includegraphics[width=.4\textwidth]{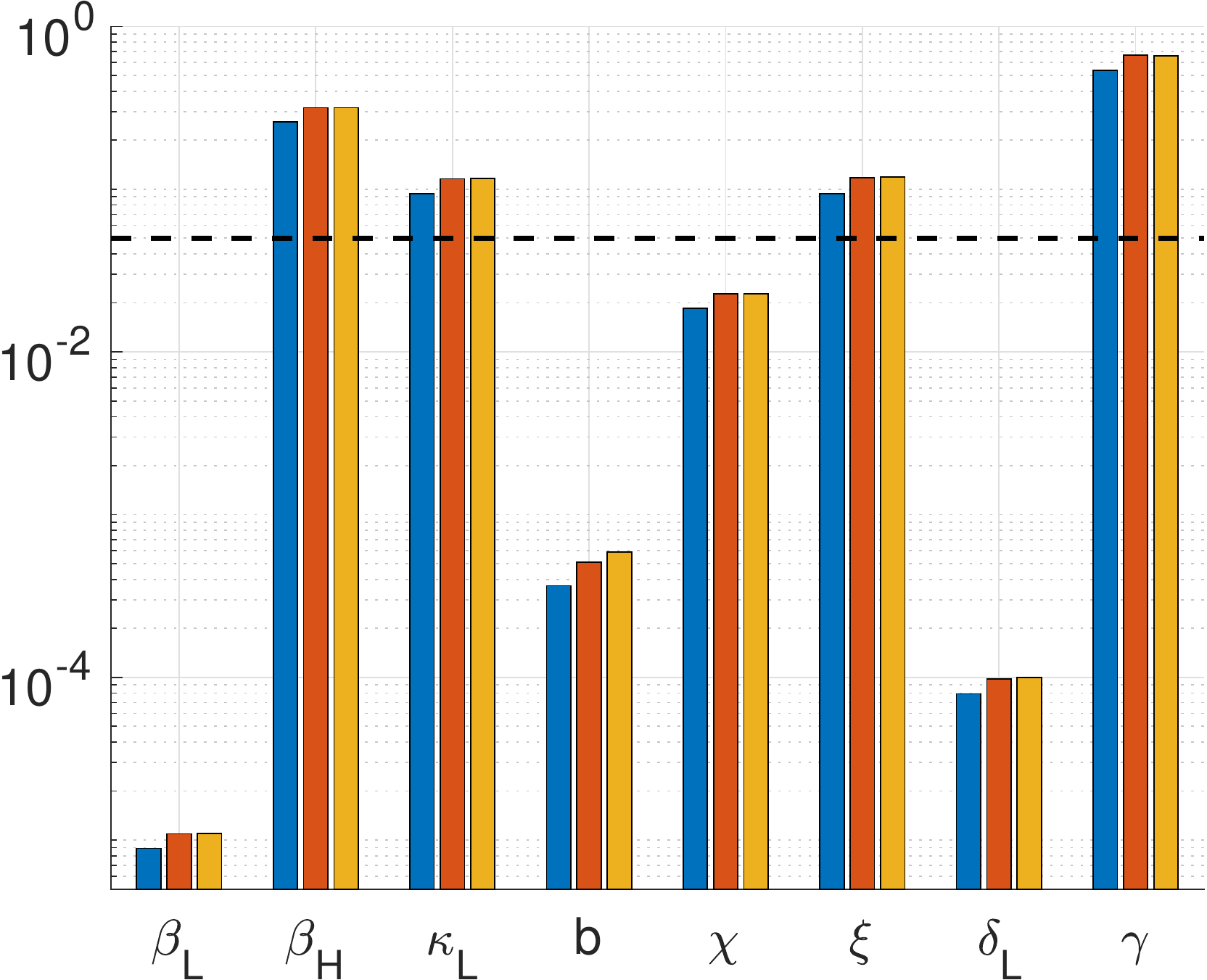}
\caption{Left: The functional Sobol indices and the corresponding bounds proven 
in \cref{thm:bound_main} for the cholera model; right: 
the same information as in the left plot, except we use log-scale on 
$y$-axis to clearly show $\ST_j(I; \X)$ and the corresponding bound, 
for small indices; the dashed black line indicates 
$y = 0.05$ that could be a reasonable tolerance to decide which random input
is unimportant.}
\label{fig:cholera_bound}
\end{figure}

\subsection{Sensitivity analysis in a subsurface flow
problem}~\label{sec:poi_spe}
In this section, we elaborate our proposed approach for sensitivity analysis
and dimension reduction on a model problem motivated by subsurface flow
applications. 

\subsubsection{Model description}
We consider the following equation modeling the fluid pressure in a 
single phase flow problem:
\begin{equation}\label{equ:poi_spe_detailed}
\begin{aligned}
-\nabla \cdot \Big(\frac{\kappa}{\eta} \, \nabla p\Big) &= b, &\quad \text{in } \D\\
p &= 0 &\quad \text{on } \Gamma_D,\\
\frac{\kappa}{\eta} \nabla p \cdot n &= 0, &\quad \text{on } \Gamma_N
\end{aligned}
\end{equation}
The domain is $\D = (-1, 1) \times (0, 1)$, $\Gamma_D$ is the union of 
the left, bottom, right parts of the boundary, 
and $\Gamma_N$ is the top boundary.
The right hand side function $b(x)$ is defined as a sum of mollified 
point sources, $b(x) = \sum_{i=1}^4 \alpha_i \delta_{x_i}(x)$, where 
\[
   \delta_{x_i}(x) = \frac{1}{2\pi L} \exp\left\{-\frac{1}{2L} \|x - x_i\|_2^2\right\},
\]
with
$x_1 = (-0.6, 0.2), x_2 = (-0.2, 0.4)$, and $x_3 = (0.2, .6)$, and
$x_4 = (0.6, 0.8)$. We chose $(\alpha_1, \alpha_2, \alpha_3, \alpha_4)
=(2, 5, 5, 2)$.
In this problem, we assume viscosity is $\eta = 1$ and 
consider uncertainties in the permeability field $\kappa$,
which is modeled as a log-Gaussian process: 
\begin{equation}\label{equ:param_process}
     \log \kappa(x, \omega) =: a(x, \omega) = \bar{a}(x) + \sigma_a z(x, \omega),
     \quad x \in \D, \omega \in \Omega,
\end{equation}
where $\Omega$ is an appropriate sample space, and 
$z(x, \omega)$ is a Gaussian process with mean zero and 
covariance function given by
\[
c_z(x, y) = 
\exp\left\{-\frac{|x_1 - y_1|}{\ell_x} - \frac{|x_2 - y_2|}{\ell_y}   \right\}, \quad x, y \in \D.
\]
In the present example, we use $\ell_x = 1/2$ and $\ell_y = 1/4$, implying
stronger correlations in the horizontal direction. The covariance 
operator $\Cp$ is defined by $\Cp u = \int_\X c_z(\cdot, y) u(y) \, dy$.
The mean of the process
$\bar{a}(x)$ is adapted from the simulated permeability data from the
Society for Petroleum Engineers (SPE) 2001 
Comparative Solutions Project~\cite{SPE}; see \cref{fig:spe_mean}.
For this problem we use $\sigma_a = 1.6$.
\begin{figure}\centering
\includegraphics[width=.5\textwidth]{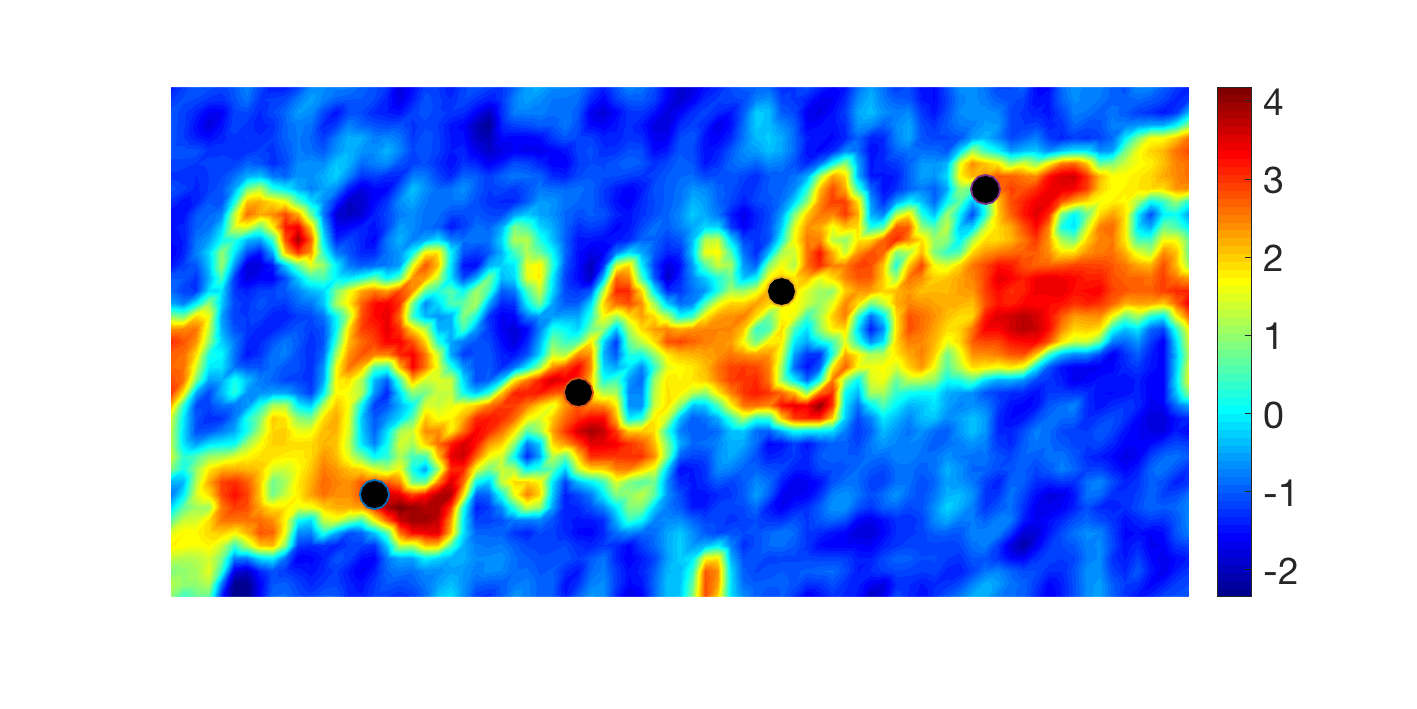}
\caption{Mean log-permeability field. The black dots indicate 
point source locations.}
\label{fig:spe_mean}
\end{figure}
We use a truncated KL expansion to represent the log-permeability field:
\begin{equation}\label{equ:truncated_KLE}
a(x, \omega) \approx
\bar{a}(x) + \sum_{k=1}^\Np \sqrt{\lambda_k(\Cp)} \theta_k(\omega) e_k(x),
\end{equation}
where $\theta_k,~k=1, 2, \ldots, \Np$ are independent standard normal random
variables, and $\lambda_k(\Cp)$ and $e_k(x)$ are the eigenpairs of the
covariance operator $\Cp$ of $a(x, \omega)$ (which is defined
in terms of the correlation function $c_z$ as before).
Note that when using the truncated KL expansion,
the uncertainty in the log permeability field is
characterized by the random vector $\theta=(\theta_1,\theta_2,\ldots,
\theta_\Np)^\top \in \R^\Np$.

To establish the truncation level, we consider the ratio
\[ 
 r_k = \frac{\sum_{i=1}^k \lambda_i}{\sum_{i=1}^\infty \lambda_i},
 \quad k = 1, 2, 3, \ldots, 
\]
where $\lambda_i$'s are the eigenvalues of the covariance operator $\Cp$. 
We depict the normalized
eigenvalues, $\lambda_k / \lambda_1$ in \cref{fig:spe_spectrum}~(left) and plot
the ratios $r_k$, for $k = 1, \ldots, 1000$. We find that $r_k > 0.9$, for $k =
126$; thus, we retain $\Np = 126$ in the KL expansion of the log-permeability
field. We will see shortly (see \cref{sec:spe_gsa}) that this is  an
unnecessarily large parameter dimension for the quantity of interest under
study.

As an illustration, 
in \cref{fig:perm_samp}, 
we show two realizations of the resulting log-permeability 
field (left) along with 
the corresponding pressure fields (right) obtained by 
solving \cref{equ:poi_spe_detailed}.
\begin{figure}\centering
\includegraphics[width=.4\textwidth]{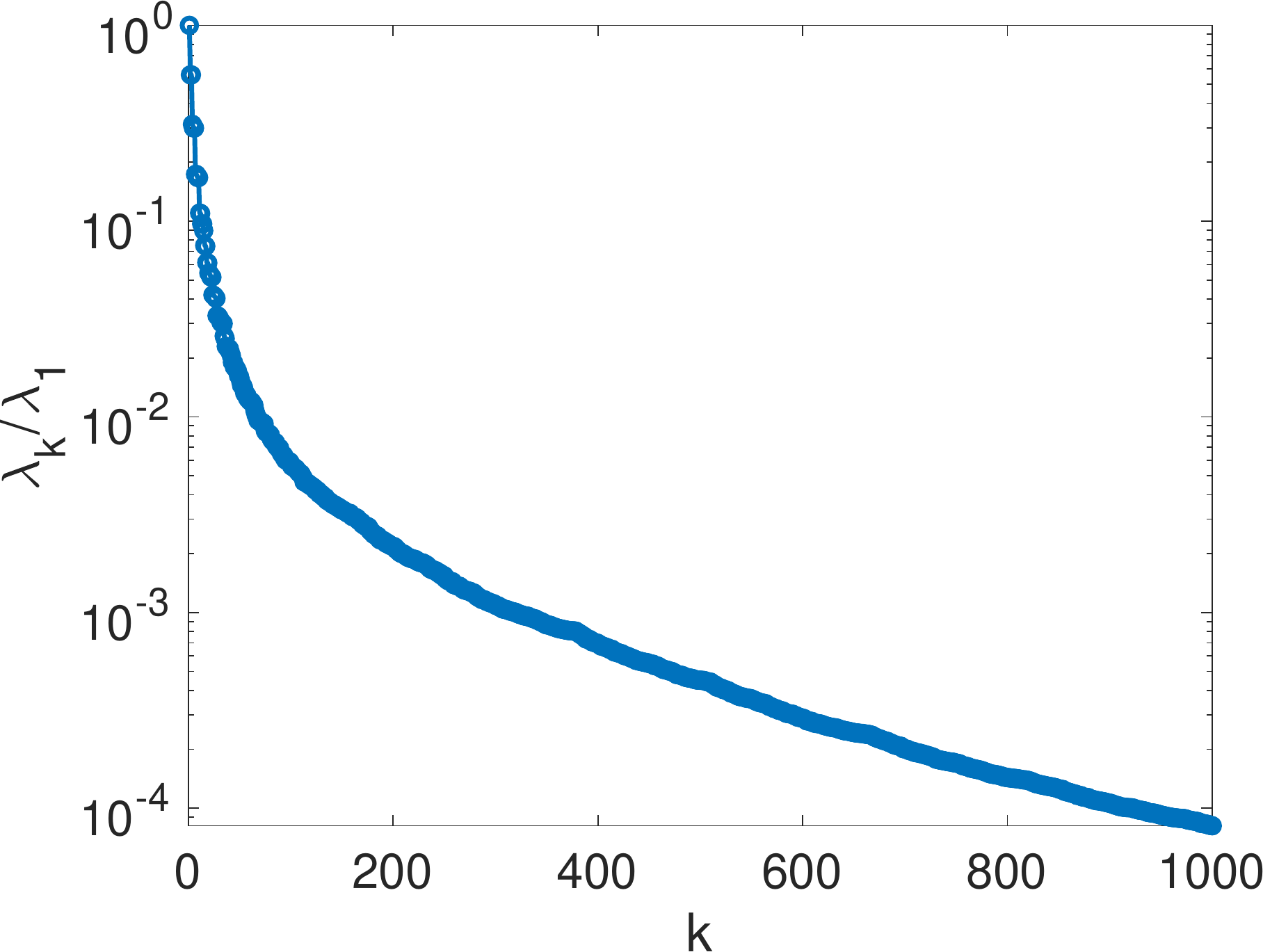}
\includegraphics[width=.4\textwidth]{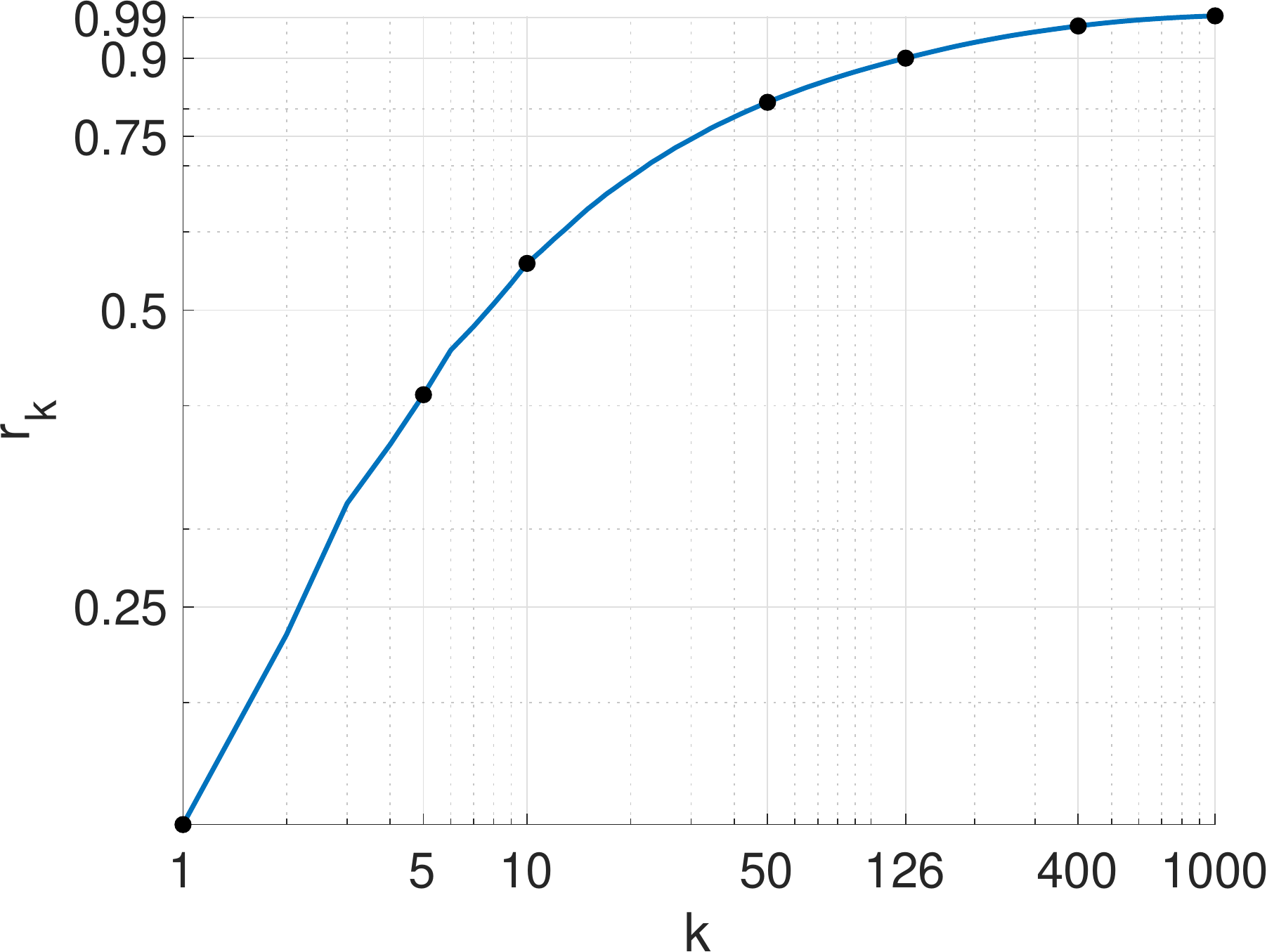}
\caption{Left: the normalized eigenvalues of the
log-permeability field covariance operator; right: the ratios $r_k$, 
for $k = 1, \ldots, 1000$.}
\label{fig:spe_spectrum}
\end{figure}

\begin{figure}\centering
\begin{tabular}{cc}
\includegraphics[width=.45\textwidth]{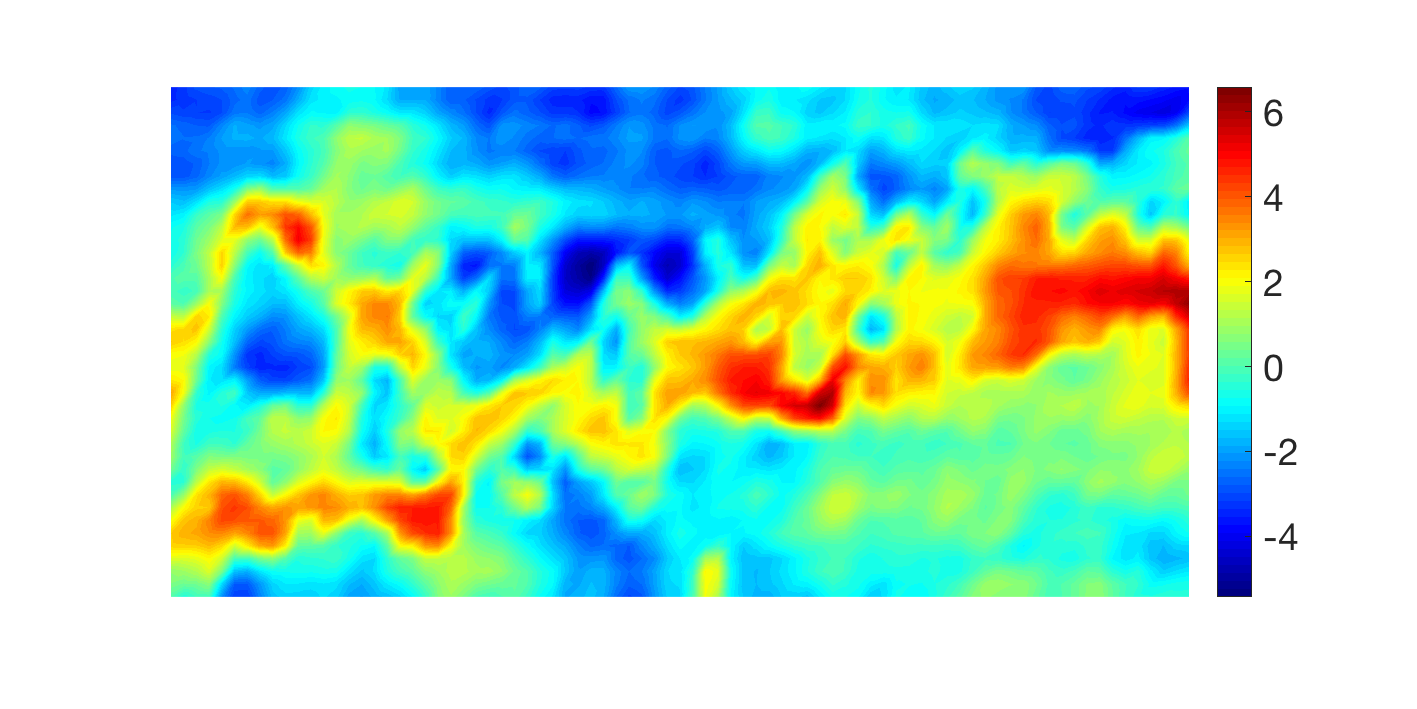}&
\includegraphics[width=.45\textwidth]{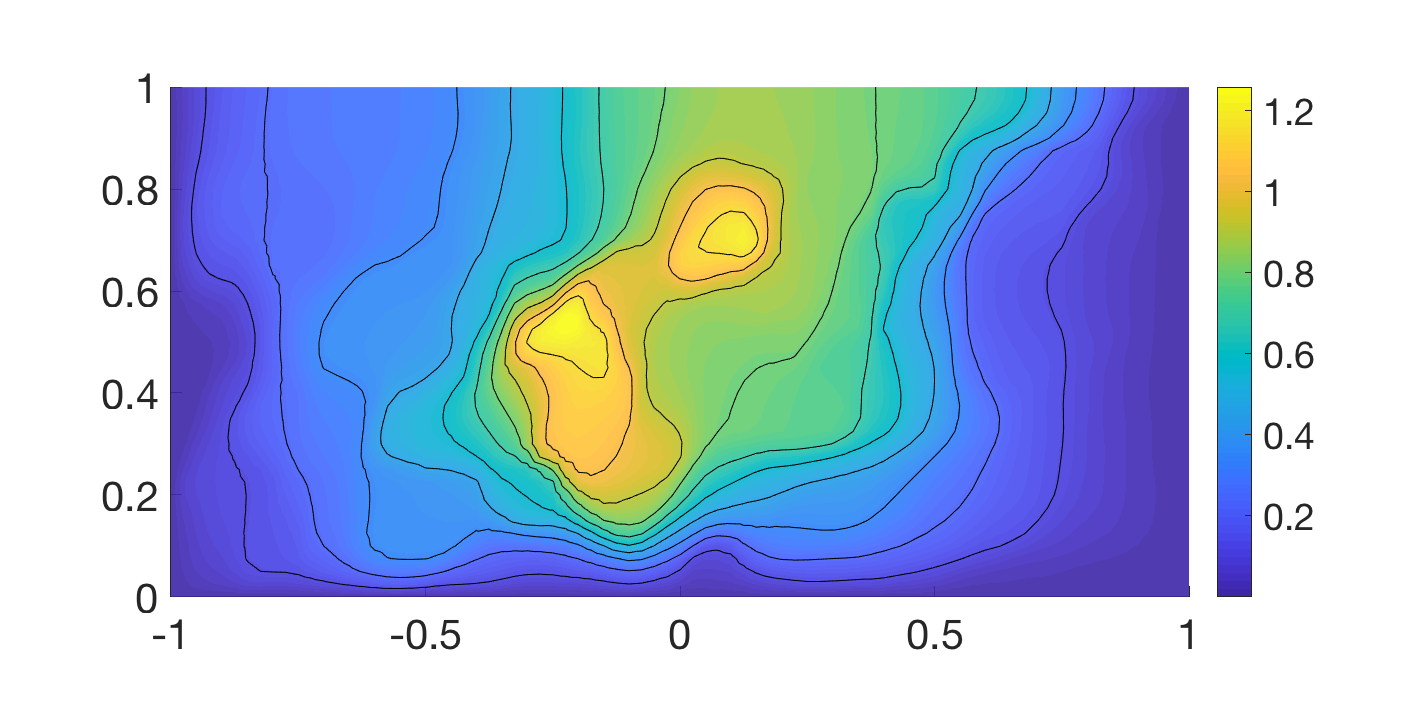}\\
\includegraphics[width=.45\textwidth]{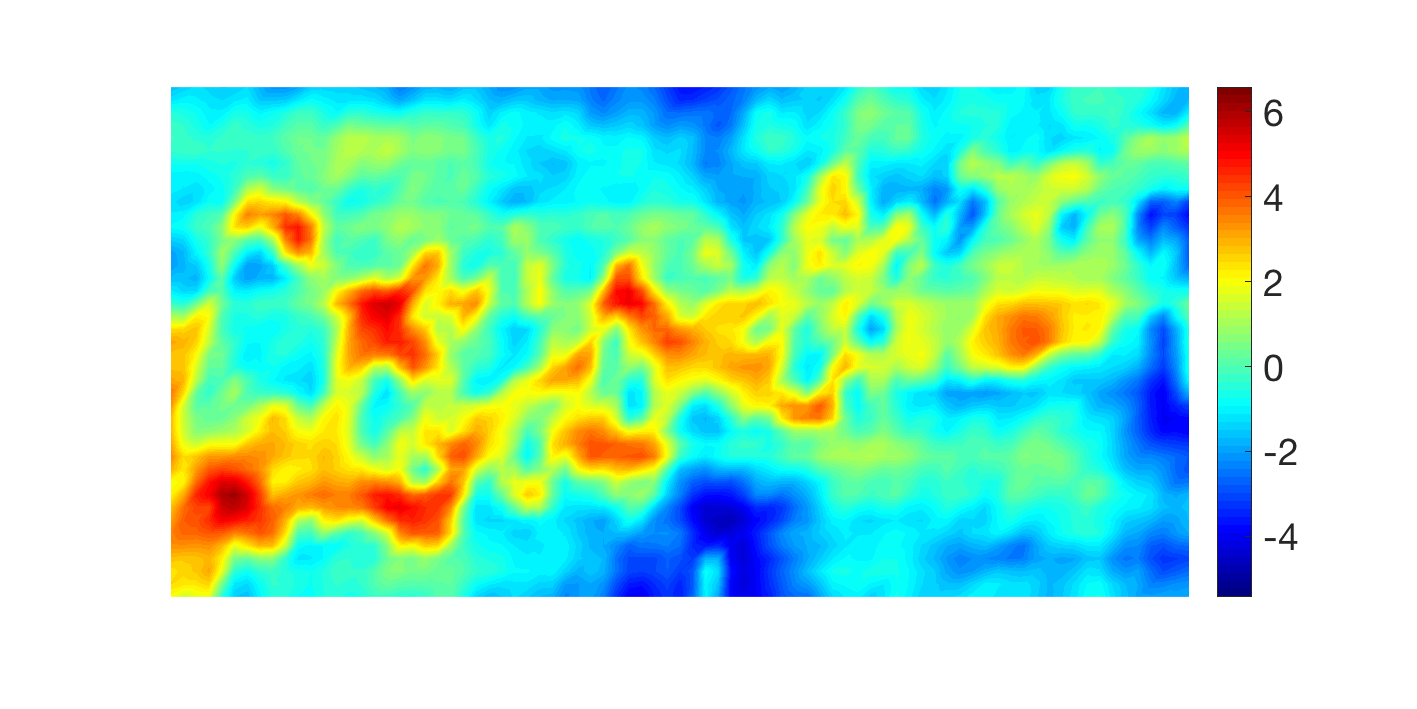}&
\includegraphics[width=.45\textwidth]{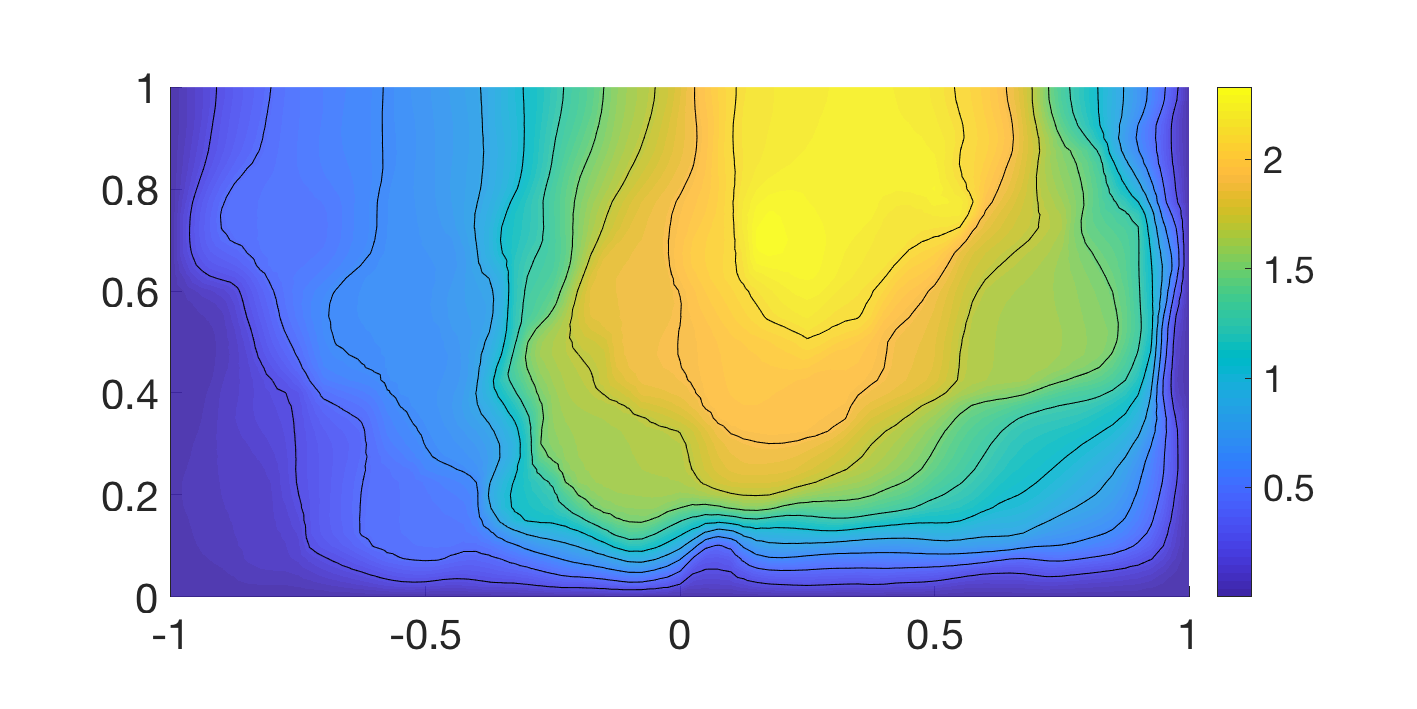}
\end{tabular}
\caption{Realizations of the log-permeability field (left), 
and the corresponding pressure fields (right).}
\label{fig:perm_samp}
\end{figure}

\subsubsection{The quantity of interest and its spectral representation}
\label{sec:qoi}
We consider the following quantity of interest:
\[
    f(x, \theta) := \restr{p(x, \theta)}{\Gamma_N}.
\]
A few realizations of $f(x, \theta)$ are plotted in
\cref{fig:spe_qoi}~(left).  To compute the KL expansion of $f$, we
use a sample average approximation of its covariance function, which is then
used to solve the discretized generalized eigenvalue problem for its KL modes.
The first 30 normalized eigenvalues of the covariance operator of $f$, which we denote by
$\Cq$, are plotted in \cref{fig:spe_qoi}~(middle, red color); 
these correspond to
computing the KL expansion of the QoI using sampling with a Monte Carlo (MC) sample of
size $\Ns = 1000$. We also plot the eigenvalues of the 
log-permeability field covariance
operator $\Cp$, in the same plot (blue color); note that 
the eigenvalues of $\Cq$ decay significantly faster than those of $\Cp$, as expected.
To assess the impact of the MC sample size on computation
of the dominant eigenvalues of $\Cq$, we report the normalized eigenvalues of $\Cq$
computed using successively larger sample sizes, in \cref{fig:spe_qoi}~(right). We observe that a sample 
of size $\mathcal{O}(100)$  
can be used for computing the dominant eigenvalues reliably.
\begin{figure}[ht!]\centering
\begin{tabular}{ccc}
\includegraphics[width=.32\textwidth]{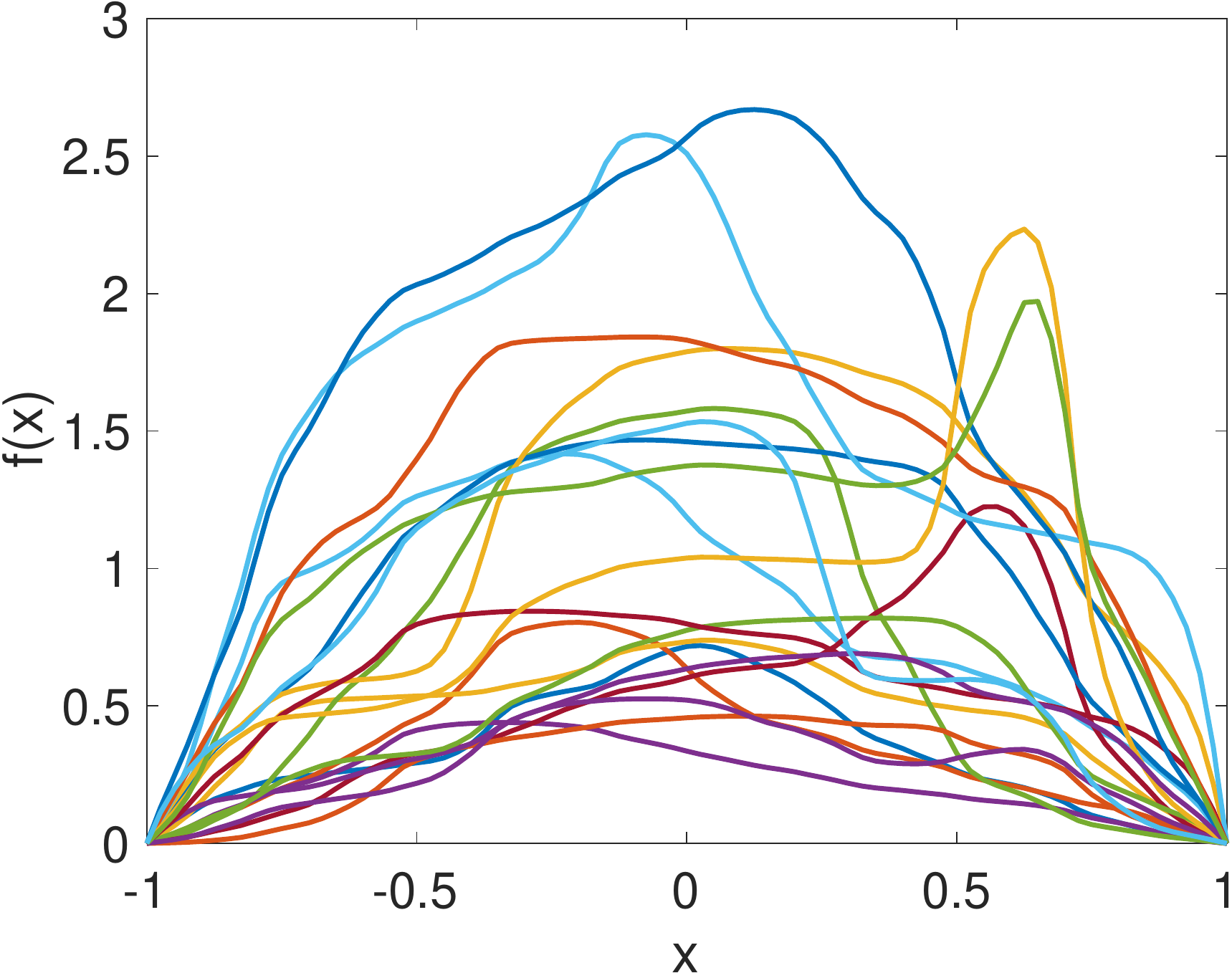}&
\includegraphics[width=.32\textwidth]{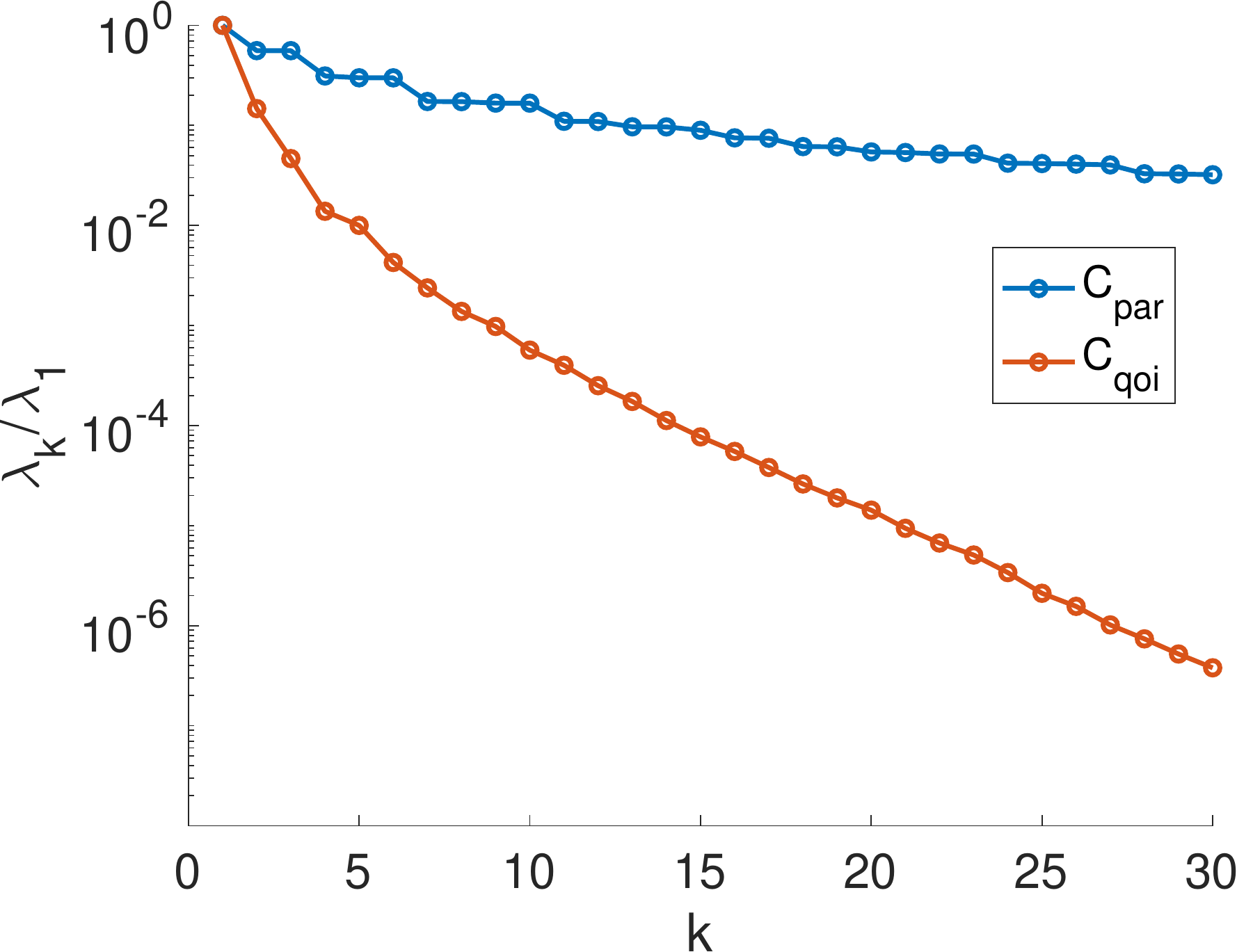}
\includegraphics[width=.32\textwidth]{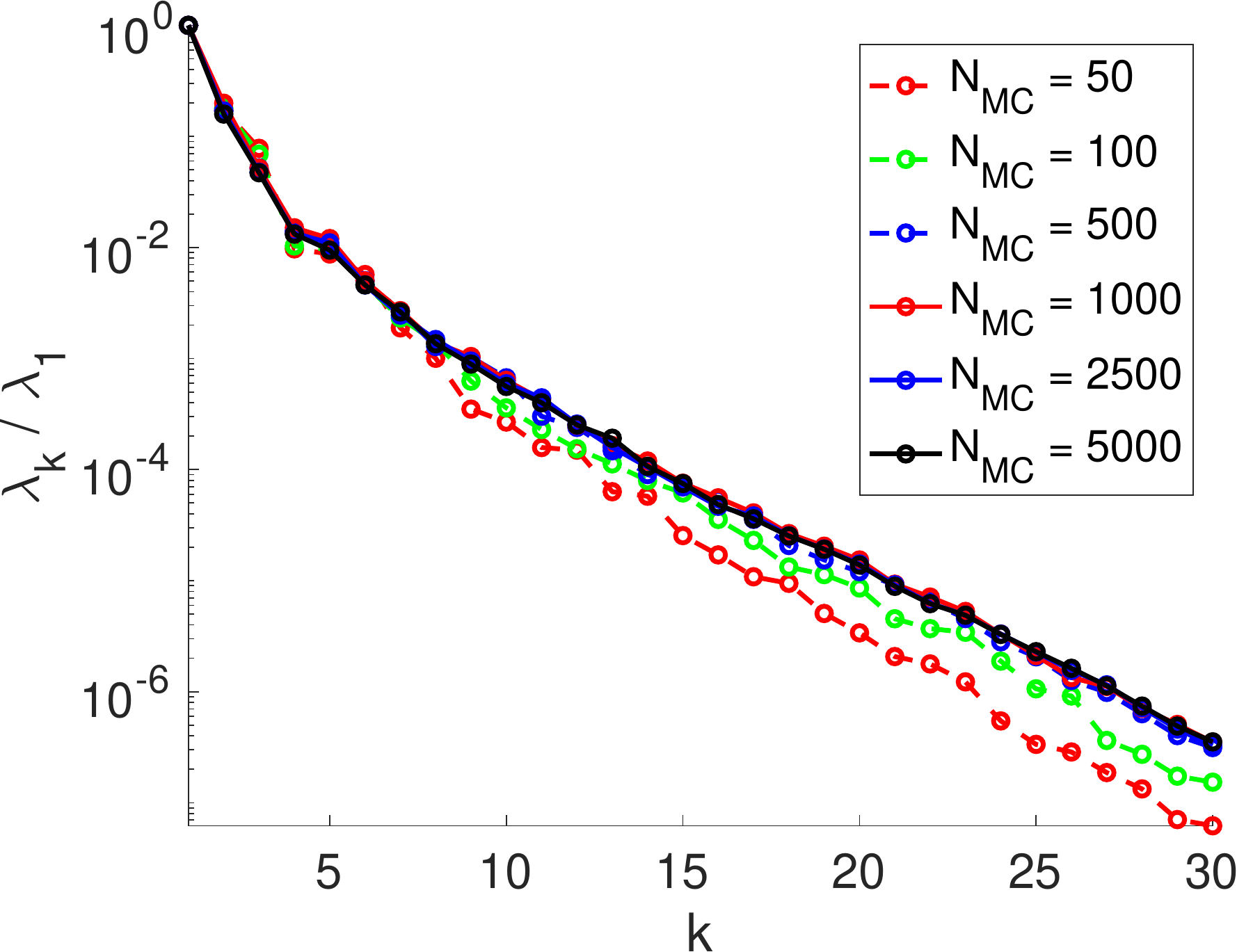}&
\end{tabular}
\caption{A few realizations of the QoI (left), 
eigenvalues of the output covariance operator 
versus those of the log-permeability field (middle).
Eigenvalues of the output covariance, with successively 
larger MC samples sizes for computing the output KLE (right).} 
\label{fig:spe_qoi}
\end{figure}

The fast decay of eigenvalues of $\Cq$ indicates the potential for
output dimension reduction. We note four orders of magnitude reduction in 
the size of the eigenvalues of $\Cq$ with only 15 modes in \cref{fig:spe_qoi}~(right). Hence, we consider a low-rank approximation 
of $f$, 
\begin{equation}\label{equ:trunc_f}
    f(x, \theta) \approx \hat{f}(x, \theta) 
    = \bar{f}(x) + \sum_{i=1}^\Nq \sqrt{\lambda_i(\Cq)} f_i(\theta) \phi_i(x) 
\end{equation}
with $\Nq = 15$. While this provides a low-rank approximation to $f$, the dimension of 
$\theta$ is still high, and is determined by the truncation of the
KL expansion of the 
log-permeability field at $\Np = 126$. Below, we use global sensitivity analysis to 
reduce the dimension of $\theta$.

\subsubsection{Derivative-based GSA}\label{sec:spe_gsa}
\begin{figure}[ht!]\centering
\includegraphics[width=0.4\textwidth]{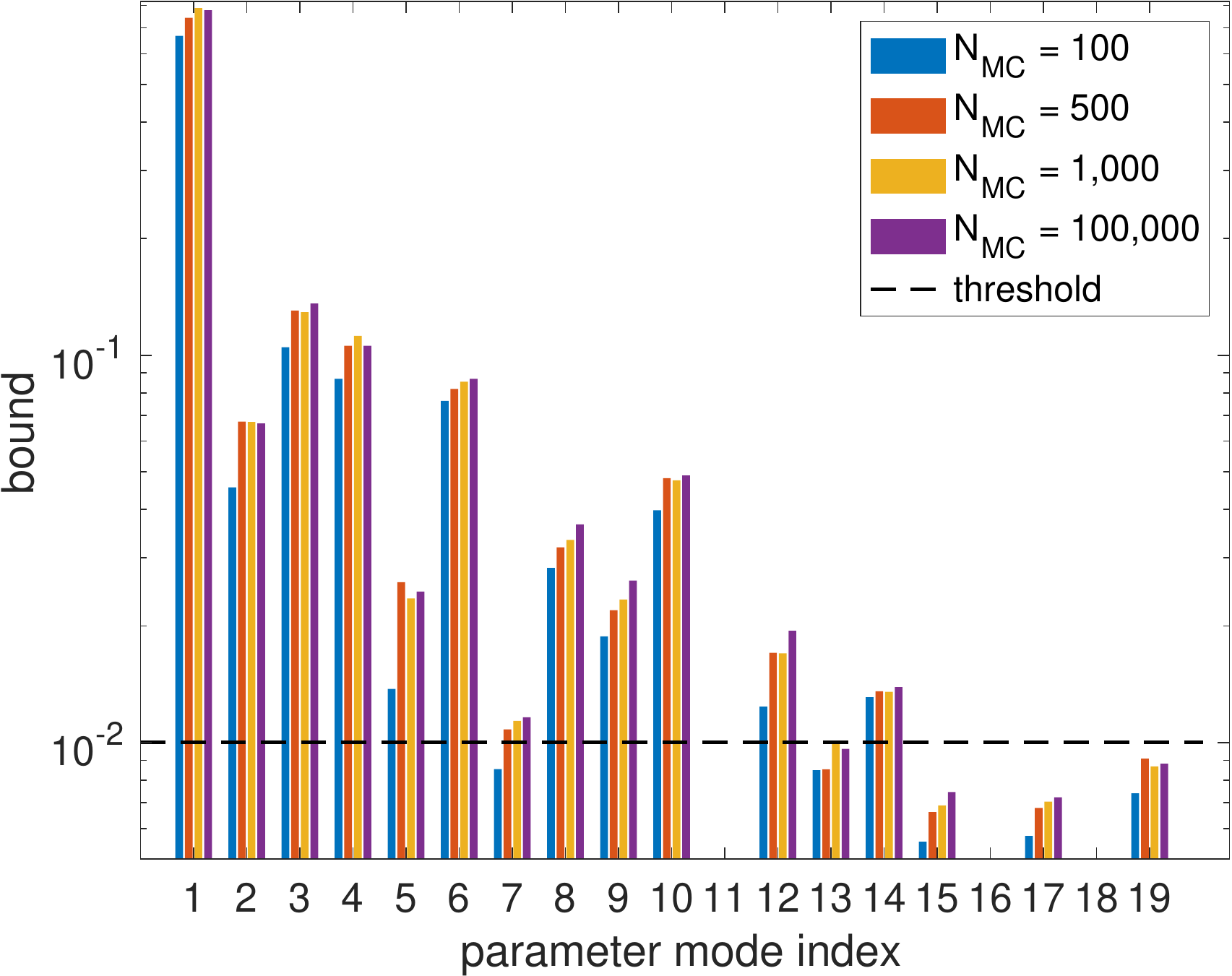}
\includegraphics[width=.4\textwidth]{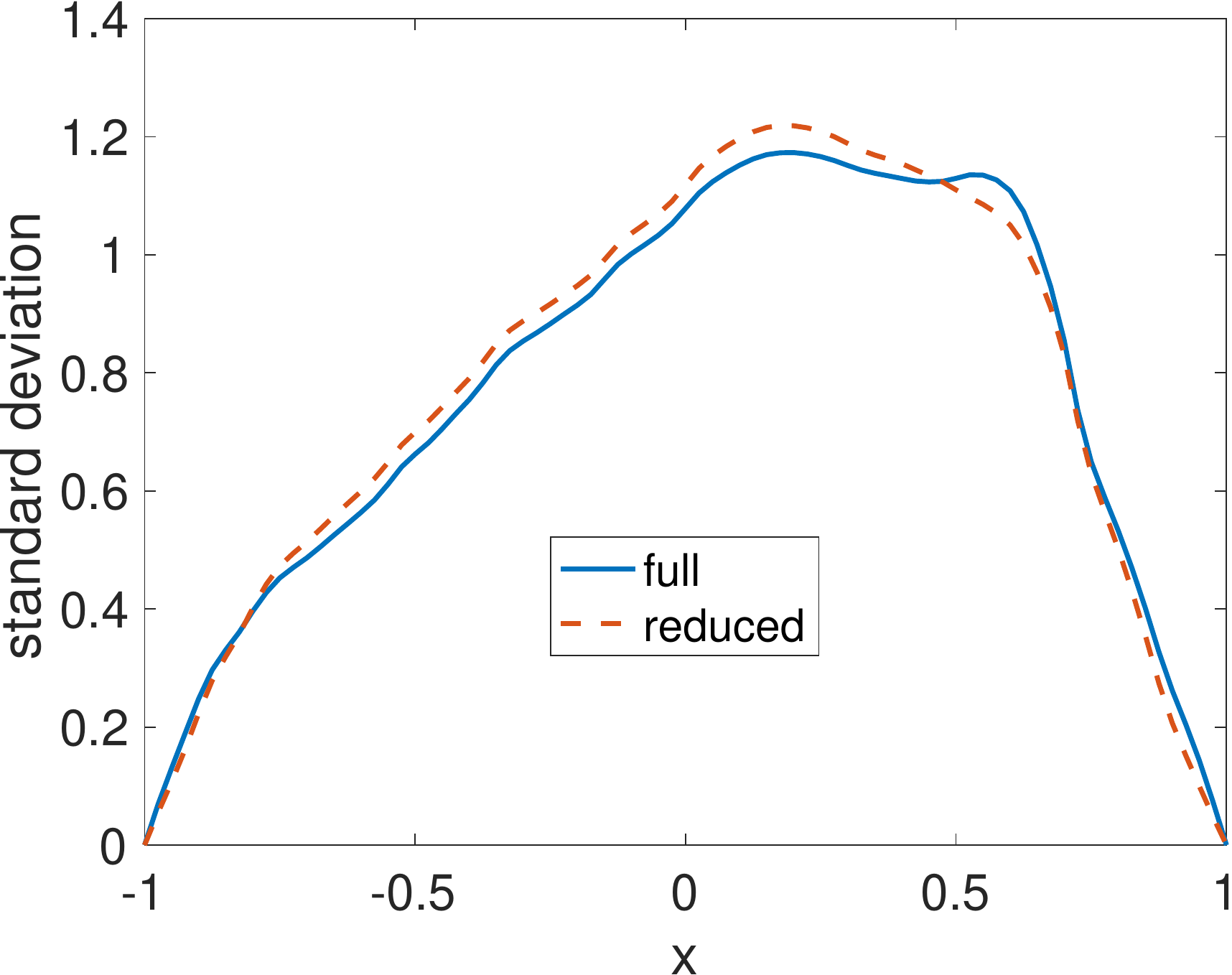} 
\caption{Left: 
DGSM-based bound in \cref{thm:bound_main} calculated for various sample sizes.
Right: standard deviation fields for full model versus 
that of the reduced model.} 
\label{fig:spe_dgsm}
\end{figure}
We begin by calculating the DGSM-based bounds on functional Sobol' indices from
\cref{thm:bound_main} for $\hat{f}$ defined in~\cref{equ:trunc_f}.
As seen before, this process requires sampling the QoI; we compute the
DGSM-based bounds by using MC samples of size $\Ns = 100, 500, 1{,}000$, and
$100{,}000$. The resulting bounds for the first $19$ parameters are reported in
\cref{fig:spe_dgsm}~(left). 

Note that \cref{fig:spe_dgsm}~(left) displays the bounds for only the first
$19$ modes, because the bounds for the remaining 107 modes were all well below the
chosen importance threshold of 0.01.  We note that the results calculated with $\Ns = 500, 1{,}000$, and $100{,}000$ provide a consistent
classification of important and unimportant parameters. This indicates that in
practice, a modest sample size is sufficient for obtaining informative
estimates of the DGSM-based bounds from \cref{thm:bound_main}.

\newcommand{\fred}{f^r}
The computed DGSM-based bounds indicate that the parameter KL modes $\theta_j$, with
$j \in \{1, 2, 3, 4, 5, 6, 7, 8, 9, 10, 12, 14\}$ were above the chosen importance threshold of 0.01 and
the remaining modes can be fixed at a nominal value of zero. This effectively
reduces the parameter dimension from $\Np = 126$ to $\Np = 12$. We denote the
resulting reduced model, now a function of only $12$ variables, by $\fred$.  To
test that $\fred$ reliably captures the variability of the true model $f$, we
sample both reduced and full models $10^5$ times to compare their statistical
properties. In \cref{fig:spe_dgsm}~(right), we
compare the standard deviation of the full and reduced models over 
the spatial domain $\X=[-1, 1]$ of the QoI.  
In \cref{fig:kde_spe} we report PDFs of $f(x, \cdot)$ and $\fred(x,
\cdot)$, at $x = -0.75,~ -0.25,~ 0.25,~ 0.75$. We note that the reduced model
captures the distribution of the QoI at the considered points closely.

\begin{figure}[ht!]\centering
\begin{tabular}{cc}
\includegraphics[width=.35\textwidth]{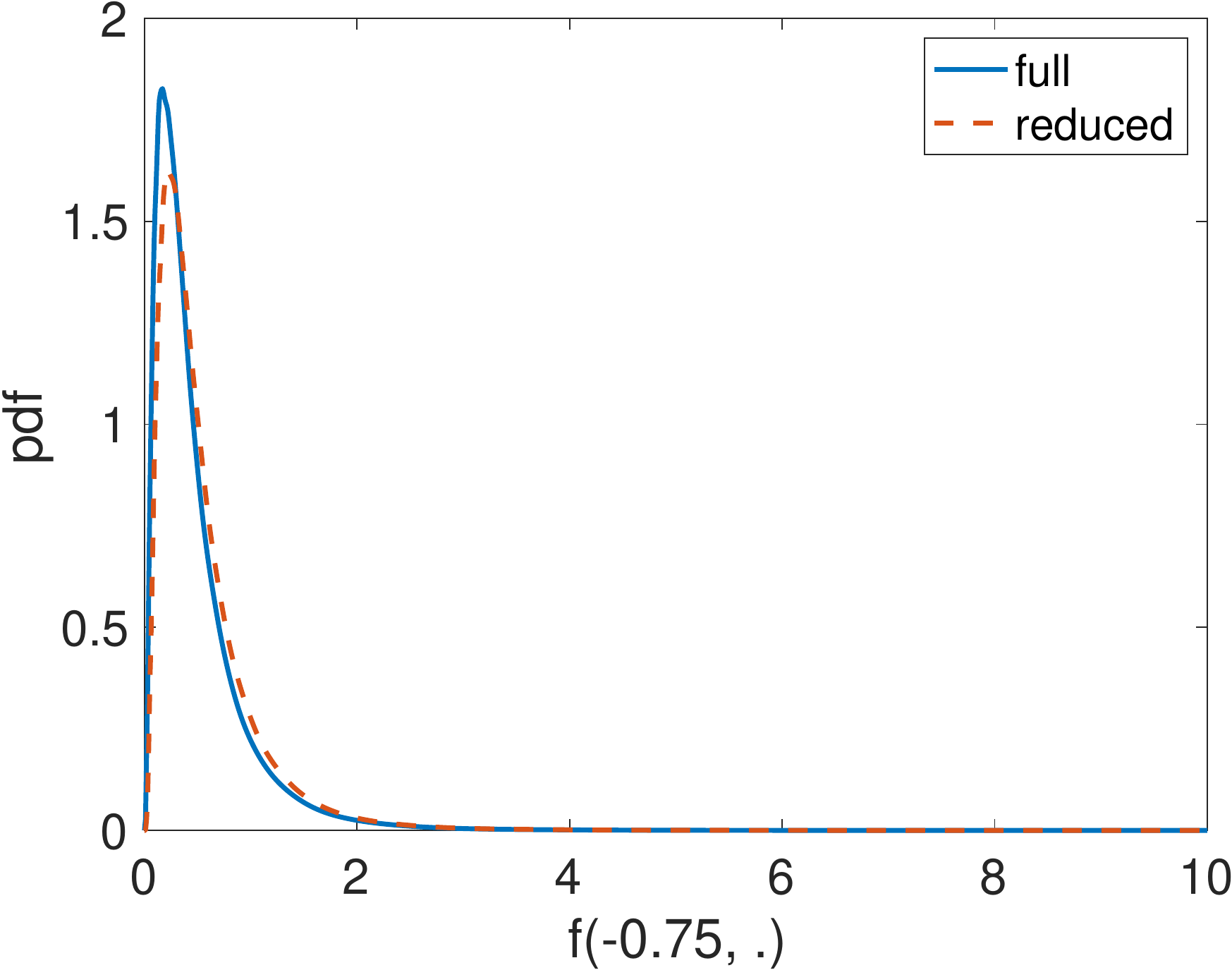}&
\includegraphics[width=.35\textwidth]{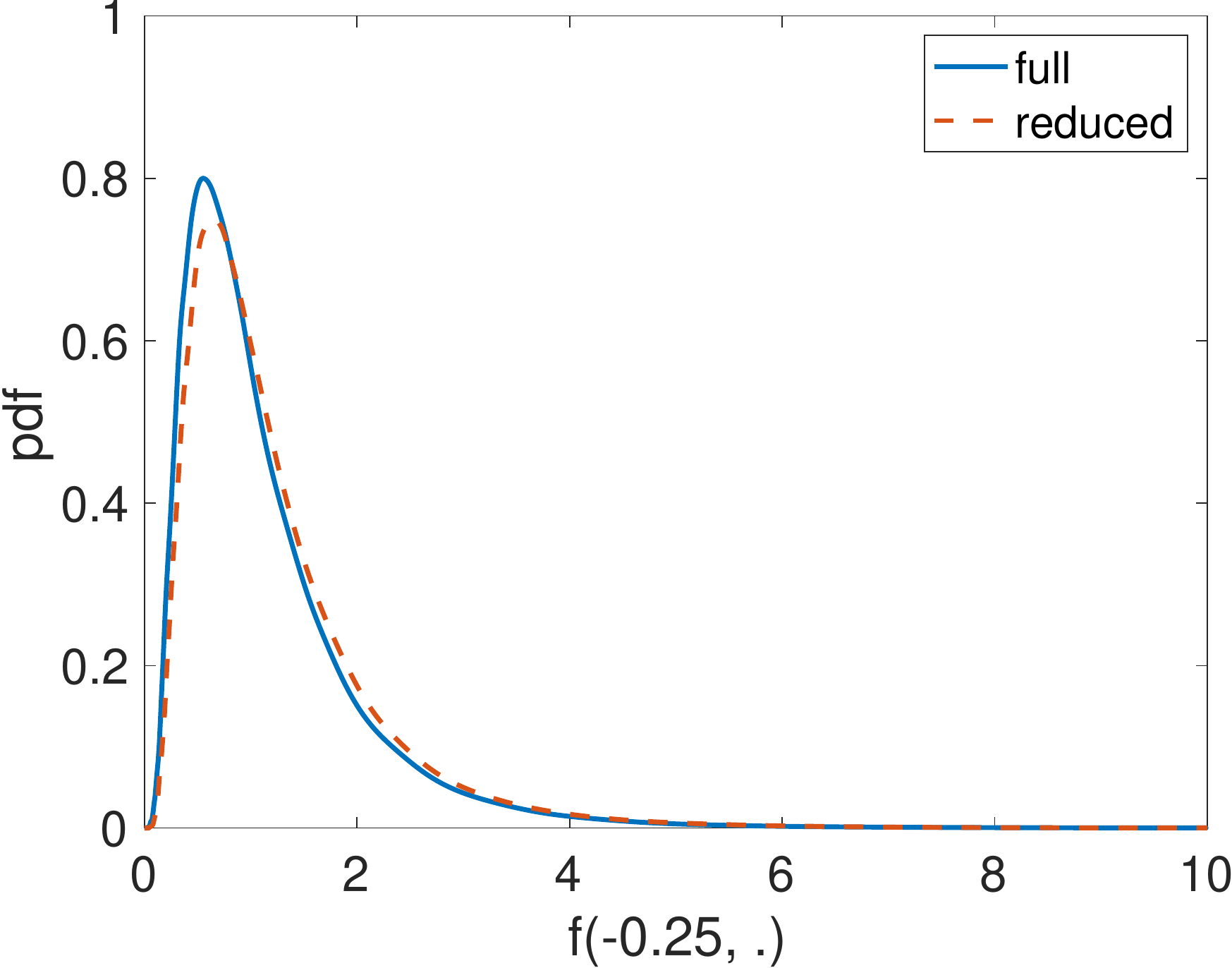}\\
\includegraphics[width=.35\textwidth]{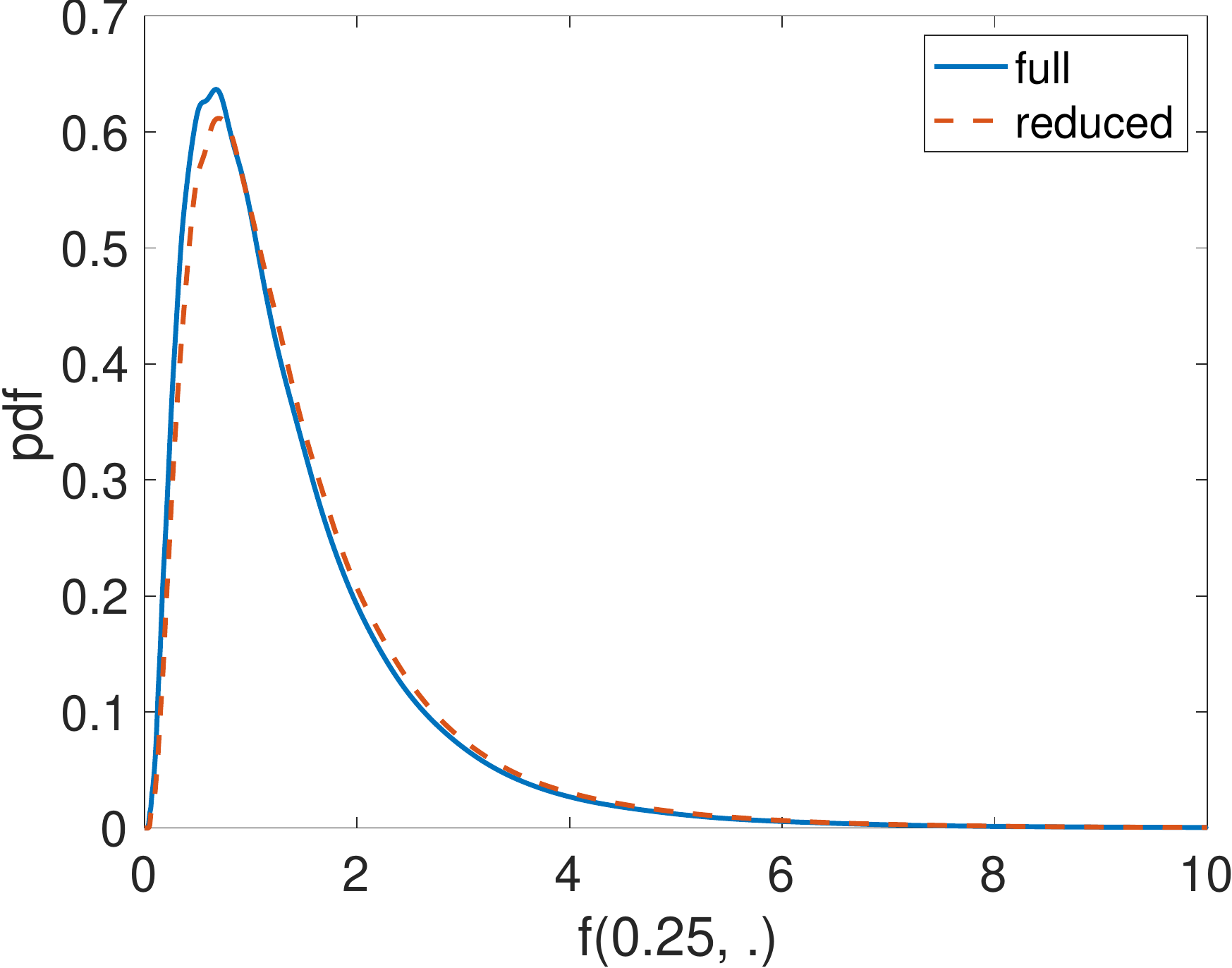}&
\includegraphics[width=.35\textwidth]{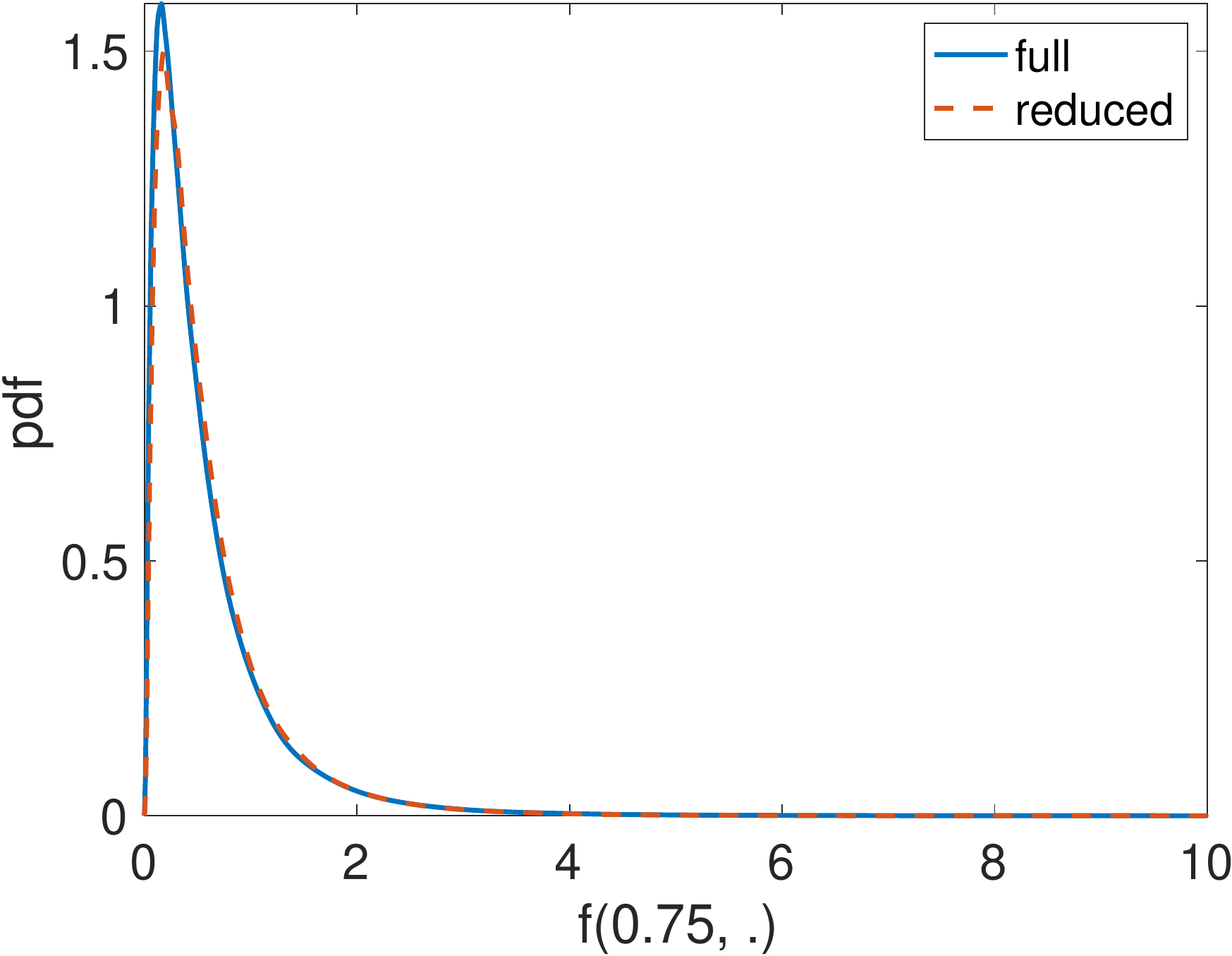}
\end{tabular}
\caption{pdf estimate for equally spaced points, $[-1,1]$} 
\label{fig:kde_spe}
\end{figure}

\subsection{Application to biotransport in tumors}\label{sec:biotransport}
In this section, we apply our derivative-based GSA  methods to a biotransport
problem.  Specifically, we consider biotransport in cancerous tumors with
uncertain material properties.  We focus on the resulting uncertainties in the pressure field
in a spherical tumor when a single needle injection occurs at the center of the
tumor.

\subsubsection{Model description}
Restricting our attention to a 2D cross-section, we consider Darcy's
law constrained by conservation of mass in a 2D physical domain $\D \subset
\R^2$ given by a circle of radius $R_\text{tumor} = 5$ mm with an inner circle
of radius $R_\text{needle} = 0.25$ mm, modeling the injection site, removed;
see \cref{fig:dom}. The inner and outer boundaries of the physical domain
$\D$ are denoted by $\Gamma_\text{N}$ and $\Gamma_\text{D}$, respectively. 
\begin{figure*}[ht!]
\begin{minipage}{.58\linewidth}
  \captionof{table}{Model parameters for the biotransport problem.}
  \begin{tabular}{lll}
  \toprule
  Parameter & Symbol & Nominal Value [unit]\\
  \bottomrule
  Permeability   &  $\kappa$           & $0.5~[md]$\\ 
  Viscosity & $\eta$   & $8.9 \times 10^{-4}~[Pa \cdot s]$\\ 
  Inflow rate & $Q$               & $1~[mm^2/min]$\\ 
  \end{tabular}
\label{tbl:biotransport_params}
  \end{minipage}%
  \begin{minipage}{.42\linewidth}
    \centering\includegraphics[width=.95\linewidth]{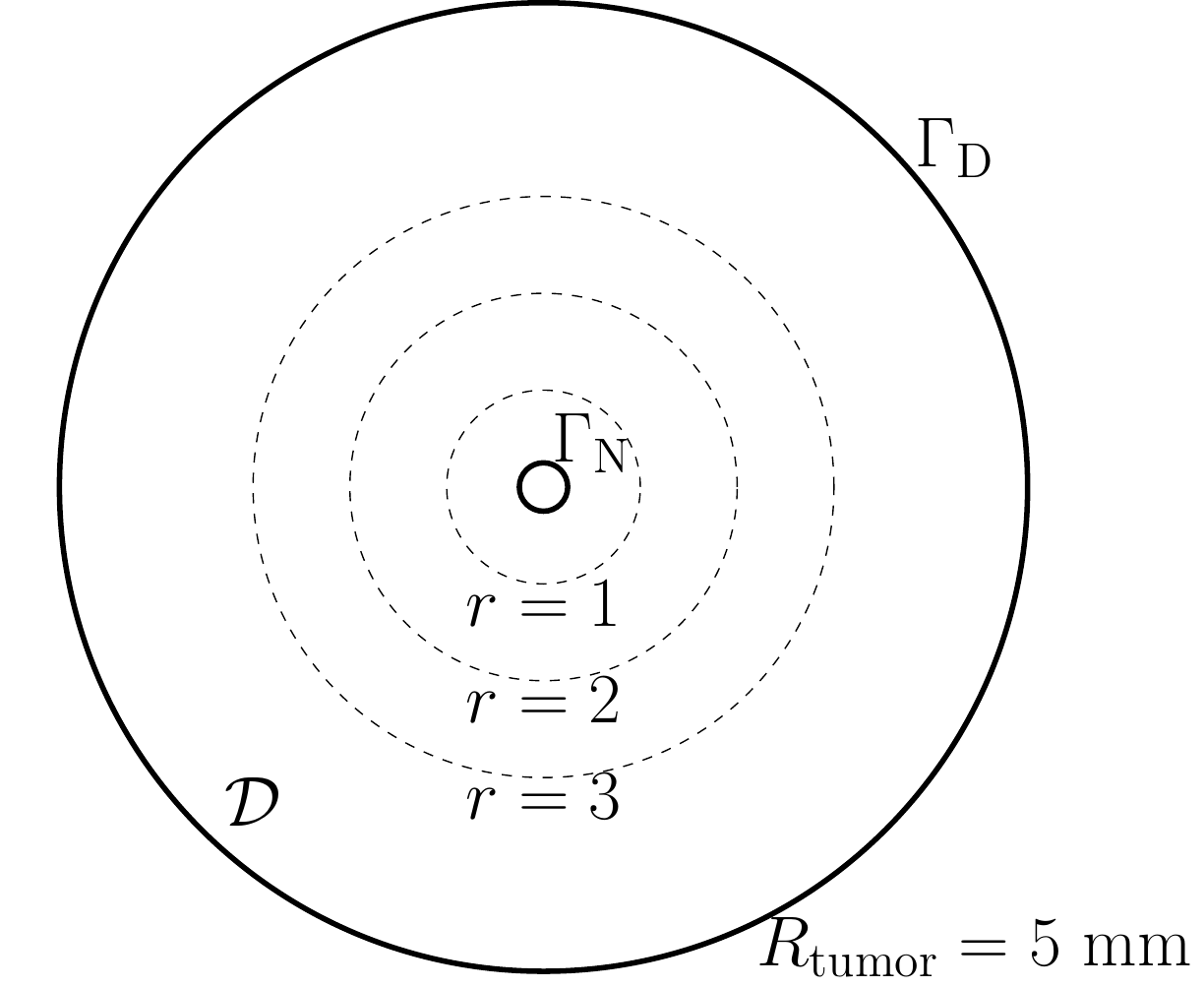}
    \captionof{figure}{The domain $\D$. The inner and 
    outer boundaries 
    are equipped with Neumann and Dirichlet boundary conditions
    and are denoted by $\Gamma_\text{N}$ and $\Gamma_\text{D}$, respectively.}
    \label{fig:dom}
\end{minipage}
\end{figure*}%

The fluid pressure $p$ is governed by the following elliptic PDE:
\begin{equation} \label{equ:2D_Darcy}
\begin{aligned}
-\nabla \cdot \left( \frac{\kappa}{\eta} \nabla p\right) &= 0 \quad \text{in } \D,\\ 
p &= 0 \quad \text{on } \Gamma_\text{D}, \\
\nabla p \cdot n &= \frac{Q \eta}{2\pi R_\text{needle} \kappa} 
\quad \text{on } \Gamma_\text{N}.
\end{aligned}
\end{equation}
Here $\kappa$ is the absolute permeability field, $\eta$ is the fluid dynamic
viscosity, $Q$ represents the volume flow rate per unit length, and ${n}$ is
the outward-pointing normal of the inner boundary $\Gamma_\text{N}$.  The
nominal values for the parameters in \cref{equ:2D_Darcy} are given in
\cref{tbl:biotransport_params}. These values are selected
according to those used in previous experimental and numerical studies of fluid
transport in tumors~\cite{maher:08,ma:12TF,Chen:07}. As has been discussed by
many researchers, tumor structure can be highly complicated due to its invasive
nature. In general, a tumor consists of loosely organized abnormal cells,
fibers, vasculature, and lymphatics~\cite{Clark:91}. This results in randomly
formed tumor tissues with structural heterogeneity.   

In this subsection, the permeability field is modeled as a log-Gaussian
random field as follows. 
Let $z(x, \omega)$ be a centered Gaussian process with 
the following covariance function:
\begin{equation}
c_z({x},
{y}) = \exp\left\{- \frac{1}{\ell} \| x-y\|_1\right\}, \quad x, y \in \D,
\end{equation}
where $\ell > 0$ is the correlation length.
Then, we define the log-permeability $a = \log\kappa$ as in \cref{equ:param_process},
where the pointwise mean and variance of the process are given 
by 
$\bar a \equiv \ln(0.5) + \sigma^2_a$ and
$\sigma^2_a = 0.25$, respectively. Note that $\bar a$ is selected to ensure that the
mode of the $\kappa$ distribution at each spatial point is $0.5~md$, which is the nominal value for $\kappa$ given in \cref{tbl:biotransport_params}.
We can represent $a(x,\omega)$ using a truncated KL expansion as in 
\cref{equ:truncated_KLE}. 

\subsubsection{The quantity of interest and its spectral representation}
\label{sec:qoi_Bio}
We consider the following QoI: 
\begin{equation}
f(x, \theta) =  
\Q p, 
\end{equation} 
where, as in \cref{sec:DGSM_PDE}, $\Q$ is the restriction operator to a
closed subset $\X$ of $\D$.  In this case, $\X$ is an annulus with the inner
boundary given by the inner boundary $\Gamma_\text{N}$ of $\D$ and with the
outer boundary having a radius $R_{out} = 1~mm$, $2~mm$, or $3~mm$ (see 
\cref{fig:dom}). 
The corresponding truncated KL expansion of $f$ reads 
\begin{equation} \label{equ:trunc_p}
\hat{f}(x, {\theta}) := \bar{f}(x) + \sum_{k = 1}^{\Nq} 
\sqrt{\lambda_k(\Cq)} f_k({\theta}) \phi_k(x),
\end{equation}
where the KL modes $f_i$ are defined as before, and
$\lambda_k(\Cq)$ and $\phi_k(x)$ are the 
eigenpairs of the QoI covariance operator $\Cq$. 

\subsubsection{Derivative-based GSA}\label{sec:spe_gsa_Bio}
As in \cref{sec:spe_gsa}, we calculate the DGSM-based bounds on
functional Sobol' indices from \cref{thm:bound_main} for the QoI defined
in~\cref{equ:trunc_p} and follow the adjoint-based framework outlined in
\cref{sec:DGSM_PDE}. 
As mentioned previously, a small DGSM-based bound for a given parameter implies
that the corresponding functional total Sobol' index is small and
thus, the parameter is deemed unimportant.  In the experiments in this section,
we set an \emph{importance threshold} of $0.025$. 
In \cref{fig:GSA_Test_Bio}, we study the effects of the MC
sampling size $\Ns$, the KL expansion dimension $\Np$ of the input and $\Nq$ of
the output, annulus size (i.e., size of $\X$), and correlation length $\ell$ on
DGSM-based bounds.  
Note that \cref{fig:GSA_Test_Bio} displays the DGSM-based bounds for the first $37$
modes, beyond which the DGSM-based bounds were all below the chosen importance threshold. 
Below, we explain the numerical experiments 
reported in \cref{fig:GSA_Test_Bio}, in detail.

In the first test, we examine the effect of the MC sample size $\Ns$ 
as needed in our approach for computing DGSMs (cf.\ \cref{alg:DGSM_KLE}).
Similar to the
observation in \cref{sec:spe_gsa}, a modest sample size is sufficient
for obtaining informative estimates of the DGSMs. Specifically, we 
present one set of test
results in \cref{fig:GSA_Test_Bio} (top left). Here, the outer
radius of the annulus is $1~mm$, the correlation length is $0.5~mm$, 
and we consider an input
dimension of $\Np = 150$, and an output dimension of $\Nq = 50$. 
We observe that a
sample size of $\Ns~=~750$ is sufficient for obtaining 
a reliable estimation of DGSMs-based bounds.
Therefore, the MC sample size in the following tests is fixed at $\Ns~=~750$.

We then test the effects of the annulus size and correlation length on DGSMs. In
these tests, the input and output dimensions are $\Np = 150$ and
$\Nq = 50$, respectively. 
From \cref{fig:GSA_Test_Bio} (top right and bottom left), 
we observe that when the annulus size increases or the correlation
length decreases, the QoI is sensitive to more KL terms of the input.
Interestingly, most of these sensitive parameters are from 
relatively high-order terms. For example, as shown in \cref{fig:GSA_Test_Bio} (bottom left), when the correlation length decreases 
from $2.0~mm$ to $0.5~mm$, the importance of KL
modes $\theta_j$, with $j \in \{9, 10, 17, 22, 23\}$ gradually grow. 
Implication of such issues on reduced-order modeling (ROM) will be 
discussed in the next section. 
Next, we examine the impact of increasing $\Np$ 
and $\Nq$. As seen in 
\cref{fig:GSA_Test_Bio}~(bottom right), 
increasing the input and output dimensions beyond the selected 
values of $\Np = 150$
and $\Nq = 50$ does not result in noticeable changes in DGSM estimates.

\begin{figure}[ht!]\centering
\begin{tabular}{cc}
\hspace{-0.6cm}
\includegraphics[width=.5\textwidth]{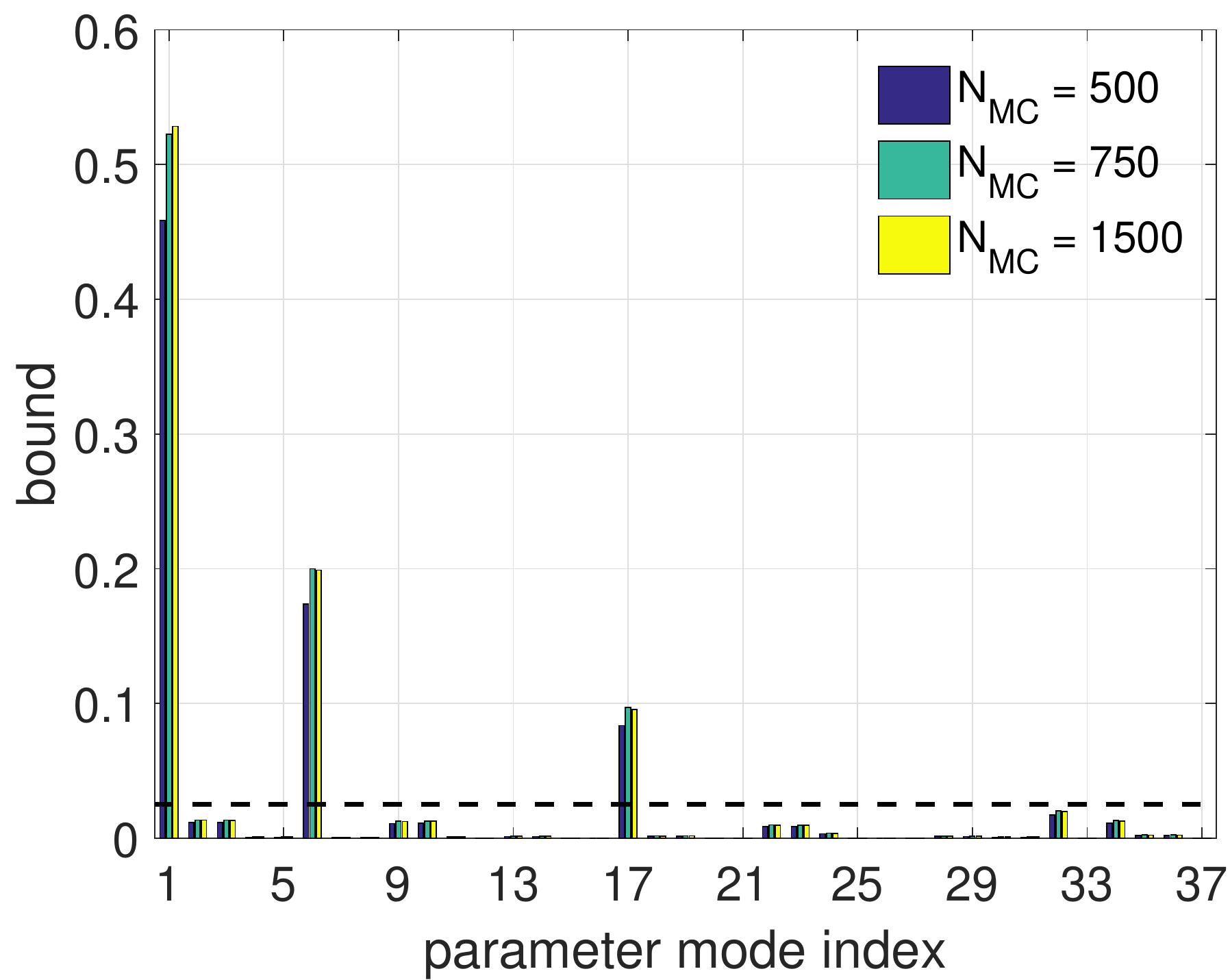}
\includegraphics[width=.5\textwidth]{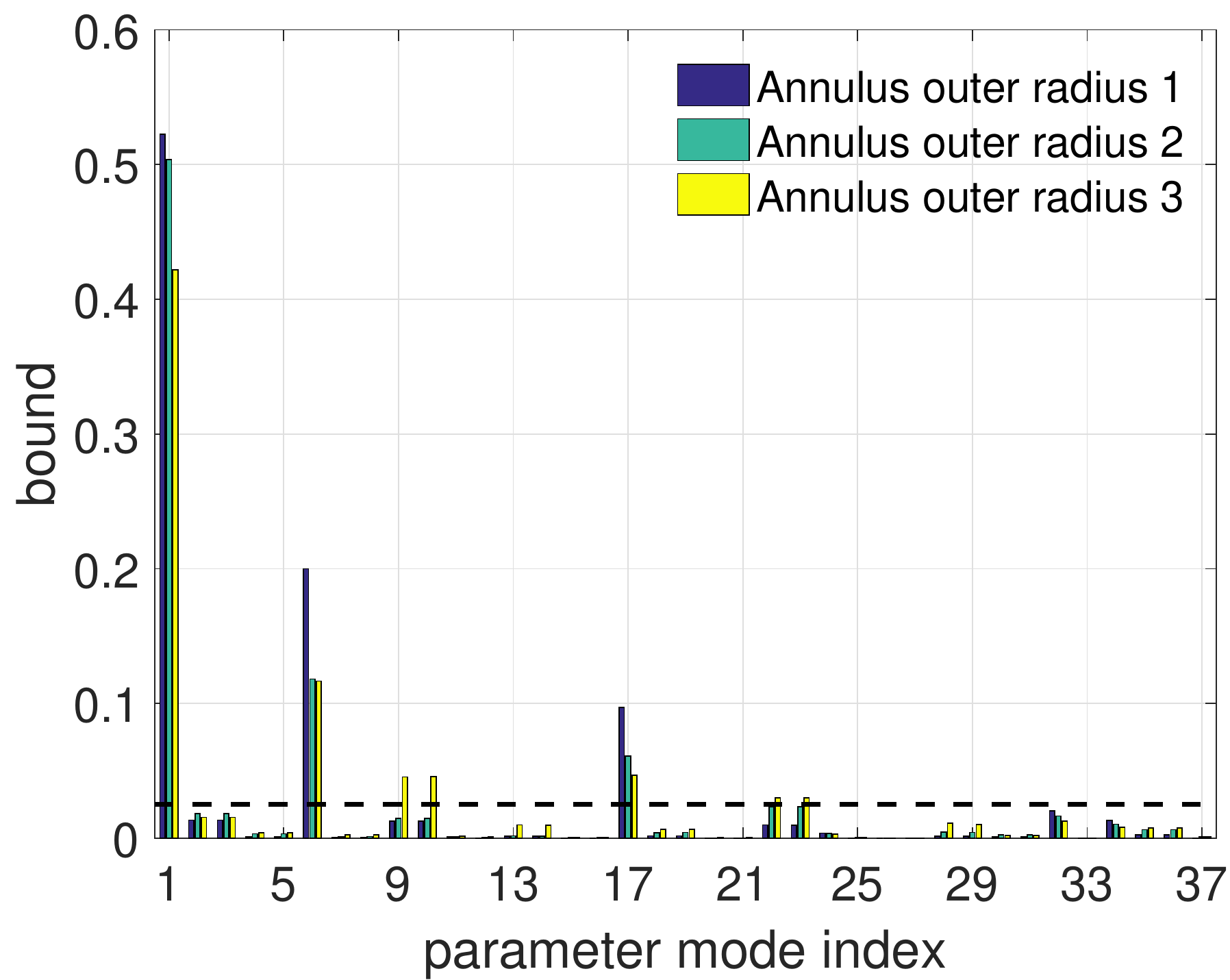}\\
\hspace{-0.6cm}
\includegraphics[width=.5\textwidth]{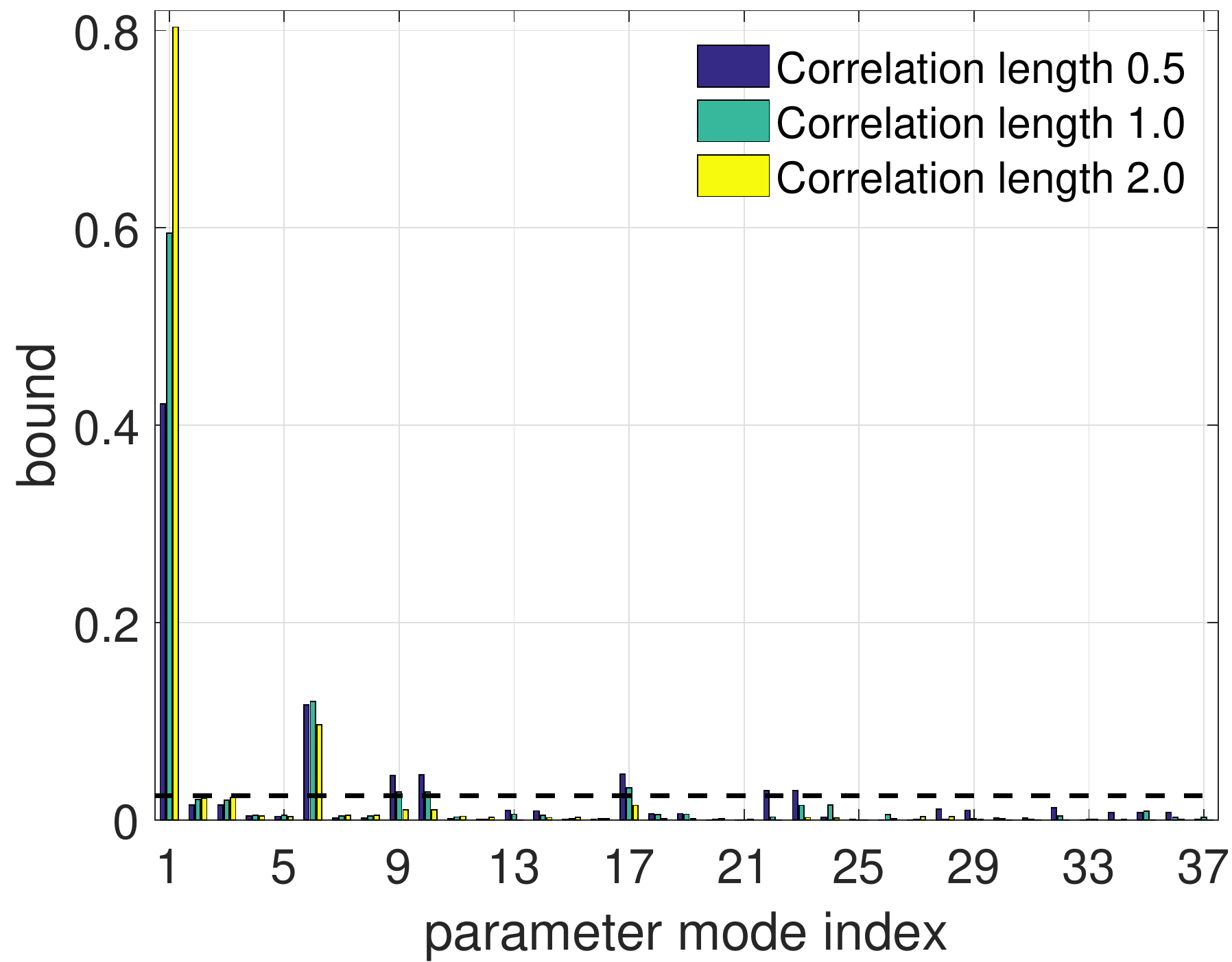}
\includegraphics[width=0.5\textwidth]{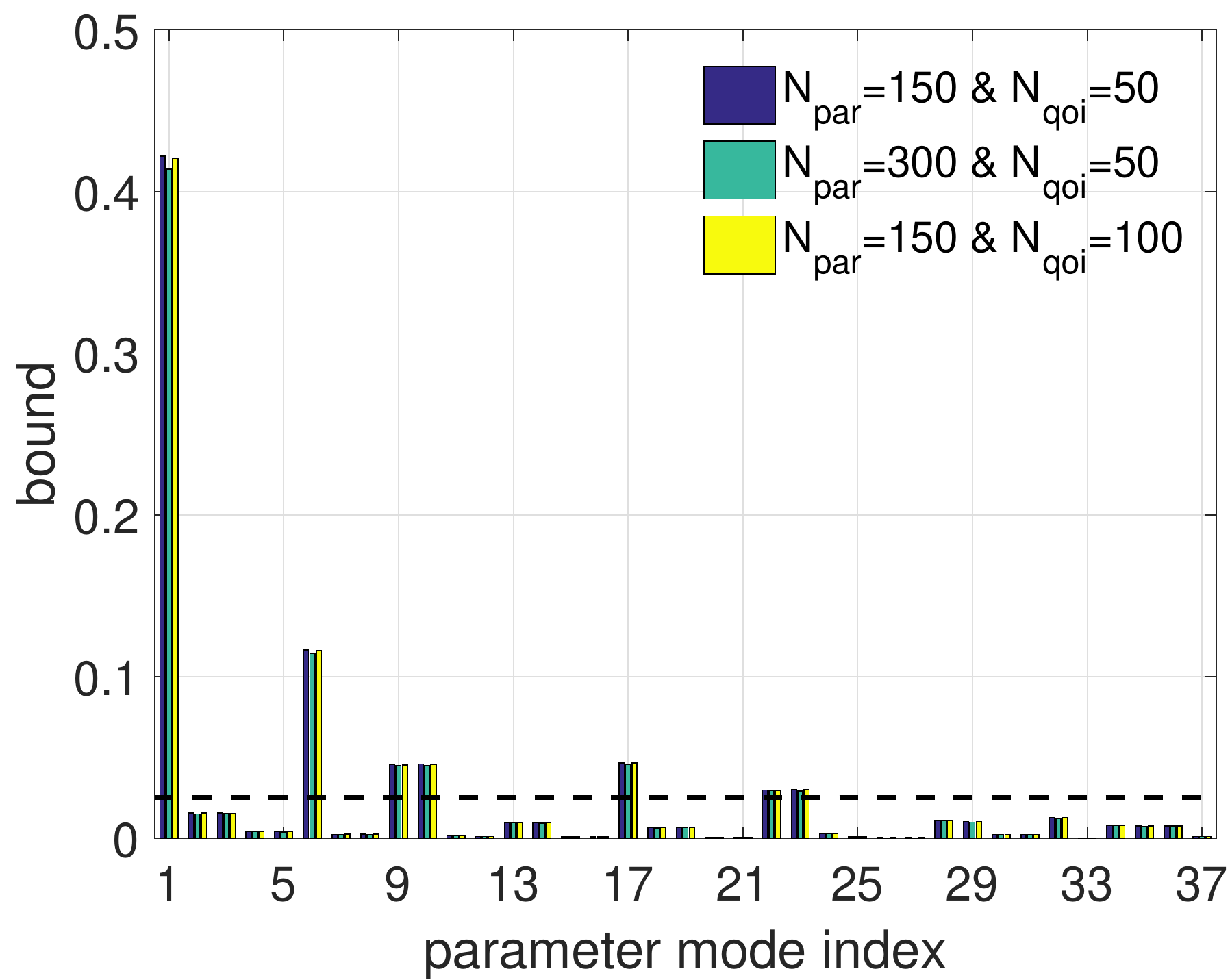}\\
\end{tabular}
\caption{The functional 
DGSM-based bounds of pressure fields in a tumor with uncertain permeability. Top left: Convergence study with the MC sample size $N_{MC}~=~500,~ 750,$ and $1500$. Top right: Comparison of DGSM-based bounds for different annulus sizes, namely the annulus outer radii of $1~mm$, $2~mm$, and $3~mm$. Bottom left: Comparison of DGSM-based bounds for different correlation lengths, namely $0.5~mm$, $1~mm$, and $2~mm$. Bottom right: DGSM-based bounds calculated with different combinations of the KL expansion dimensions of the input and output.} 
\label{fig:GSA_Test_Bio}
\end{figure}

\subsubsection[]{Insights on ROM assisted by DGSMs}\label{sec:ROM}
From the global sensitivity analysis, we find that the QoI is only sensitive to
several selected KL terms of the input. This can be used to guide 
ROM based on DGSMs.
In this section, we compare two ROM approaches: one is based on the GSA with DGSMs (termed as DGSM-based ROM) and the other is
based on directly selecting the first $k$-terms of the KL expansion of the random
input field (termed as KL-based ROM). Generally, the reduced-order model of the
input can be written as follows:
\begin{equation}\label{equ:a_rom}
    \tilde{a}(x, \omega) = 
    \bar a(x) + \sum_{k \in \mathcal{S}} \sqrt{\lambda_k(\Cp)} 
    \theta_k(\omega) e_k(x), 
\end{equation}
where $\mathcal{S}$ is the set which consists of the indices of the KL terms used in ROM.
We evaluate the
performance of the two ROM methods on recovering the PDFs of pressures at
different locations in the flow field.

As shown in \cref{fig:Points_Dis}, we select three points on the mesh
with different distances from the center of the domain:
the point
$P_1$ is on the inner boundary with a large relative standard deviation (RSD) 
of the pressure ($RSD = 0.143$); the point $P_2$ is close 
to the inner boundary  with a moderate RSD ($RSD = 0.105$); and 
the point $P_3$ is far from the inner boundary  with 
a relatively small RSD ($RSD = 0.0845$).
In the DGSM-based ROM, the first $n$ KL terms which the QoI is most sensitive
to are used to reconstruct the reduced-order model of the pressure field. In
the KL-based ROM, the first $n$ KL terms, corresponding to the $n$ largest
eigenvalues of the input covariance operator, are used to reconstruct the
reduced-order model. An MC sampling approach is used to construct PDFs
from the full model, which includes all the KL terms, and those from the
reduced-order models with different fidelities. The case with a small
correlation length ($\ell = 0.5~mm$) and a large annulus size 
($R_{out} = 3~mm$) is studied here. An MC sample of size $6000$ was found sufficient
for constructing the PDFs.

From \cref{fig:GSA_PDF_Bio}, we observe that at $P_1$, where the pressure
variance is large, the reduced-order model with only the first seven most sensitive
KL terms can nearly recover the PDF of the full model. Its performance is
comparable to that of the KL-based ROM with the first 30 KL terms. This is not
a surprise, because, as seen from the last figure in
\cref{fig:GSA_Test_Bio}, the first seven most sensitive KL terms
$\theta_j$, with $j \in \{1, 6, 9, 10, 17, 22, 23\}$, are within the first 30
KL terms used in the KL-based ROM.  Similar conclusions can be drawn at $P_2$
where a moderate pressure variance is observed. At $P_3$, we find that the
DGSM-based ROM with the first seven most sensitive KL terms does not recover the
PDF well; however, the PDFs obtained using the DGSM-based ROMs with more KL
terms, such as that with the first 15, 30 and 45 most sensitive KL terms,
gradually approach the PDF of the full model. On the other hand, with the same
number of KL terms, the KL-based ROM makes very slow progress towards the full
model PDF.  All these observations indicate that the DGSM-based ROM can be a
much more efficient reduced-order modeling approach than the KL-based ROM that
involves a priori truncation of the input field KL terms. 

\begin{figure}\centering
\includegraphics[width=.5\textwidth]{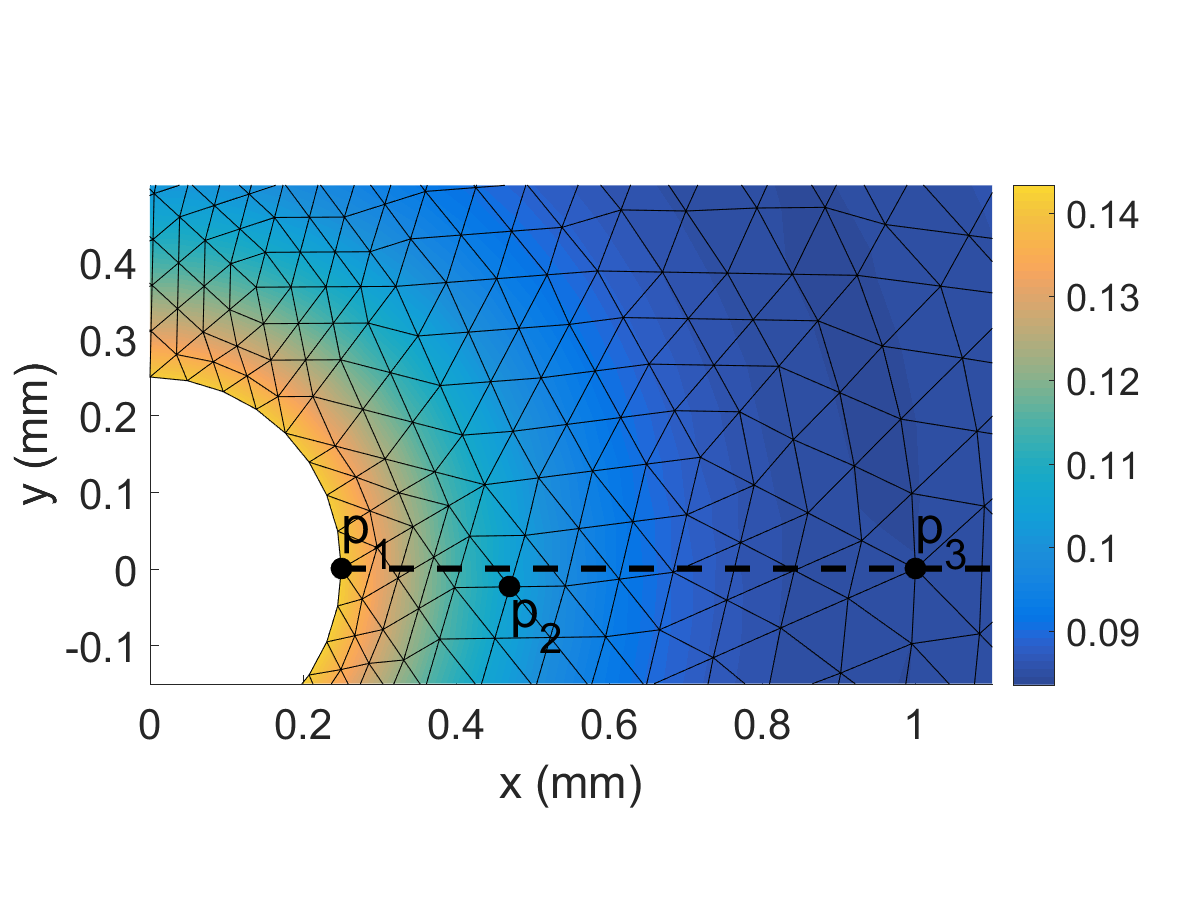}
\caption{Distribution of the points where PDFs of pressures are extracted, and the corresponding RSD field (contour).}
\label{fig:Points_Dis}
\end{figure}
\begin{figure}[ht!]\centering
\begin{tabular}{cc}
\hspace{-1cm}
\includegraphics[width=0.53\textwidth]{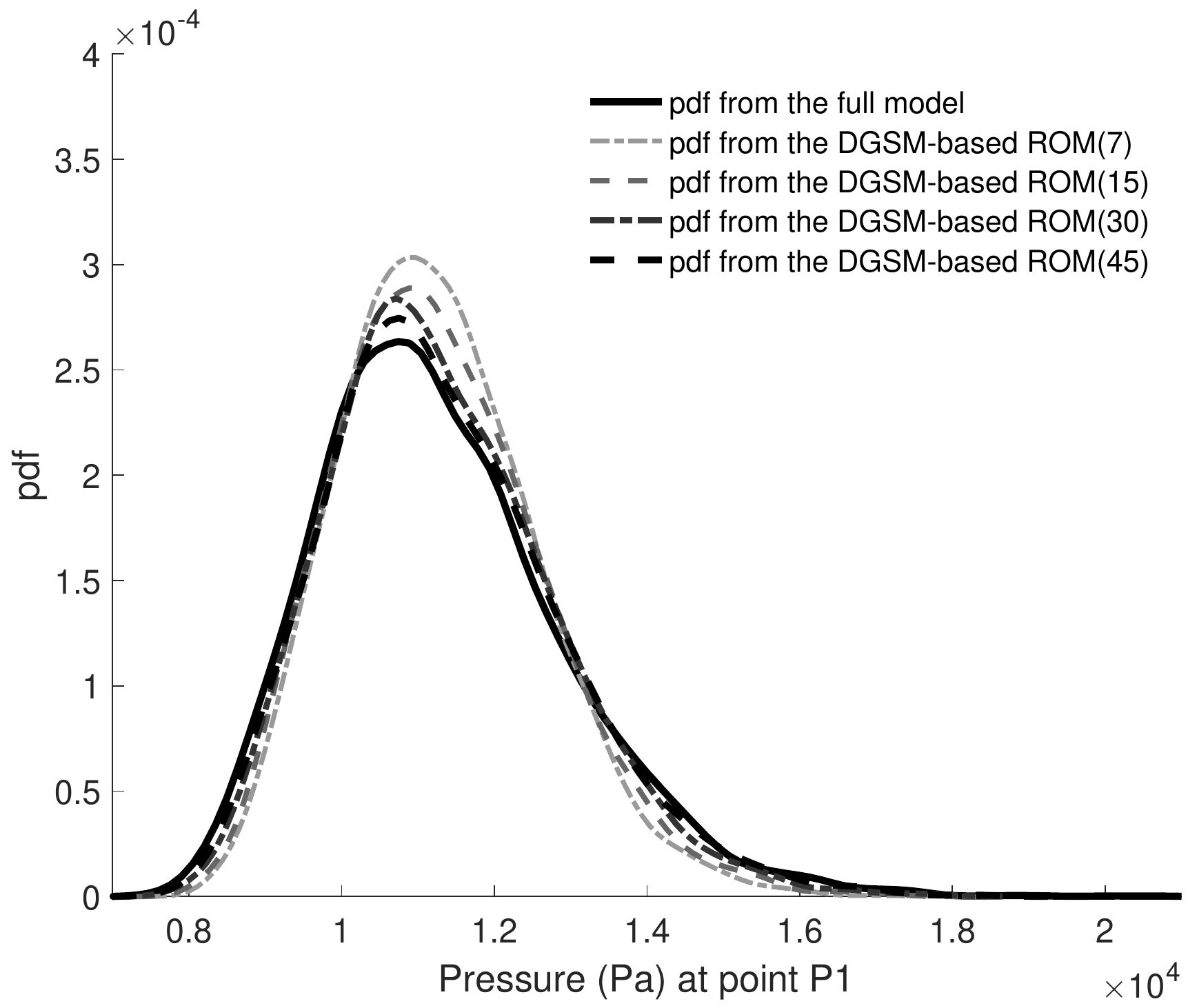}&
\hspace{-0.1cm}
\includegraphics[width=0.53\textwidth]{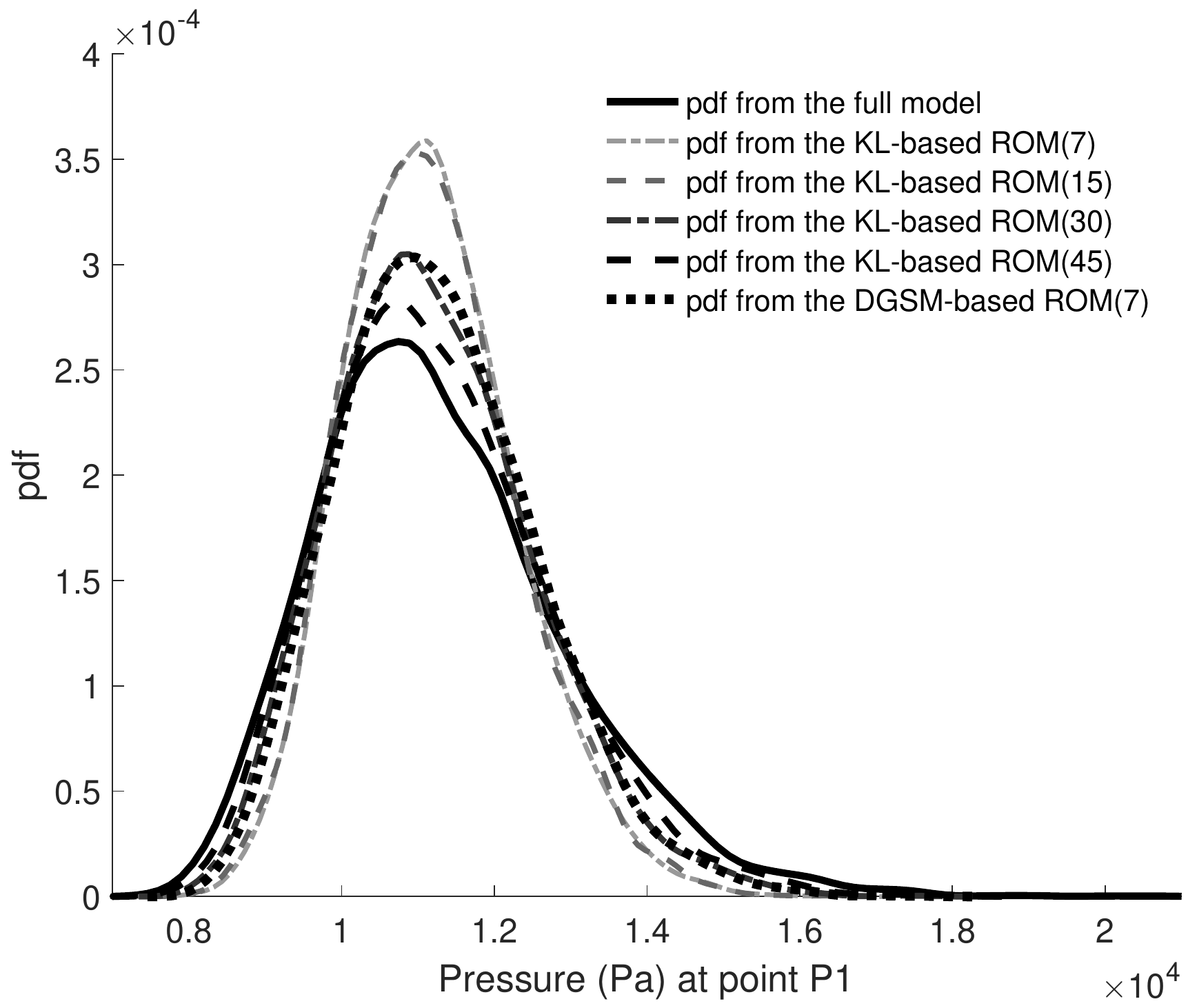}\\
\hspace{-1cm}
\includegraphics[width=0.53\textwidth]{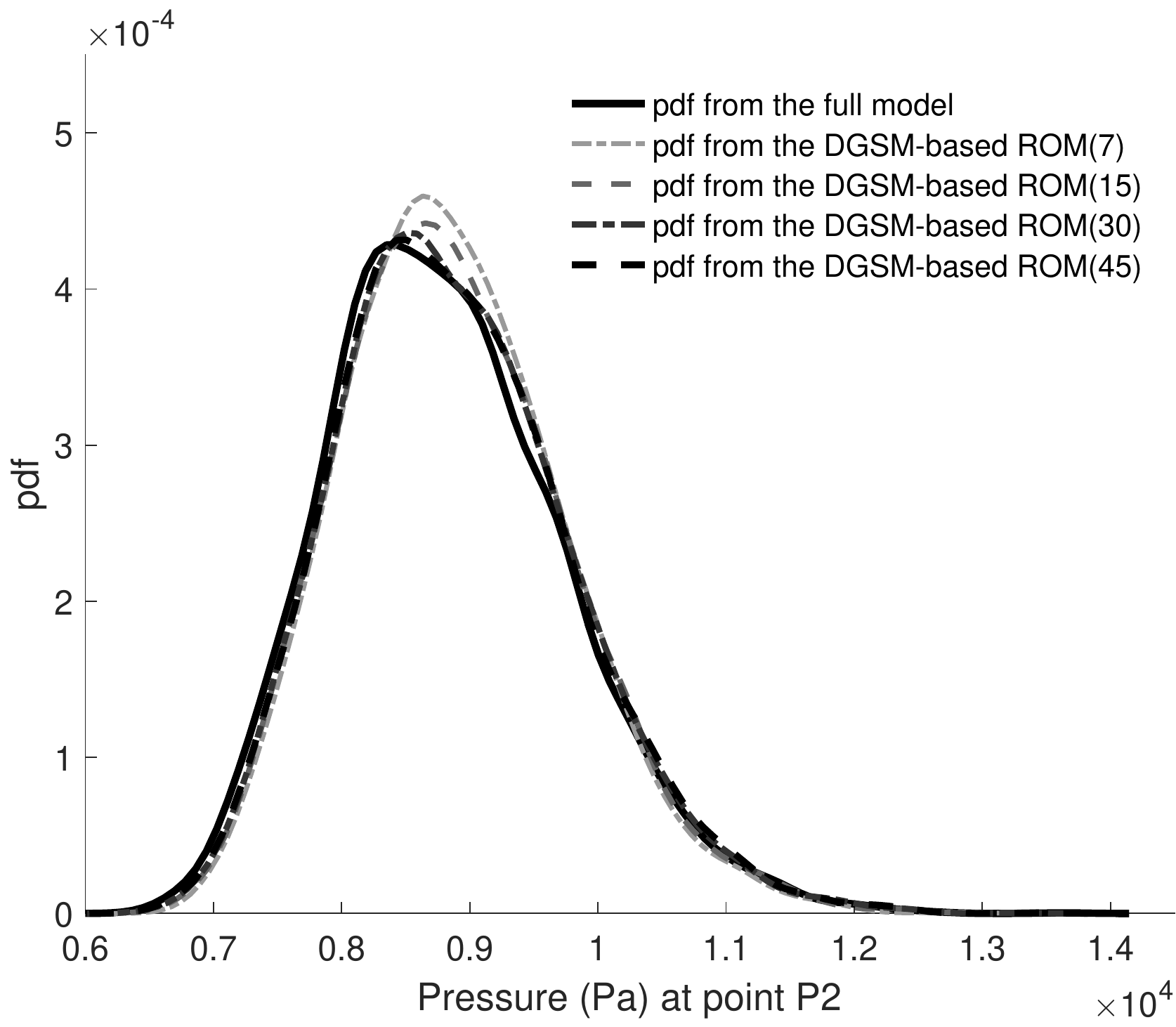}&
\hspace{-0.1cm}
\includegraphics[width=0.53\textwidth]{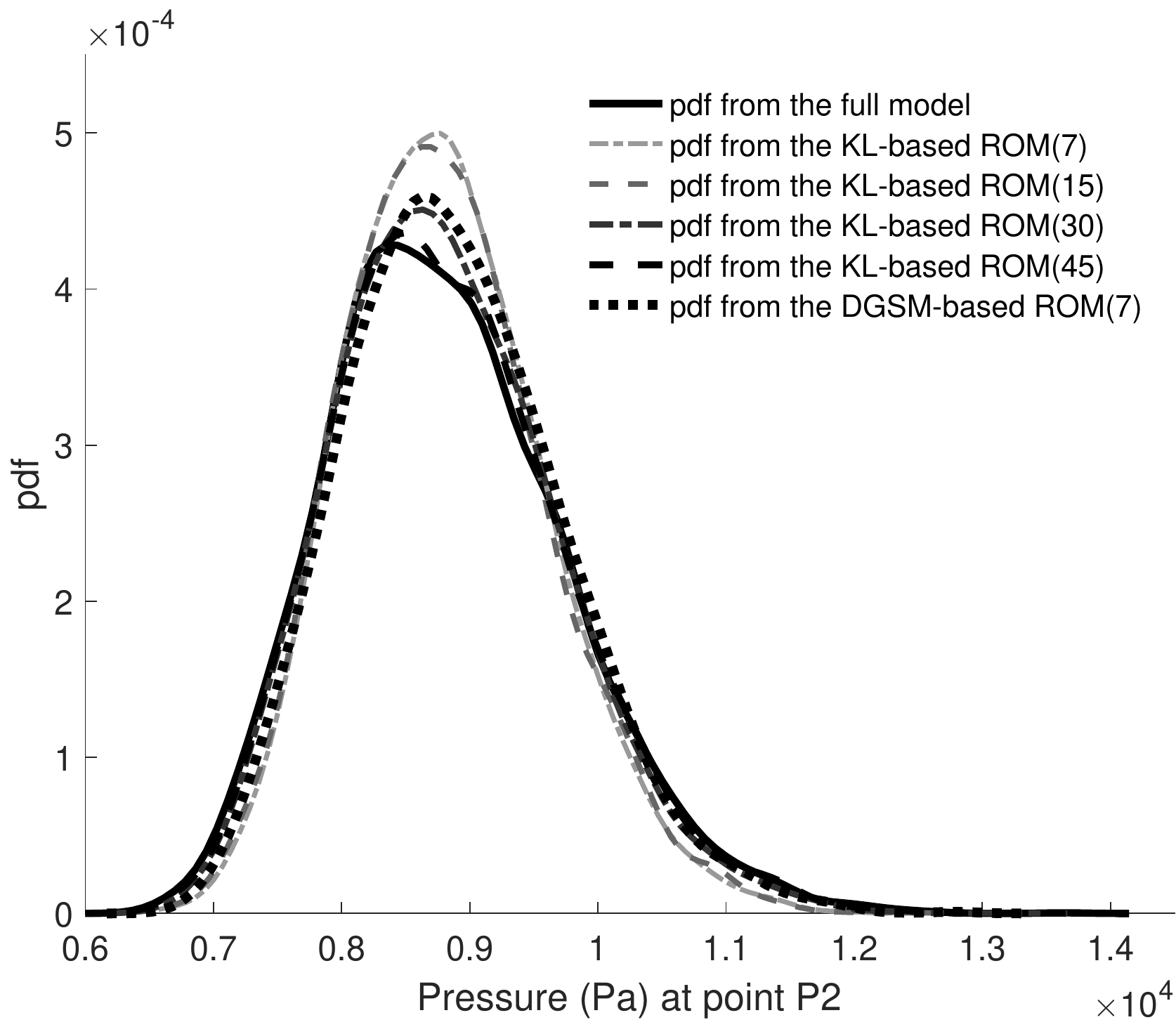}\\
\hspace{-1cm}
\includegraphics[width=0.53\textwidth]{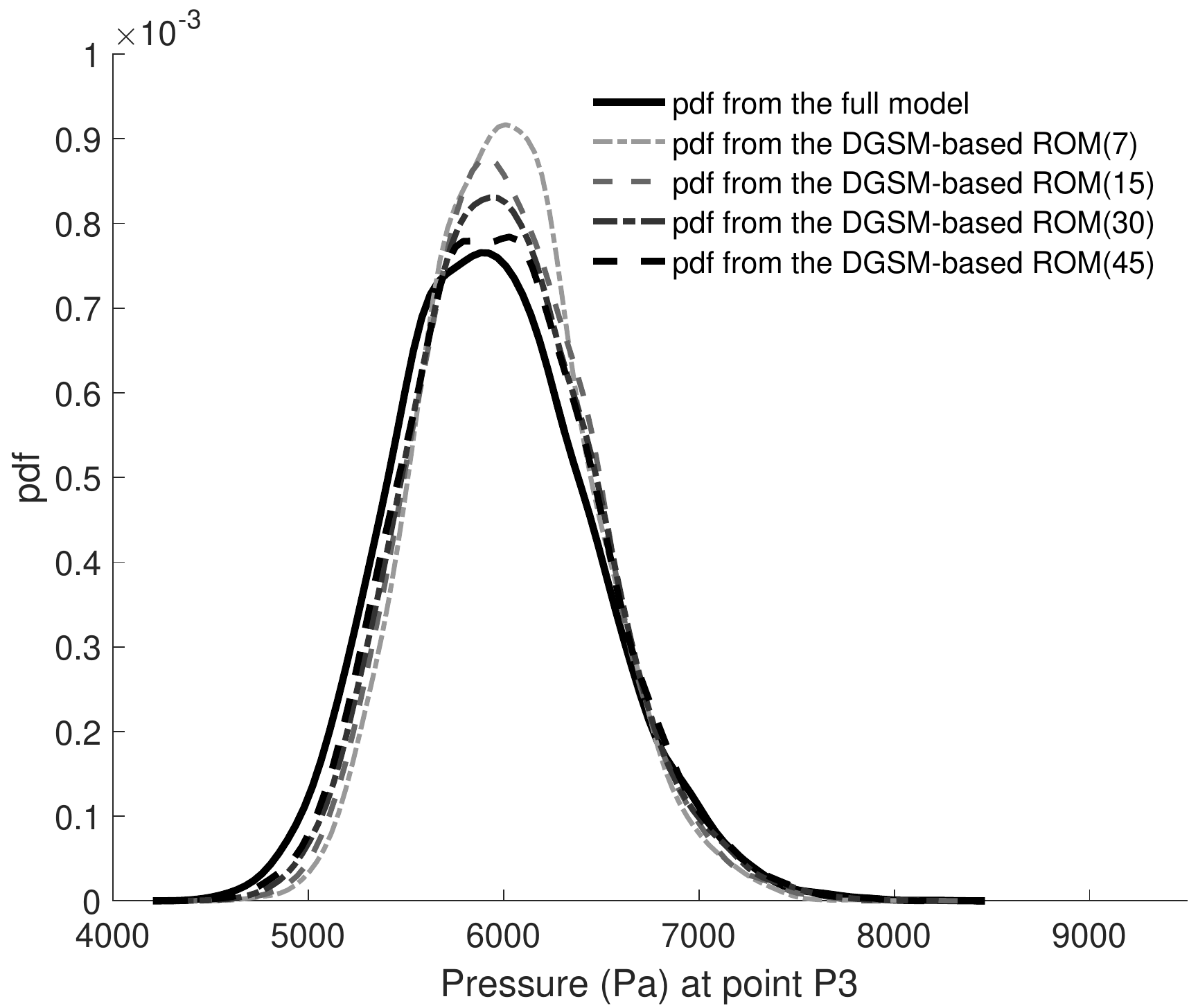}&
\hspace{-0.1cm}
\includegraphics[width=0.53\textwidth]{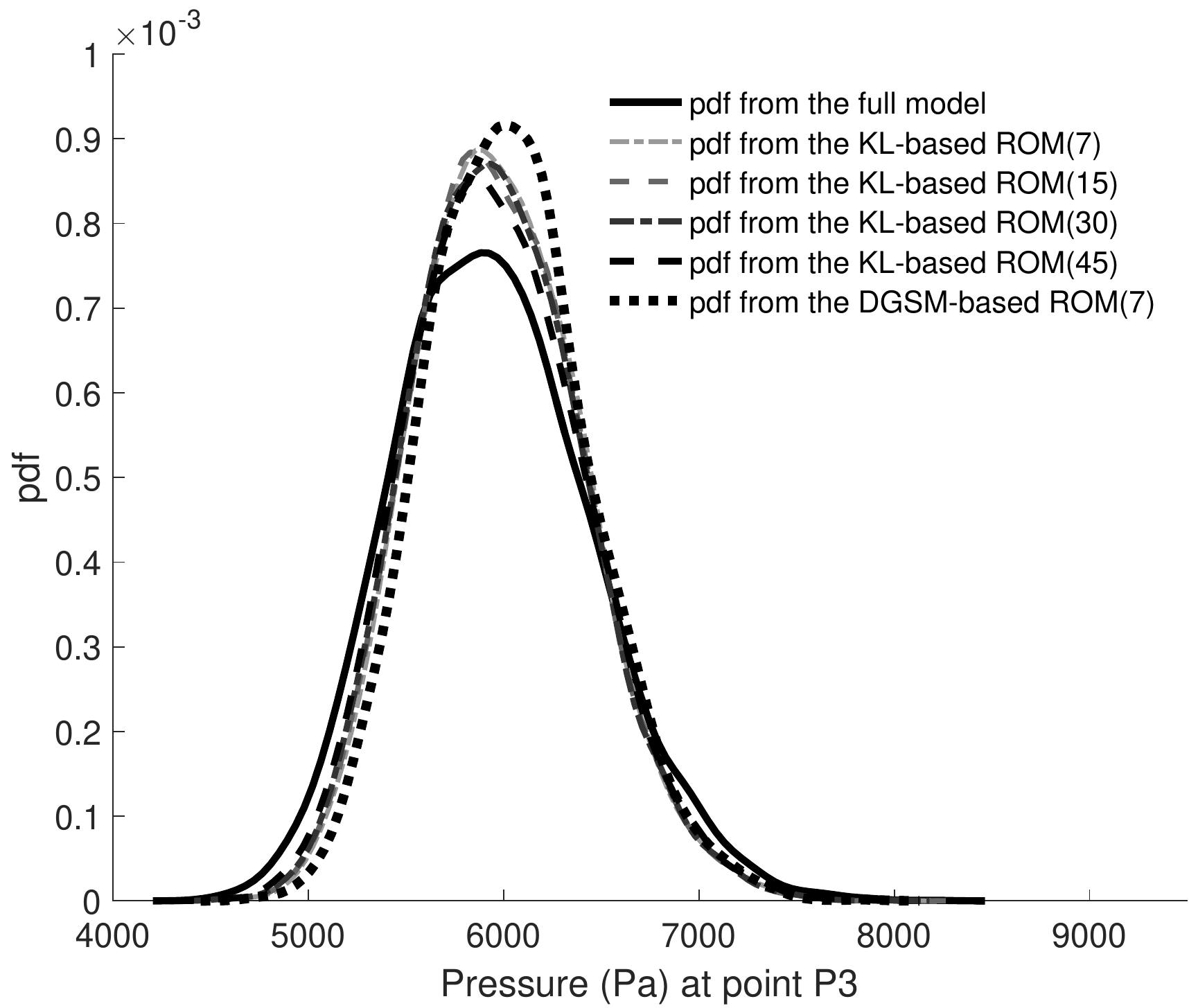}\\
\end{tabular}
\caption{Comparison of PDFs constructed from the DGSM-based ROM (left column) and KL-based ROM (right column) with variable fidelity at points $P_1$, $P_2$ and $P_3$.}
\label{fig:GSA_PDF_Bio}
\end{figure}

\section{Conclusions}\label{sec:conc}
We have presented a mathematical framework for GSA of models with functional
outputs, and have proposed an efficient computational method for identifying
unimportant inputs that is suitable for models with high-dimensional
parameters.  The latter is done by combining the proposed functional DGSMs,
``low-rank'' KL expansions of output QoIs, and adjoint-based gradient
computation.  In particular, the computational complexity of the proposed
approach, in terms of the number of required model evaluations, does not scale
with  dimension of the parameter. The effectiveness of the proposed framework
is illustrated numerically in applications from epidemiology, subsurface flow,
and biotransport. 

The proposed approach is effective in finding unimportant input parameters.
This approach also paves the way for an efficient surrogate modeling approach:
the low-rank KL expansion of the model output can be used to construct
efficient-to-evaluate surrogate models by computing surrogate models for the KL
modes, in the reduced parameter space, which is identified using the functional
DGSMs. The latter can be done using various methods including orthogonal
polynomial approximations~\cite{LeMaitreKnio10,Xiu10,Smith13}, 
multivariate adaptive regression splines~\cite{Friedman91}, 
or active subspace approaches~\cite{Constantine15}. 
We mention that active subspace methods have also been used directly for
dimension reduction in models with vectorial outputs.
Namely,~\cite{ZahmConstantinePrieurEtAl18} presents a gradient-based input
dimension reduction method for such models.  The method proposed
in~\cite{ZahmConstantinePrieurEtAl18} finds a set of important 
\emph{directions} in
the input parameter space by considering ridge approximations of the model
output and by minimizing an upper bound on the approximation error.  The
approach in~\cite{ZahmConstantinePrieurEtAl18} is related to the present work
when the goal of GSA is input dimension reduction.  

In future work, we seek to investigate generalizations to cases of
models with correlated inputs. While the proposed DGSMs can be computed for
such models in the same way, the corresponding variance-based indices need to
be generalized. We are also interested in applying the proposed method to 
more complex physical applications such as multiphase flow in geological 
formations.

\section*{Acknowledgments}
The research of A.~Alexanderian and R.C.~Smith was partially supported by the
National Science Foundation through the grant DMS-1745654. 
The research of R.C.~Smith was supported in part by the Air Force Office 
of Scientific Research (AFOSR) through the grant AFOSR FA9550-15-1-0299.
M.L.~Yu gratefully
acknowledge the faculty startup support from the department of mechanical
engineering at the University of Maryland, Baltimore County (UMBC). 

\setlength{\abovedisplayskip}{2pt}
\setlength{\belowdisplayskip}{2pt}

\appendix 
\section{Proof of Proposition~\ref{prop:L2}}\label{sec:proof_L2}

We will need the following key lemma, which is 
is based on the arguments in~\cite{Sobol07}.
\begin{lemma}\label{lem:ptwise_var}
For every $s \in \X$, 
$\int_\Thetaz \err{f}{\nominal} \mu(d\nominal)=
    2 D_{U^c}^\text{tot}(f; s)$.
\end{lemma}
\begin{proof}
Let $s \in \X$ be fixed. 
Consider the ANOVA decomposition of $f(s, \theta)$, as defined
in~\cref{equ:ANOVA}:
\[
f(s, \theta) = f_0(s) + 
f_1(s, \theta_U) + f_2(s,\theta_{U^c})+ f_{12}(s, \theta_U, \theta_{U^c}).
\]
By substituting this into the expression for $\err{f}{\nominal} $ and
simplifying we have
\begin{equation}\label{equ:error_exp}
\err{f}{\nominal} = 
\int_\Theta 
\big[
f_2(s, \theta_{{U^c}}) + 
   f_{12}(s, \theta_U, \theta_{U^c}) - 
   f_2(s, \nominal) - f_{12}(s, \theta_U, \nominal)
   \big]^2\mu\,(d\theta).
\end{equation}
Using the properties of ANOVA~\cite{Sobol07,Sobol:2001}, 
\[
   \int_{\Theta_{U^c}} f_2(s, \theta_{{U^c}}) \,\mu(d\theta_{U^c}) 
   = \int_{\Theta_{U^c}} f_{12}(s, \theta_U, \theta_{U^c})\,\mu(d\theta_{U^c}) 
   = \int_{\Theta_U} f_{12}(s, \theta_U, \theta_{U^c})\,\mu(d\theta_U) = 0,
\]
we can simplify~\cref{equ:error_exp} to get, for a fixed $\nominal \in \Thetaz$,
\[
\begin{aligned}
\err{f}{\nominal} &=
\int_\Theta \big[f_2^2(s, \theta_{{U^c}}) 
   + f_{12}^2(s, \theta_U, \theta_{U^c}) 
   + f_2^2(s, \nominal) + f_{12}^2(s, \theta_U, \nominal)\big]\,\mu(d\theta) \\
&= 
   D_{U^c}(f; s) + D_{U^c,U}(f;s) + f_2^2(s, \nominal) + \int_{\Theta_U} f_{12}^2(s,\theta_U,\nominal)\,\mu (d\theta_U).
\end{aligned}
\]
Integrating the above expression over $\Theta_{U^c}$ gives 
the desired result.
\end{proof}

\begin{named_proof}[Proof of Proposition~\ref{prop:L2}]
First note that the denominator is a constant and
\begin{equation}\label{equ:denom}
  \int_\X \int_\Theta f(s, \theta)^2\, \mu(d\theta) ds
= \int_\X \Bigg[ D(f; s) +
          \Big(\int_\Theta f(s, \theta)\mu(d\theta)\Big)^2\Bigg] \,ds
          \geq \int_\X  D(f; s) ds.
\end{equation}
Next, consider the expectation of the numerator in~\cref{equ:globalL2}:
\begin{multline}\label{equ:num}
\int_\Thetaz \int_\X \int_\Theta
(f(s, \theta)- \freduced(s, \theta_U))^2 \, \mu(d\theta) ds \mu(d\nominal)
=
\int_{\Theta_{U^c}} \int_\X \err{f}{\nominal}\, ds\mu(d\nominal)\\
=
\int_\X \int_{\Theta_{U^c}} \err{f}{\nominal}\, \mu(d\nominal)\,ds
= 2\int_\X D_{U^c}^\text{tot}(f; s),
\end{multline}
where changing the order of integration is justified by Tonell's theorem,
and the last equality follows from \cref{lem:ptwise_var}.
The desired result follows from~\cref{equ:num} and~\cref{equ:denom}.
\end{named_proof}

\section{Proof of Propositions~\ref{prp:BD} and~\ref{prp:BD_functional}}\label{sec:proof_BD}

We recall the following result:
if a random variable
$X$ satisfies $a \leq X \leq b$  and $\E{X} = m$,
then
\begin{equation}\label{equ:BD}
\var{X} \leq (b - m)(m - a) \leq  (b - a)^2/4.
\end{equation}
The first inequality is known as
the Bhatia--Davis inequality~\cite{BhatiaDavis00}.
The second inequality gives a corollary of the Bhatia-Davis inequality,
known as Popoviciu's inequality, that says
$\var{X} \leq (b - a)^2/4$, for
a random variable satisfying $a \leq X \leq b$.

\begin{named_proof}[Proof of Proposition~\ref{prp:BD}]
Note that clearly $a_j \leq \nu_j(g) \leq b_j$, for $j = 1, \ldots, \Np$.
Applying the inequality~\cref{equ:BD} with
$X = \big(\frac{\partial g}{\partial \theta_j}\big)^2$ and
$(b, m, a) = (b_j, \nu_j(g), a_j)$, $j = 1,\ldots, \Np$, we obtain
$\mathrm{Var}\big\{ \big(\frac{\partial g}{\partial \theta_j}\big)^2\big\} \leq
     \big(b_j - \nu_j(g)\big)\big(\nu_j(g) - a_j\big)
     \leq \frac{1}{4} (b_j - a_j)^2$.
Therefore,
for $j = 1, \ldots, \Np$,
\[
   \var{\nu_j^{(\Ns)}(g)} = \frac1\Ns \var{ \Big(\frac{\partial g}{\partial \theta_j}\Big)^2}
                        \leq \frac1\Ns \big(b_j - \nu_j(g)\big)\big(\nu_j(g) - a_j\big)
                        \leq \frac{1}{4\Ns} (b_j - a_j)^2.
\]
\end{named_proof}

\begin{named_proof}[Proof of Proposition~\ref{prp:BD_functional}]
First note that $\N_j^{(\Ns)}(f; \X)$ is indeed an estimator for $\N_j(f; \X)$. This is seen by
noting that, using Tonelli's theorem,
\[
\N_j(f; \X) = \int_\X \nu_j(f; s) \, ds = \int_\X \int_\Theta
\Big(\frac{\partial f}{\partial\theta_j}(s, \theta)\Big)^2 \mu(d\theta) ds
 = \int_\Theta \int_\X \Big(\frac{\partial f}{\partial\theta_j}(s, \theta)\Big)^2
ds \, \mu(d\theta).
\]
Then, applying Popoviciu's inequality to the random variable
$G_j(\theta) = \int_\X \Big(\frac{\partial f}{\partial\theta_j} (s, \theta)\Big)^2\, ds$,
which satisfies  $\|a_j\|_{L^1(\X)} \leq G_j \leq \| b_j \|_{L^1(\X)}$, gives:
\[
\var{ G_j } 
\leq 
\frac{1}{4} \left(\| b_j\|_{L^1(\X)} - \|a_j \|_{L^1(\X)}\right)^2
\leq 
\frac{1}{4} \| b_j - a_j\|_{L^1(\X)}^2, \quad j = 1, \ldots, \Np,
\]
where we also used the reverse triangle inequality. This completes the proof.
\end{named_proof}

\bibliographystyle{siamplain}
\bibliography{refs}
\end{document}